\documentclass[11pt]{article}
\usepackage{graphicx} 
\usepackage[utf8]{inputenc}
\usepackage{microtype}
\usepackage{amsmath,amsthm,amssymb, color, tikz, mathtools}
\usepackage[margin = 2.50cm]{geometry}
\usepackage{lineno}
\usepackage{authblk}

\usepackage{bbm}

\usetikzlibrary{positioning}
\usetikzlibrary{shapes}
\usetikzlibrary{decorations.pathreplacing}
\usepackage{enumerate, enumitem} 

\usepackage{thm-restate}
\usepackage[hidelinks]{hyperref}
\usepackage{cleveref}

\usepackage{graphicx}
\graphicspath{{./image/}}

\usepackage{amsmath,amssymb,amsthm}
\usepackage{algorithm}
\usepackage[noend]{algpseudocode}  

\usepackage{comment}

\hypersetup{colorlinks,linkcolor={blue},citecolor={blue},urlcolor={blue}}
\usepackage{thm-restate}

\renewcommand{\hat}{\widehat}

\newcommand{\R}{\mathbb{R}}

\newtheorem{theorem}{Theorem}[section]
\newtheorem{corollary}[theorem]{Corollary}
\newtheorem{remark}[theorem]{Remark}
\newtheorem{lemma}[theorem]{Lemma}
\newtheorem{claim}[theorem]{Claim}
\newtheorem{definition}[theorem]{Definition}

\newtheorem{observation}[theorem]{Observation}

\newcommand{\zone}{\{0, 1\}}
\newcommand{\ftwo}{\mathbb{F}_2}

\renewcommand{\Tilde}{\widetilde}

\newcommand{\supp}{\mathrm{supp}}

\newcommand{\G}{\mathcal{G}}
\newcommand{\wt}{wt}
\newcommand{\Or}{\mathcal{O}}
\newcommand{\poly}{\mathrm{poly}}

\newcommand{\LCM}{LCM}
\newcommand{\Codim}{Codim}

\newcommand{\Ann}{Ann}

\newcommand{\size}{size}
\newcommand{\lcmG}{\tau}
\newcommand{\ExpImb}{ExpImb}
\newcommand{\pmo}{\{-1,+1\}}
\newcommand{\Density}{\Lambda}
\newcommand{\Tower}{\mathrm{Tower}}

\title{
Constructive Arithmetic Regularity for Abelian Groups
}

\author[1]{Swarnalipa Datta}
\author[1]{Arijit Ghosh}
\author[2]{Chandrima Kayal\thanks{Funded by French PEPR integrated project EPiQ
(ANR-22-PETQ-0007) and ANR Grant FLITTLA (ANR-21-CE48-0023).}}
\author[1]{Manmatha Roy}

\affil[1]{Indian Statistical Institute, Kolkata, India}
\affil[2]{Universit\'e Paris Cit\'e, CNRS, IRIF, Paris, France}

\date{}

\begin{document}

\maketitle

\begin{abstract}
    The arithmetic regularity lemma of Green is a fundamental tool in additive combinatorics, providing a decomposition of arbitrary subsets of finite Abelian groups into pseudorandom pieces with respect to Fourier analysis. Over the years, it has found numerous applications in combinatorics, property testing, and theoretical computer science. However, Green's proof is existential, and many algorithmic applications require an efficient procedure for constructing the underlying regular decomposition. De, Nadimpalli, and Servedio (SODA 2022) developed a
query-efficient constructive regularity framework over
$\mathbb{F}_2^n$ and explicitly raised the problem of
extending their dimension-independent sumset-estimation
algorithm to other finite Abelian groups.

In this work, we resolve this question for finite Abelian groups of constant torsion. We develop the first constructive version of Green's {\em Arithmetic Regularity Lemma} in this general setting. Given oracle access to a subset of the group, our algorithm computes an implicit description of a regular decomposition together with an oracle representation of a structured subset satisfying the guarantees of Green's lemma. The algorithm runs in time and query complexity polynomial in the regularity parameters and the tower bound, while remaining independent of the size of the ambient group.

As applications, we obtain efficient algorithms for
removing linear configurations and approximately
counting linear images, including weighted sumsets
and suitable matrix systems of Cauchy--Schwarz
complexity one, over finite Abelian groups of
constant torsion. Our results provide a unified
algorithmic framework extending dimension-independent
constructive arithmetic regularity and its applications
beyond $\mathbb{F}_2^n$.

\end{abstract}
\newpage
\tableofcontents

\newpage

\section{Introduction}

Regularity lemmas are among the central tools of modern combinatorics and
theoretical computer science. They decompose large combinatorial objects into
a bounded number of pieces that behave almost randomly, thereby reducing
global questions to the analysis of locally structured components. In the
arithmetic setting, the fundamental result of this kind is Green's arithmetic
regularity lemma~\cite{green2005szemeredi}: an analogue of Szemer\'edi's graph
regularity lemma~\cite{szemeredi1975regular} which asserts that every subset
of a finite Abelian group can be partitioned into cosets of a suitable
subgroup so that, on almost all cosets, the set is pseudorandom in the Fourier
sense: after restricting the indicator function of the set to a typical
coset, all of its non-trivial Fourier coefficients are uniformly small. This
notion of pseudorandomness is captured by the following definition.

\begin{definition}[$\varepsilon$-regular]
A function $f: \G \to \mathbb{C}$ is \emph{$\varepsilon$-regular} if
$\sup_{\chi \neq \chi_0} | \widehat{f}(\chi)| \leq \varepsilon$. Let $H$ be a
subgroup of $\G$. For a coset $H_i = H + a_i$ with $a_i \in \G$, the
\emph{restriction} of $f$ to $H_i$ is the function
$f|_{H_i} : H \to \mathbb{C}$ given by $f|_{H_i}(x) = f(x+a_i)$ for all
$x \in H$. The restriction $f|_{H_i}$ is $\varepsilon$-regular if
$\sup_{\chi_0 \neq \chi \in \Hat{H}}| \widehat{f|_{H_i}}(\chi)| \leq \varepsilon$,
where $\Hat{H}$ is the dual of $H$.
\end{definition}

\begin{definition}[$(\eta, \varepsilon)$-regular decomposition]
Let $\G$ be an Abelian group and let $A \subseteq \G$. The coset decomposition
of $\G$ with respect to a subgroup $H \leq \G$ is an
\emph{$(\eta,\varepsilon)$-regular decomposition} of $\G$ with respect to $A$
if, for all but an $\eta$-fraction of the cosets $H_i$ of $H$, the restricted
function $\mathbbm{1}_A|_{H_i} : H \to \zone$ is $\varepsilon$-regular. Here
$\mathbbm{1}_A : \G \to \zone$ is the indicator function of $A$,
\begin{align*}
    \mathbbm{1}_A(x) =
    \begin{cases}
        1 & \text{if } x \in A, \\
        0 & \text{otherwise}.
    \end{cases}
\end{align*}
\end{definition}

Green proved that such a decomposition always exists, with respect to a
subgroup of bounded codimension:

\begin{theorem}[Green's arithmetic regularity lemma~\cite{green2005szemeredi}]
\label{th:Greens_regularity}
Let $A\subseteq \G$ and let $0<\eta,\varepsilon<1$. Then there exists a
subgroup $H\leq \G$ with co diemnsion at most a tower of twos of height \(\bigl(O(\eta^{-1}\varepsilon^{-2})\bigr)
\) such that the coset decomposition of $\G$ with respect to
$H$ is an $(\eta,\varepsilon)$-regular decomposition.
\end{theorem}

Green's lemma has found numerous applications in additive
combinatorics~\cite{green2010arithmetic} and property
testing~\cite{de2022approximating}.
However, Green's proof, like the classical proof of Szemer\'edi's lemma,
is purely existential, whereas many algorithmic applications in learning
theory, property testing, and approximate counting, require access to the
regular decomposition itself. This raises the central question of this paper:

\begin{center}
\emph{Can a regular decomposition be constructed efficiently, given only query
access to the set?}
\end{center}

De, Nadimpalli, and Servedio~\cite{de2022approximating}
answered this question affirmatively over $\ftwo^n$,
developing a query-efficient constructive regularity
framework and applying it to approximate sumset-size
estimation. In the same work \cite{de2022approximating}, they explicitly
raised the problem of extending their dimension-independent
sumset-estimation algorithm to other finite Abelian
groups, identifying the lack of suitable implicit
Fourier-analytic procedures as a central obstacle.

\emph{In this work, we resolve this open problem
for finite Abelian groups of constant exponent.}
We develop a constructive arithmetic regularity
algorithm whose query complexity depends only on
the regularity parameters and the group exponent,
and is independent of the size of the ambient group.
Our framework also yields an approximate sumset
estimator in this setting, together with more general
algorithmic applications.

In general, extending a framework from $\ftwo^n$ to arbitrary finite Abelian groups is far from routine. The vector-space setting of $\ftwo^n$ benefits from an underlying field
structure in which
characters can be identified with linear forms and, the concept \emph{independence} is simply the
ordinary linear independence. For a group of torsion $\tau$, by contrast, the
coefficient ring $\mathbb{Z}_\tau$ may contain zero divisors, characters may
interact across cyclic components of different orders, and ordinary linear
independence must be replaced by the more flexible notion of
pseudo-independence (Definition~\ref{defn_independence}) defined in \cite{chakraborty2024fourier, DBLP:conf/approx/DattaGKPR25}.

There is, however, a more fundamental obstacle to extending the regularity
algorithm from $\mathbb{F}_2^n$ to a general finite Abelian group $\G$. The
difficulty is not merely that individual Fourier-analytic tools become more
delicate. Green's existential proof of the arithmetic regularity lemma relies
crucially on the structure of \emph{Bohr neighbourhoods}. In the vector-space setting,
the relevant Bohr neighbourhoods coincide with subspaces and
therefore admit an exact coset structure on which Fourier analysis can be
performed directly. \emph{For a general Abelian group, positive-radius Bohr
neighbourhoods need not be subgroups and do not partition the ambient group
into cosets}. This prevents a direct algorithmic translation of Green's
argument. Green overcomes this difficulty in the existential setting by
working with nested and smoothed Bohr neighbourhoods; algorithmically,
however, this would require efficient procedures for manipulating these
additional objects and goes substantially beyond the direct use of
first-order Fourier-analytic primitives.

We overcome these obstructions by developing an implicit subgroup-based regularity framework for finite Abelian groups. However, this raises a natural question: where does the subgroup structure come from, given that \cite{green2005szemeredi} works with Bohr neighbourhoods? To address this, we work with subgroups arising as common kernels of characters, equivalently, \emph{zero-radius Bohr neighbourhoods}. Each such neighbourhood is a normal subgroup of $\G$, a structural property that provides the first and most important step toward developing an efficient regularity algorithm for finite Abelian groups. We provide a more detailed exposition in Subsection~\ref{subsec_proof_ideas}. 

Although these subgroups provide an exact coset partition, the main algorithmic challenge is to construct and refine this partition implicitly, without explicitly identifying the defining characters. Several fundamental Fourier-analytic tools for finite Abelian groups have been developed in \cite{chakraborty2024fourier,DBLP:conf/approx/DattaGKPR25}, which we use in our proofs. However, these tools alone are insufficient for our purposes. For example, our regularity algorithm constructs a
character decision tree in the spirit of \cite{de2022approximating}, but each
node now involves randomized implicit procedures and must operate over a much
more complicated character group.

A further ingredient is therefore a precise relation between the Fourier
coefficients of restricted and truncated functions
(Definitions~\ref{defn_restricted_function} and~\ref{defn_truncated_function}). In the
vector-space setting this relationship is essentially transparent, whereas
for general Abelian groups it is obscured by the presence of cyclic
components of different orders and by zero divisors in $\mathbb{Z}_\tau$.
Nevertheless, using the dual-group structure of $\widehat{\G}$, we establish
a clean correspondence between the restricted and truncated Fourier spectra
(Theorem~\ref{thm:main_result_1}). This correspondence is what allows the
algorithm to detect local non-regularity using global Fourier-analytic
subroutines, while keeping its query complexity independent of $|\G|$.

Conceptually, this viewpoint reveals that regularity in general finite Abelian
groups can still be organized through exact subgroup fibres, despite the loss
of the vector-space structure available over $\ftwo^n$. Thus, beyond its algorithmic role, our approach provides a concrete way to
understand how an arbitrary subset of a bounded-torsion Abelian group
decomposes into structured cosets on which its Fourier behaviour is
pseudorandom. In place of positive-radius Bohr neighbourhoods, the
subgroup--coset structure becomes the basic object underlying our constructive
decomposition.

Having established the constructive regularity theorem, we demonstrate its
versatility through two algorithmic applications. Although regularity
decompositions often reduce global problems to the analysis of structured
local pieces, extending the corresponding applications from \(\ftwo^n\) to
general finite Abelian groups will introduce additional obstructions due to the missing vector space structure and the zero-divisor problem. In particular, multiplication by nonzero coefficients
need not be invertible, while general matrix systems
require additional determinantal non-degeneracy
and torsion-compatible Cauchy--Schwarz complexity
conditions. Establishing the corresponding counting
and covering lemmas therefore requires additional
arguments beyond those used over $\ftwo^n$.
We use these ingredients to obtain efficient
algorithms for the following two problems.


\textbf{Removal of linear configurations.}
One of the most celebrated consequences of Szemer\'edi's graph regularity
lemma is the graph removal lemma of Ruzsa and
Szemer\'edi~\cite{ruzsa1978triple}: a graph containing few copies of a fixed
subgraph can be made entirely free of that subgraph by deleting a small
fraction of its edges. Removal lemmas have since become a cornerstone of
extremal graph theory, additive combinatorics, and property
testing~\cite{fox2017tight, gishboliner2017removal}.
Green's arithmetic regularity lemma yields an analogous arithmetic removal
lemma, in which graphs are replaced by subsets of finite Abelian groups and
forbidden subgraphs by solutions to homogeneous systems of linear equations;
removing triangles from a graph becomes a special instance of removing
solutions to a linear configuration from a set. We give an query-efficient, randomized algorithm for this task. For a single linear equation, and more generally for any system of
Cauchy--Schwarz complexity one (the class
studied structurally by Green and Tao~\cite{GreenTao2010}), our algorithm extracts a large
configuration-free subset by deleting only a small fraction of the original
set, with complexity independent of the ambient group size. We describe in more detail in
Section~\ref{Subsection: Our results}.

\textbf{Approximate image counting.}
Our second application concerns estimating the size of additive images of a
set. The prototypical problem is \emph{sumset-size estimation}: given query
access to $A \subseteq \G$, estimate the density of $A+A$.
Sumsets and their structure are a classical theme in additive
combinatorics~\cite{ruzsa2009sumsets, green2007freiman}, and the exact
computational problem is intractable: even recognizing sumsets is
NP-complete~\cite{abboud2025recognizing}.
De, Nadimpalli, and Servedio~\cite{de2022approximating} showed that, over
$\ftwo^n$, the constructive regularity lemma yields an efficient estimator for
the density of the sumset of a slight, density-preserving perturbation of $A$
(a relaxation that is necessary, as they also prove a matching lower bound for
the exact problem). The same framework also extends to \(A+B\), given query access to both
\(A\) and \(B\), and to \(kA\) for any fixed \(k\). Note that such approximate counting becomes significantly more challenging and inherently different when moving from $\ftwo^n$ to a general finite Abelian group $\G$ because of the group structure. This is also why \cite{de2022approximating} left this task as an open problem.

We generalize this in two directions at once. For an
Abelian group $\G$ of constant torsion, integers $a_1,\dots,a_m$, and the
linear form $L(x)=\sum_{j=1}^m a_jx_j$, consider the image set
\[
  L(A^m)\;=\;a_1A+\cdots+a_mA\;=\;\Bigl\{\,\textstyle\sum_{j=1}^m a_jx_j \;:\; x_j\in A\,\Bigr\}\subseteq \G,
\]
which recovers the sumset $A+A$ when $m=2$ and $a=(1,1)$. We show that, given only query access to $A$,
one can efficiently approximate the density of
$L(A'^m)$ for a density-preserving perturbation
$A'\subseteq A$. We further extend this result
to suitable matrix images $M(A'^m)$, under
determinantal non-degeneracy and torsion-compatible
complexity-one assumptions, while keeping the
query complexity independent of $|\G|$. We give a furture detailed Definition and exposition in Section~\ref{Subsection: Our results}.

\
Thus, beyond resolving the open problem of
De, Nadimpalli, and Servedio~\cite{de2022approximating} our framework substantially generalizes the
original sumset application to general linear images and to suitable systems
of linear equations over finite Abelian groups.
\


\section{Our Results}
\label{Subsection: Our results}

We first introduce some notation. $\G$ denotes a finite Abelian group $\mathbb{Z}_{p_1^{m_1}}\times \cdots \times \mathbb{Z}_{p_T^{m_T}}$, where $p_i$ are primes for all $i \in [T]$ and may not be distinct, and $\omega_\lcmG$ denotes a primitive $\lcmG^{th}$ root of unity, where $\lcmG = \LCM \{p_1^{m_1}, \ldots, p_T^{m_T}\}$. Throughout this article we assume $\G$ is an Abelian group with constant torsion/constant exponent that is $\tau$ is constant. $\Tower(\tau, t)$ denotes the tower of $\tau$ of height $t$. $\mathbbm{1}_{A}:\G \to \{0,1\}$ denotes the \emph{Indicator function} of the set $A \subseteq \G$.
\subsection{An Overview of Our Results}
\label{subsec:our result theorem}
Before diving into the statements let us describe some basic definitions that we will be using:

 Given arbitrary subsets $A, B \subseteq \G$ of some abelian group $\G$, we denote the following:
\begin{align}
     \Density(A) := \frac{|A|}{|\G|},
\end{align}
 and 
 \begin{align}
     \Density_B(A) := \frac{|A \cap B|}{|B|}.
 \end{align} where $|.|$ denotes the size of the set.
 Let $f:\G \to \zone$. A randomized algorithm $\mathcal{O}$ with query access to $f$, is said to be a \emph{probabilistic oracle machine} for $g: \G \to \zone$ if for any input $x \in \G$, the algorithm $\Or^f$ outputs a bit $\Or^f(x)$ such that 
    \begin{align*}
        \Pr[\Or^f(x) = g(x)] \geq \frac{2}{3}
    \end{align*}
    where the probability is over internal randomness of $\Or^f$. The query complexity of the machine is the number of oracle calls made by $\Or$ to $f$, and the running time of the machine is the number of time steps it takes in the worst case. 
 Let $f:\G \to \zone$ and $g:\G \to \zone$ be two functions. An algorithm $\mathcal{A}$ with query access to $f$ is said to output a $(\delta, q, T)$-oracle $\Or_g^f$ for the function $g$ if it outputs \emph{a representation of a probabilistic oracle machine} $\Or_g^f: \G \to \R$ such that the following holds:

    \begin{align}
        \Pr_{x \in \G}[\Or_g^f(x) \neq g(x)] \leq \delta 
    \end{align}
     where the query complexity of  $\Or_g^f$ is at most $q$ and the running time is at most $T$.   
   If $\delta = 0$ we call $\Or_g^f$ to be an exact oracle.

The Regularity Lemma is one of the central tools in extremal combinatorics and additive combinatorics, with numerous applications in property testing, learning theory, and approximation algorithms. Green proved an arithmetic regularity lemma for finite Abelian groups~\cite{green2005szemeredi}. However, many algorithmic applications, such as learning and property testing, require a constructive version.~\cite{de2022approximating} gave an efficient constructive regularity lemma over $\mathbb F_2^n$, and left open the problem of obtaining an efficient algorithm for general finite Abelian groups. Our first and main contribution resolves the open problem posed by
De, Nadimpalli, and Servedio~\cite{de2022approximating}: we obtain a
query-efficient constructive version of Green's arithmetic regularity lemma
for finite Abelian groups of constant torsion, with complexity independent
of the size of the ambient group. This structural result forms the basis for
the more general algorithmic applications developed below.

\begin{theorem}[A Constructive Regularity lemma for Abelian groups(Informal Version)]
\label{th:Informal_statement_regularity_lemma}
Let $A$ be an arbitrary subset of some abelian group $\G$. Given query access to $A$, a regularity parameter $\varepsilon>0$ and, a density threshold $\zeta > 0$, there exists an algorithm that makes $\poly(\Tower(\tau, O(\frac{1}{\varepsilon^3})), \frac{1}{\zeta})$ queries to $A$ and does a $\poly(\Tower(\tau, O(\frac{1}{\varepsilon^3})), \frac{1}{\zeta})$ time computation and returns with probability $\frac{9}{10}$ returns an \emph{implicit description} of a subset $A'$ such that it approximately captures the \emph{density} of $A$, where $\tau$ is the torsion of the Abelian group $\G$.
\end{theorem}




More precisely, with probability at least $9/10$, it outputs an probabilistic oracle $\mathcal{O}_{A'}$ and a character decision tree $T$ for a subgroup $H \le \G$ of co-dimension $\Tower(\tau, O(1/\varepsilon^{3}))$ such that one can use the oracle to construct a structural subset $A' \subseteq A$ satisfying the following:
\begin{enumerate}
  \item $\Density(A \setminus A') \le \varepsilon + \zeta$;
  \item Every coset $K$ of $H$ satisfies either $\Density_{K}(A') = 0$ or $\Density_{K}(A') \ge \frac{\zeta}{2}$;
  \item For every coset $K$ with $\Density_{K}(A') > 0$, the restricted indicator function $\mathbbm{1}_{A'}|_{K}$ is $\varepsilon$-regular.
\end{enumerate}
The number of cosets of $H$ is bounded by $M := |\G|/|H| \le \Tower(\lcmG, O(1/\varepsilon^{3}))$. The query complexity remain independent of the size of the group $|\G|$.

For a more formal statement we refer to section~\ref{section3:regularity algorithm}, Theorem~\ref{th: Regularity for Abelian}. We prove this theorem in Subsection~\ref{section:proof_of_efficient_regularity}.

The regularity structure provides a natural decomposition of the domain, a feature that lies at the heart of many learning and testing algorithms. Consequently, designing an efficient procedure for identifying such a regularity structure immediately leads to a variety of algorithmic applications. In this section, we highlight some of the well-studied  applications that follow from Theorem~\ref{th:Informal_statement_regularity_lemma}. Precisely, they are known for the usual Boolean domain, and it continues to hold for the Abelian domain as well. For each of these problems, we show how the constructive regularity lemma can be leveraged to obtain efficient algorithms with provable guaranties. 


One of the central themes in additive combinatorics is understanding how large a subset can be while avoiding a prescribed additive configuration. Classical examples include triangle-free sets, progression-free sets, and sum-free sets, all of which can be viewed as forbidding solutions to a suitable system of linear equations. More generally, given a linear equation, $L(\mathbf{x})=\sum_{j=1}^{m}a_jx_j=0, $ one may ask whether a set containing only a small number of such configurations can be converted into a completely configuration-free set by deleting only a few elements. Green's arithmetic removal lemma answers this question existentially. Our goal is algorithmic, using only query access to the set, can we efficiently construct an almost largest subset that is entirely free of the forbidden configuration, while making a number of queries independent of the ambient dimension?

Before moving to the statement we define the configuration density operator $\Lambda_L$ for functions $f_1, \dots, f_m: \G \to [0,1]$ and subsets $A \subseteq \G$ as:
\[
   \Lambda_{L}(f_1,\dots,f_m)
   =\mathbb{E}_{L(\mathbf{x})=0}\;\prod_{j=1}^{m}f_j(x_j),
   \qquad
   \Lambda_{L}(A)=\Lambda_{L}(\mathbbm{1}_A,\dots,\mathbbm{1}_A).
\]

\begin{theorem}[Single linear configuration removal: informal version]
\label{thm:remove-single_informal_statement}
Let $\G$ be a finite Abelian group with constant torsion $\tau$ and $L(\mathbf{x})=\sum_{j=1}^{m}a_jx_j=0$ with $m\ge3$ and $\gcd(a_j,\lcmG)=1$. For every target tolerance $\theta>0$, there exists a threshold $\delta=\delta(\theta,m,\tau)>0 \footnote{$\delta =\delta(\theta,m,\tau)$ is a function of $(\theta,m,\tau)$.}$, such that the following holds for :
Given query access to $A\subseteq\G$, there exists an algorithm that makes $\poly\Bigl(\Tower\bigl(\tau,\;O_m(\theta^{-3m})\bigr)\Bigr)\footnote{$O_m$ means the hidden constant is allowed to depend on the parameter m.}$
queries and does $\poly\Bigl(\Tower\bigl(\tau,\;O_m(\theta^{-3m})\bigr)\Bigr)$ time computation  and with probability at least $9/10$, it outputs a subset $A'\subseteq A$ satisfying $|A\setminus A'|\le\theta\,|\G|$ such that, \emph{whenever} $\Lambda_{L}(A)\le\delta$, the set $A'$ contains absolutely no solutions to $L(\mathbf{x})=0$.
\end{theorem}
For the triangle equation $x+y=z$, every solution is uniquely determined by choosing the two free variables $x,y \in \G$, after which $z=x+y$ is forced. Thus, the solution space is naturally identified with $\G^2$. In particular, every triangle in the parameter space corresponds to exactly one triangle in the original set, and vice versa. For general linear equations, we would like the same one-to-one correspondence between the relevant space and the solution set. This is precisely the role of our \emph{non-degeneracy assumption}: it guarantees that the parameterization remains bijective, so that the density and count of solutions are faithfully preserved under the reduction. We have a more formal exposition of these in Subsection~\ref{ssec:single}, where we will prove this argument formally via \emph{Counting Lemma}~\ref{lem:count-single}.

\begin{remark}
Our non-degeneracy hypothesis requires
\(\gcd(a_j,\lcmG)=1\) for every coefficient \(a_j\), ensuring that multiplication by each coefficient is an automorphism of $\G$. This guarantees a bijective parametrization of the solution set by the free variables. Moreover, the corresponding maps on the dual group $\widehat{\G}$ are also automorphisms, preventing degeneracies arising from zero divisors in $\mathbb{Z}_{\lcmG}$. From the Fourier-analytic perspective, it also ensures that the induced maps on the dual group $\widehat{\G}$ are automorphisms, preventing any collapse arising from zero divisors in $\mathbb{Z}_{\lcmG}$.
\end{remark}

Our first theorem considers a single forbidden linear equation. A natural next question is whether the same guarantees hold when several linear constraints are imposed simultaneously. Such systems capture a much richer class of additive configurations, including many classical patterns studied in additive combinatorics. We show that, for every fixed (complexity-$1$) system of linear equations,  $M\mathbf{x}=0$ where $M\in\mathbb{Z}^{r\times m}$ is an integer matrix, whose number of pseudo--independent (Definition~\ref{defn_independence}) rows is at most $r$, one can again efficiently extract an almost largest subset avoiding all forbidden configurations. Remarkably, the algorithm remains dimension-independent, with both the query complexity and running time depending only on the system and the desired accuracy. 

First, we define the generalised configuration density operator $\Lambda_M$ for functions $f_1, \dots, f_m: \G \to [0,1]$ and subsets $A \subseteq \G$ as:
\[
   \Lambda_{M}(f_1,\dots,f_m)
   =\mathbb{E}_{M\mathbf{x}=0}\;\prod_{j=1}^{m}f_j(x_j),
   \qquad
   \Lambda_{M}(A)=\Lambda_{M}(\mathbbm{1}_A,\dots,\mathbbm{1}_A).
\]

Before describing our result about system-removal we introduce two fundamental non-degeneracy conditions, which are essential for the well-definedness of the problem.

\begin{definition}[Determinantal non-degeneracy]
\label{def:det}
The $r$-th determinantal divisor $d_r(M)$ is the greatest common divisor of all $r\times r$ minors of $M$. We require the matrix to be co-prime to the group torsion: $\gcd\bigl(d_r(M),\lcmG\bigr)=1$.
\end{definition}

This condition functions as our matrix-level regularity filter. It ensures that the reduction of $M$ to $[\,I_r\mid 0\,]$ over $\mathbb{Z}_\lcmG$ is completely invertible, allowing the solution space $\{\mathbf{x}:M\mathbf{x}=0\}$ to be bijectively parameterized by $\mathbf{t}\in\G^{m-r}$ through linear forms $x_j=\ell_j(\mathbf{t})$. Under this mapping, the total solution count is exactly $|\G|^{m-r}$.

Next let us define formally the configuration that we want to remove:


\begin{definition}[Torsion-compatible Cauchy--Schwarz complexity-$1$]
\label{def:cs1_our_result}
Let $\ell_1,\dots,\ell_m:\G^{m-r}\to\G$ be the
coordinate forms parameterizing the solution space of
$M\mathbf{x}=0$.

We say that $M$ has Cauchy--Schwarz complexity-$1$
over $\mathbb{Z}_{\lcmG}$ if, for every $i\in[m]$,
there exists a partition
\(
[m]\setminus\{i\}=J_1\sqcup J_2
\)
such that, for each $s\in\{1,2\}$,
\(
\ell_i\left(
\bigcap_{j\in J_s}\ker\ell_j
\right)=\G.
\)
\end{definition}


Definition~\ref{def:cs1_our_result} is a
torsion-compatible version of the usual Cauchy--Schwarz complexity-$1$ condition (similar to the notion that was defined in Green-Tao (Definition~1.5~\cite{GreenTao2010}) specialized to $s=1$). Over groups of prime exponent, it is equivalent to the standard non-containment condition over the corresponding finite field. Over groups of composite exponent, the surjectivity condition is necessary for the Cauchy--Schwarz argument used below;
noncontainment in rational linear spans alone does not guarantee it. We give a formal exposition in
Subsection~\ref{ssec:systems}. We show the following:


\color{black}

\begin{theorem}[Matrix configuration removal for complexity-$1$ systems: informal version]
\label{thm:remove-system1_informal_statement}
Let $M\in\mathbb{Z}^{r\times m}$, whose number of pseudo--independent (Definition~\ref{defn_independence}) rows is at most $r$, complexity-$1$, and $\gcd(d_r(M),\lcmG)=1$. For every target tolerance $\theta>0$, there exists a threshold $\delta=\delta(\theta,m,r,t)>0$, such that for every finite Abelian group $\G$ with constant torsion $\tau$, there exists an algorithm that makes $\poly\Bigl(\Tower\bigl(t,\;O_m(\theta^{-6m})\bigr)\Bigr)$
queries and does $\poly\Bigl(\Tower\bigl(t,\;O_m(\theta^{-6m})\bigr)\Bigr)$ time computation, with probability at least $9/10$, it outputs an exact oracle for a subset $A'\subseteq A$ satisfying $|A\setminus A'|\le\theta\,|\G|$ such that, whenever $\Lambda_{M}(A)\le\delta$, the set $A'$ contains absolutely no solutions to $M\mathbf{x}=0$.
\end{theorem}

\begin{remark}
Note that our removal result is restricted to systems of Torsion Compatible Cauchy-Schwarz complexity-$1$. This is essentially the broadest class of linear configurations that can be treated using first-order Fourier analysis (or $U^2$
-uniformity) . In general, systems of higher Cauchy--Schwarz complexity require higher-order Fourier analysis and inverse theorems for Gowers uniformity norms. In other words, more general systems require control over higher-order uniformity norms $(U^k
)$ and consequently the machinery of higher-order Fourier analysis. So, the current result reflects the scope of first-order Fourier analysis itself.
\end{remark}

The preceding results demonstrate that arithmetic regularity enables efficient algorithms for removing forbidden additive configurations. We next show that the same structural framework is useful for a different class of additive problems. For the sake of simplicity, we first describe one simplest version of the problem that asks the following:

\begin{center}
\emph{ Given any subset $A \subseteq \G$, where $\G$ is an Abelian group with constant torsion. Given oracle access to $A$, precisely, given query access to the indicator function $\mathbbm{1}_{A}$ (See Definition~\ref{defn_indicator}) can we estimate the size of the sum-set $A+A$ efficiently? Precisely, can we compute the \emph{density} of the sumset $A'+A'$ with some additive error, where $A'$ is such that $\Density(A\setminus A')$ is small.}
\end{center}

Note that finding the exact size of the sumset is a hard, and it is not possible to obtain an efficient algorithm that is independent of $n$. Naturally one can ask if we approximately capture the density of $A$. The regularity decomposition compresses the input $A$ into a structured subset $A'$ on which subset sums can be simulated efficiently, leading once again to algorithms whose complexity is independent of the ambient dimension. 
Towards this we prove the following:


\begin{theorem}[Approximate sumset: informal version]
\label{th:simulating_sumset_intro}
Let \(\G\) be a finite Abelian group of constant torsion \(\tau\), and let \(A\subseteq \G\). For every target accuracy parameter \(\theta\in(0,1)\), given query access to \(A\), there exists a randomized algorithm with query complexity and running time
\(\poly\bigl(\Tower(\tau,O(\theta^{-O(1)}))\bigr)\),
both independent of \(n\) and \(|\G|\) such that with probability at least \(9/10\), the algorithm outputs a value \(v\) and a \((\theta,0,O(n))\)-oracle for \(A'+A'\), where \(A'\subseteq A\) is a density-capturing perturbation satisfying
\(\Lambda(A\setminus A')\leq \theta\),
such that
\(\bigl|v-\Lambda(A'+A')\bigr|\leq \theta\).
\end{theorem}





\begin{remark}
Theorem~\ref{th:simulating_sumset_intro}
is a special case of our more general approximate
image-counting result for a single linear form
(Theorem~\ref{thm:lin-image}), obtained by taking
$m=2$ and $a_1=a_2=1$.

The problem of approximate sumset estimation was
studied by De et al.~\cite{de2022approximating},
who gave a query-efficient algorithm for subsets
of $\ftwo^n$ and posed its extension to other
finite Abelian groups as an open problem.
We resolve this problem for finite Abelian groups
of constant exponents and, more generally, obtain
an approximate image-counting algorithm. We therefore omit a separate proof of
Theorem~\ref{th:simulating_sumset_intro}.
Instead, we state the more general
Theorem~\ref{thm:lin-image} below and then
briefly explain how the approximate sumset
result follows as a corollary.
\end{remark}

\color{black}

\begin{definition}[Linear image set]
For integers $a_1,\dots,a_m\in\mathbb Z$ and the form $L(\mathbf x)=\sum_{j=1}^m a_jx_j$,
the \emph{image set} of $A\subseteq \G$ under $L$ is
\[
  L(A^m)\;=\;a_1A+\cdots+a_mA\;=\;\Bigl\{\,\textstyle\sum_{j=1}^m a_jx_j \;:\; x_j\in A\,\Bigr\}\subseteq \G.
\]
For $m=2$, $a=(1,1)$ this is the sumset $A+A$.
\end{definition}

Note the genuine image $L(A^m)$ cannot be estimated to any
nontrivial accuracy with $|\G|$-independent query complexity (see \cite{de2022approximating}), so we estimate the density of the
image of a \emph{density-capturing} perturbation $A'\subseteq A$, exactly as in
Theorem~\ref{th:Informal_statement_regularity_lemma} and that is sufficient to simulate the density of the image set with some additive error. 

\textbf{Non-degeneracy hypothesis.}
We assume every coefficient is a unit modulo the torsion,
$\gcd(a_j,\tau)=1$ for all $j$. This is the single-form instance of
pseudo-independence (Definition~A.13): it makes each map
$\chi\mapsto\chi^{a_j}$ an automorphism of $\widehat\G$, which is what the
Parseval steps below require in the presence of zero divisors. No
complexity-$1$ assumption is needed: at $r=1$ the second moment closes for all
$m\ge2$.

\begin{theorem}[Approximate image counting for a linear system]
\label{thm:matrix-image}
Fix $M\in\mathbb Z^{r\times m}$, 
$m>r$,
$\gcd(d_r(M),\lcmG)=1$, and complexity-$1$ over $\mathbb Z_\lcmG$. For every
$\theta>0$ and every finite Abelian $\G$ with constant torsion $\lcmG$, given query
access to $A\subseteq\G$, Algorithm~\ref{alg:matrix-image} makes
\[
   \poly\Bigl(\Tower\bigl(\lcmG,\,O_m(\theta^{-O(m)})\bigr)\Bigr)
\]
queries and time, both independent of $n$ and $|\G|$, and with probability
at least $9/10$ outputs a value $v$ and an oracle $\mathcal O_{M(A'^m)}$ such that,
for a density-capturing perturbation $A'\subseteq A$ with
$\Density(A\setminus A')\le\theta$,
\[
   \bigl|\,v-\Density\bigl(M(A'^m)\bigr)\,\bigr|\le\theta,
\]
and $\mathcal O_{M(A'^m)}$ is an oracle for the image set $M(A'^m)\subseteq \G^{r}$.
\end{theorem}


\subsection{Overview of the Proof}\label{subsec_proof_ideas}

We will start by describing the existing efficient regularity algorithm by \cite{de2022approximating} and why it can not be extended to Abelian groups. The goal of the algorithm is to construct a bounded-codimension subspace \(H\leq \mathbb F_2^n\) whose cosets partition \(\mathbb F_2^n\) so that \(A\) becomes regular on almost all cosets. Equivalently, it seeks a partition into structured pieces on which all nontrivial Fourier coefficients of the corresponding restrictions of \(A\) are small.\\

\textbf{The existing Algorithm for $\ftwo^n$ by \cite{de2022approximating}:}

 The algorithm constructs the desired partition iteratively by growing a parity decision tree. Initially, the tree consists of a single leaf, corresponding to the whole space, $\ftwo^n$. At any subsequent stage, the parity labels appearing at the internal nodes determine a subspace \(H\), and the leaves of the tree correspond precisely to the cosets of \(H\).

\begin{itemize}
    \item For each current leaf, corresponding to a coset \(H_i\), the algorithm  search for a nontrivial Fourier coefficient witnessing that \(A_{H_i}\) is not \(\varepsilon\)-regular. If no such witness is found for at least a \(1-\varepsilon\)-fraction of the leaves, the refinement process terminates. Otherwise, witnessing parities are extracted from the irregular leaves. Whenever such a parity is linearly independent of the parities already used in the tree, every current leaf is split according to that parity. Thus, adjoining a new parity refines the current collection of leaves and, equivalently, produces the cosets of a smaller subspace.
    \item Termination is established through an energy-increment argument. Green~\cite{green2005szemeredi} defined the `progress measure' $\ExpImb$ (see Equation~\ref{progress_measure}) and proved that at each step, it increases by at least $\varepsilon^3$. Since, the progress measure $\ExpImb$ is at most $1$, therefore, the algorithm terminates after at most $O(\frac{1}{\varepsilon^3})$ many iterations, such that all but $\varepsilon$-fraction of the surviving cosets are $\varepsilon$ regular.

    \item Note that each node of the algorithm needs to find one heavy Fourier coefficient. One can apply famous \emph{Goldreich Levin} to find such heavy Fourier co-efficients but it is itself costly, takes exponential many queries, and we need to do it at each step, even for each node in the worst-case scenario. To make this step efficient what \cite{de2022approximating} is that, they apply another famous trick called \emph{List correction} framework, where instead of finding a heavy Fourier co-efficient explicitly, we find them implicitly and work with the linear dependency with each other. \cite{de2022approximating} call this procedure \emph{Implicit Goldreich Levin} algorithm. Note that each node of the algorithm now contains other random procedures, and they accumulate some error that needs to be handled carefully till the end. 
\end{itemize}
 
 \textbf{Proof idea of Theorem~\ref{th: Regularity for Abelian}.}

While our algorithm for general Abelian groups follows the basic skeleton of the framework in \cite{de2022approximating}, it should first be noted that it is not a direct generalization. Indeed, this issue is also mentioned as future work in \cite{de2022approximating}, where the authors point out some natural obstacles to extending the Goldreich-Levin framework beyond their setting. Here, we briefly describe our algorithm and highlight the new elements required in the general Abelian group setting.

Algorithm~\ref{algo: efficient regularity lemma} takes as black-box access a subset $A$ of $\G$, along with the regularity parameters $\varepsilon$ and $\zeta$. It constructs a character decision tree~\ref{defn_char_tree} that recursively partitions the group $\G$ into cosets, until all but $\varepsilon$-fraction of the resulting cosets~\ref{defn_cosets_groups} become $\varepsilon$ regular. Now, finding one heavy Fourier coefficient is itself costly, and we need to do it at each step, even for each node in the worst-case scenario.

\begin{itemize}
    \item \textbf{Not every Bohr neighbourhood is a subgroup, hence not possible to algorithmize Green's proof.} Let $\G$ be a finite Abelian group and let $\Gamma=\{\chi_{\gamma_1},\ldots,\chi_{\gamma_d}\}\subseteq \widehat{G}$ be a set of characters. A Bohr neighbourhood is defined in the following way:
$$B_{\Gamma,\delta}
=
\left\{
x\in G:
|\arg \chi_{\gamma}(x)|\leq 2\pi\delta
\text{ for every } \chi_\gamma\in\Gamma
\right\}.$$
Here $d$ is called the rank of $B_{\Gamma, \delta}$, and $\delta$ is the radius. These $B_{\Gamma,\delta}$ has the following property:
$$B_{\Gamma,\delta}
+
B_{\Gamma,\delta}
\subseteq
B_{\Gamma,2\delta},$$
which is because, if $x,y \in B_{\Gamma,\delta}$, then, since for every $\chi_\gamma \in \Gamma$, $\chi_\gamma(x+y)=\chi_\gamma(x)\chi_\gamma(y)$, and so $\arg \chi_\gamma(x+y) = \arg \chi_\gamma(x) + \arg \chi_\gamma(y)$, by triangle inequality we have $|\arg \chi_\gamma(x+y)| \leq |\arg \chi_\gamma(x)| + |\arg \chi_\gamma(y)| \leq \delta + \delta =2\delta$, which implies $B_{\Gamma,\delta}
+
B_{\Gamma,\delta}
\subseteq
B_{\Gamma,2\delta}$. 

Also, for a subgroup $H$ of $\G$, we have $H+H=H$. So, Bohr neighborhoods are not subgroups in general. Moreover, Green points out that a rank-$d$ Bohr neighborhood behaves geometrically more like a
$d$-dimensional convex body than like a subgroup. In particular, one
typically expects $|B_{\Gamma,\delta}+B_{\Gamma,\delta}|$ to be closer to $2^d|B_{\Gamma,\delta}|$ than to $|B_{\Gamma,\delta}|.$ Hence, a Bohr neighborhood is not invariant under addition in the way that a subgroup is, and one cannot perform harmonic analysis on it, not even approximately. Consequently, Bohr neighborhoods do not partition $\G$ into cosets in general. This prevents us from directly extending the Goldreich-Levin or Implicit Sieve type algorithmic framework uniformly to arbitrary
positive-radius Bohr neighbourhoods. Green overcomes the structural difficulty by working with
nested and smoothed Bohr neighbourhoods, but algorithmizing this framework would require corresponding
algorithms for these objects. Thus our algorithmization cannot simply use all
$B_{\Gamma,\delta}$ as replacements for subgroups and cosets.

\item \textbf{The task can be done without using the Generalised Bohr neighbourhood.} First let us make an observation that, for $\delta =0$, $$B_{\Gamma,0}
=
\left\{
x\in G:
|\arg\chi_\gamma(x)|=0
\text{ for every }\chi_\gamma\in\Gamma
\right\}.$$
Since $\arg\gamma(x)=0$ if and only if $\chi_\gamma(x)=1$, it follows that
$$B_{\Gamma,0}
=
\{x\in G:\chi_\gamma(x)=1
\text{ for every }\chi_\gamma\in\Gamma\} = V_{0,\gamma_1,\ldots,\gamma_d},$$
where $V_{0,\gamma_1,\ldots,\gamma_d}$ is defined as in Definition~\ref{defn_cosets_groups}.
Thus Green's zero-radius Bohr neighbourhood is precisely the normal
subgroup $V_{0,\gamma_1,\ldots,\gamma_d}$ appearing in our framework. In particular, at radius zero, the approximate Bohr neighbourhood becomes a subgroup. Our first contribution towards the algorithm is that bypassing generalised Bohr neighbourhood argument and work with generale subgroup structure so that we can use Fourier analytic tools. Now every process has it's cost, our bypassing the Bohr neighbourhood techniques comes with the coset that now the co-dimension of the resulting subgroup depends on $\tau$, the exponent of the domain $\G$ while in Green's original proof the exponent depends on only tower of $2$ and some other constant.

\item \textbf{List Correction framework at each node.} While trying to algorithmize this procedure, the main idea here is to save queries by not identifying the Fourier coefficient explicitly, and still get a good estimate of it. This phenomenon is known and well studied in the literature, namely as \emph{List Correction}. \cite{de2022approximating} used a similar technique to give an efficient regularity lemma for $\ftwo^n$. Eventually, for an Abelian group, one needs more structure, but thanks to the work by~\cite{DBLP:conf/approx/DattaGKPR25}, the ground of the work was already established, but not enough. But here comes more subtlety with the fact that we need to run this random procedure \emph{list correction} at each node, and we do not have the access of the charcaters, while we need to run the process each time on a new restricted subfunctions. Hence, alone \cite{DBLP:conf/approx/DattaGKPR25} framework is not sufficient, we still need more structural understanding of the Abelian group to be able to extend the List correction framework and  apply it recursively at each node for different restricted functions without knowing the character explicitly that was used for the restrictions. 
    
One very subtle but important tool of our proof is that the algorithm works with the {\em truncated} function, with respect to $A$, $f\|_A : \G \to \{0,1\}$ is 
    \[
    f\|_A(x) =
    \begin{cases}
    f(x) & x \in A, \\
    0 & x \notin A.
    \end{cases}
    \](Definition~\ref{defn_truncated_function}) on the cosets. On the other hand, the main algorithm originally has access to the restricted function and the aim to get a regularity condition also with respect to the restricted function $f|_A : A \to \{0,1\}$, with respect to $A$, is defined as  
    \[
        f\mid_A(x) := f(x)
    \]
    for all $x$ in $A$. If we were working with the vector space, establishing a relation between the Fourier coefficients is quite straightforward and simple, and was done by~\cite{de2022approximating} also. In the case of an Abelian group, this becomes quite different. 

    This issue is resolved by Theorem~\ref{thm:main_result_1} (see Section~\ref{section_relation_res_trun}), which gives a relation between the coefficients of a truncated function with the corresponding restricted function (Definition~\ref{defn_restricted_function}). While the proof is quite simple but it gives a deep insight about the structure of the coefficients with respect to Abelian groups.

    At each stage, the algorithm runs the List Correction (Algorithm~\ref{algo:Implicit_Sieve_generalized}) on the current restricted functions on cosets in order to identify the heavy Fourier coefficients. Once a heavy coefficient is found, the algorithm partitions the group further by splitting all the cosets at this stage with respect to the corresponding heavy character that it chose. Thus, one new heavy character is added at each stage, unless a good fraction of the surviving cosets are regular. Green~\cite{green2005szemeredi} defined the `progress measure' $\ExpImb$ (see Equation~\ref{progress_measure}) and proved that at each step, it increases by at least $\varepsilon^3$. Since, the progress measure $\ExpImb$ is at most $1$, therefore, the algorithm terminates after at most $O(\frac{1}{\varepsilon^3})$ many iterations, such that all but $\varepsilon$-fraction of the surviving cosets are $\varepsilon$ regular.

    We use this subroutine of List-correction which spends constant query to detect the heavy, that is, non-regular Fourier coefficient. We give a formal proof of these steps where we will use the subroutine in the main algorithm in Appendix~\ref{section_implicit_sieve} and Appendix~\ref{section_estimation}.
    \item \textbf{Query Complexity.}After obtaining the $\varepsilon$ regular cosets, the algorithm estimates the Density (see Definition~\ref{defn_density}) on each surviving cosets. Cosets whose density is at least $\zeta$ are not discarded. The union of the remaining surviving cosets defines the set $A' \subseteq A$, which satisfies the properties in Lemma~\ref{Prop: prop of Green's regularity}. So, we have obtained a good approximation $A'$ of the set $A$. Note that the character decision tree constructed by the algorithm defines an exact oracle for $A'$.

    The query complexity of Algorithm~\ref{algo: efficient regularity lemma} is $\poly(\Tower(\tau, \frac{1}{\varepsilon^3}), \zeta)$, which follows from that of the List-Correction Algorithm~\ref{algo:Implicit_Sieve_generalized}, the Coefficient Estimation Algorithm~\ref{algo: estimation}, and the fact that the Algorithm terminates after $O(\frac{1}{\varepsilon^3})$ many iterations. We present the main Algorithm~\ref{algo: efficient regularity lemma} and the proof in Section~\ref{section3:regularity algorithm}.

\end{itemize}


\begin{description}

\item[Proof idea of `Algorithmic removal of linear configurations' (Theorem~\ref{thm:remove-single} and Theorem~\ref{thm:remove-system}).] 
The goal is to establish a constructive framework for the algorithmic removal of linear configurations within finite Abelian groups of constant torsion. The proof of this follows in two main steps.

In the first step, we apply the algorithmic version of the Green's regularity lemma (see Theorem~\ref{th: Regularity for Abelian} and Algorithm~\ref{algo: efficient regularity lemma}), which partitions the group into cosets such that the set $A$ is $\varepsilon$ regular on all but $\varepsilon$-fraction of the cosets, as stated in the proof idea of Theorem~\ref{th: Regularity for Abelian}. As described earlier, it then discards the cosets that are not regular and keeps only those cosets whose estimated $\Density$ is $\geq \zeta$. The approximation $A'$ of $A$ preserves the behavior of $A$ across the cosets. Here $A'$ is the set that satisfies the properties in Lemma~\ref{Prop: prop of Green's regularity}.

We then consider the non-degeneracy hypothesis, that is, every coefficient of the equation $$L(\mathbf{x})=\sum_{j=1}^{m}a_{j}x_{j}=0$$ is a unit modulo the torsion of the group, $\gcd(a_j,\lcmG) =1$, for all $j \in [m]$. This ensures that the map that takes a character $\chi$ to $\chi^{a_j}$ is an automorphism over $\widehat{\G}$. This generalizes to Determinantal non-degeneracy for a system of equations $(M\mathbf{x}=0)$ (see Definition~\ref{def:det}), where $M$ is a matrix with Cauchy--Schwarz complexity -$1$ (see Definition~\ref{def:cs1}) and has number of pseudo-independent (Definition~\ref{defn_independence}) rows at most $r$.
This condition should allow us to apply the Counting lemma for a single equation (Lemma~\ref{lem:count-single}), and the generalized Von-Neumann for system of linear equations (Lemma~\ref{lem:count-system}), which shows that the configuration density operator $\Density_L$ for function $f_1, \dots, f_m: \G \to [0,1]$ is within $\pm \varepsilon$ error of the product of their individual averages, inside $A'$. But, here we should slow down and note that counting doesn't directly hold for Abelian groups just by given non-degeneracy condition, one has to define the refined notion of \emph{Torsion compatible Cauchy--Schwarz Complexity -$1$ } to be able to prove such counting Lemma.

Once we establish the \emph{Counting Lemma}, this implies that if there exists even one solution to the system in $A'$, it yields many solutions to the system inside the original set $A$. This further implies that if the density of the solutions inside the $\varepsilon$-regular cosets is sufficiently small, then there cannot exist any solution to the system in the approximation $A'$ of $A$. The proof follows by contradiction, where we assume that there exists a solution to the system in $A'$, and then show that this leads to many solutions in $A'$ ($\geq c |\G|^{m-r}$ to be precise), which violates our initial assumption that the density of the system of linear configurations in $A$ is sufficiently small.

The key ingredient is the assumption that the system has Cauchy-Schwarz Complexity-$1$. By repeatedly applying the Cauchy--Schwarz inequality and duplicating variables according to the partition $J_1 \sqcup J_2$ (see the definition of Cauchy--Schwarz Complexity -$1$, Definition~\ref{def:cs1}), the linear forms become separated from one another. This separation allows the error terms to be controlled through their Gowers $U^2$-norms. Note that Complexity-$1$ system could be controlled by $U^2$-norms and which is the reason we can get the correct count of configuration, which is not possible for any higher order complexity system. See Section~\ref{sec:linear-removal} for more details.

\item[Proof idea of `Approximate image counting' (Theorem~\ref{thm:lin-image} and Theorem~\ref{thm:matrix-image}).]
The goal is to prove that given a subset $A \subseteq \G$, one can efficiently estimate the density of its image under both a single general linear form ($L(A^m) = a_1 A + \dots + a_m A$) and a matrix system ($M(A^m) \subseteq \G^r$), whose number of pseudo--independent (Definition~\ref{defn_independence}) rows is at most $r$, and has Cauchy--Schwarz complexity -$1$ (see Definition~\ref{def:cs1}), with a query complexity that is entirely independent of the group's size $|\G|$. The proof of this also follows in two main steps.

The first step is again the algorithmic version of the Green's regularity lemma (Theorem~\ref{th: Regularity for Abelian} and Algorithm~\ref{algo: efficient regularity lemma}), which partitions the group into cosets such that the set $A$ is $\varepsilon$ regular on all but $\varepsilon$-fraction of the cosets, as stated in the proof idea of Theorem~\ref{algo: efficient regularity lemma}. As described earlier, it then discards the cosets that are not regular, and keeps only those cosets whose estimated $\Density$ is $\geq \zeta$. The approximation $A'$ of $A$, obtained in Lemma~\ref{Prop: prop of Green's regularity}, preserves the behavior of $A$ with the cosets. 

The second part is the covering lemma (Lemma~\ref{lem:image-cover}), which states that if the set $A$ has density $\geq \zeta$, a single linear form that satisfies the non-degeneracy hypothesis, that is, every coefficient of the equation $L(\mathbf{x})=\sum_{j=1}^{m}a_{j}x_{j}=0$ is a unit modulo the torsion of the group, $\gcd(a_j,\lcmG) =1$, for all $j \in [m]$ (which again ensures that the map that takes a character $\chi$ to $\chi^{a_j}$ is an automorphism over $\widehat{\G}$), is almost completely contained in ($a_1A+\cdots +a_mA$). More precisely, Lemma~\ref{lem:image-cover} shows that $L(A'^m)$ covers all
but an $\varepsilon^2/\zeta^{2m}$ fraction of every coset in
\(
  \bar L(S^m)\;=\;a_1S+\cdots+a_mS\;\subseteq\;\bar\G ,
\)
and lands in no other coset, where $\bar L:\bar\G^{\,m}\to\bar\G$ be the form induced by
$L$ on the quotient group $\bar\G=\G/H$. This generalizes to the case of a matrix
$M\in\mathbb Z^{r\times m}$, whose number of pseudo--independent (Definition~\ref{defn_independence}) rows is at most $r$, and Cauchy--Schwarz complexity -$1$ (see Definition~\ref{def:cs1_our_result}), where we want to compute the density of the image
\[
   M(A^m)=\{M\mathbf x:\mathbf x\in A^m\}\subseteq \G^{r}.
\]
Note that to be able to prove \emph{Covering Lemma}, here also one has to define suitable notion of Cauchy–Schwarz complexity as merely the non-degeneracy condition will not be sufficient for this purposes. 
Here we use the first moment (a generalized von Neumann inequality, Lemma~\ref{lem:gvn}), which forces the
`Cauchy--Schwarz complexity-$1$' hypothesis, but, being an $\ell^{\infty}$ bound
over targets, yields exact covering of each reachable output coset. We then show that the algorithm (Algorithm~\ref{alg:lin-image}) outputs a value $v$, which is within an error of $\pm \theta$ from the density of $(a_1A'+\cdots+a_mA')$. This gives us the desired result. This again can be extended to the general case of the matrix (Algorithm~\ref{alg:matrix-image}). See Section~\ref{sec:lin-image} for more details.
\end{description}

\subsection{Technical Challenges Encountered in Establishing the Applications}

\begin{itemize}

    \item \textbf{Non-invertible coefficients and Fourier-analytic
    counting for linear equations.}
    Unlike the setting of $\mathbb{F}_2^n$, multiplication by a
    nonzero integer need not be invertible over a general finite
    Abelian group $\G$. Consequently, the map
    $\chi\mapsto\chi^a$ need not be a permutation of
    $\widehat{\G}$, preventing a direct application of the
    standard Fourier-analytic counting argument. We overcome
    this difficulty by imposing suitable non-degeneracy
    conditions, requiring the relevant coefficients to be
    coprime to the exponent $\tau$ of $\G$. This ensures
    that the corresponding character maps are automorphisms,
    allowing Fourier reindexing and the application of
    Parseval's identity. Using these properties, we establish
    a counting lemma for general linear equations
    \[
        L(\mathbf{x})=\sum_{j=1}^{m}a_jx_j=0,
    \]
    including the inhomogeneous case needed when restricting
    configurations to cosets of a subgroup.

    \item \textbf{Determinantal non-degeneracy and exact
    parameterization of linear systems.}
    For systems of linear equations
    $M\mathbf{x}=0$, where $M\in\mathbb{Z}^{r\times m}$,
    ordinary rank considerations are insufficient over
    groups of composite exponent, since the coefficient
    ring $\mathbb{Z}_{\tau}$ may contain zero divisors.
    We address this difficulty by imposing the
    determinantal non-degeneracy condition
    \[
        \gcd(d_r(M),\tau)=1,
    \]
    where $d_r(M)$ denotes the $r$-th determinantal
    divisor of $M$. Using Smith normal form, we show
    that this condition allows the system to be reduced
    to $[I_r\mid 0]$ through invertible transformations
    over $\mathbb{Z}_{\tau}$. Consequently, the solution
    space admits a bijective parameterization by
    $\G^{m-r}$, yielding
    \[
        |\ker M|=|\G|^{m-r}.
    \]
    Moreover, since $\exp(H)\mid\tau$ for every subgroup
    $H\leq\G$, this non-degeneracy condition is preserved
    upon restriction to $H$. These properties are essential
    for obtaining the correct normalization of configuration
    densities in the generalized counting and removal
    arguments.

    \item \textbf{Torsion-compatible Cauchy--Schwarz complexity
    and generalized counting lemmas.}
    The usual Cauchy--Schwarz complexity condition,
    formulated in terms of linear spans over fields,
    does not directly provide the surjectivity properties
    required for the counting argument over finite
    Abelian groups of composite exponent. In particular,
    noncontainment in a rational linear span does not
    guarantee the required surjectivity modulo $\tau$.
    We therefore introduce a torsion-compatible
    complexity-$1$ condition requiring, for every
    $i\in[m]$, a partition
    $[m]\setminus\{i\}=J_1\sqcup J_2$ such that
    \[
        \ell_i\left(
        \bigcap_{j\in J_s}\ker\ell_j
        \right)=\G,
        \qquad s\in\{1,2\}.
    \]
    This condition guarantees the uniformity of the
    shifts arising in the two Cauchy--Schwarz steps,
    allowing us to control the counting error through
    Fourier fourth-moment estimates. We further establish
    that the torsion-compatible condition is inherited
    by every subgroup $H\leq\G$, so that the generalized
    counting lemma remains applicable to the regular
    cosets produced by our constructive decomposition.
    This is crucial for the removal argument: a single
    configuration in the structured subset $A'$ can
    then be shown to force many configurations in the
    original set $A$, yielding a dimension-independent
    algorithmic removal guarantee for the specified
    class of non-degenerate complexity-$1$ systems.

\end{itemize}

\section{An Algorithmic Regularity Lemma with Low Query Complexity}
\label{section3:regularity algorithm}
Throughout this section, $\G$ denotes a finite Abelian group $\mathbb{Z}_{p_1^{m_1}}\times \cdots \times \mathbb{Z}_{p_T^{m_T}}$, where $p_i$ are primes for all $i \in [T]$ and may not be distinct, and $\omega_\lcmG$ denotes a primitive $\lcmG^{th}$ root of unity, where $\lcmG = \LCM \{p_1^{m_1}, \ldots, p_T^{m_T}\}$. We start by proving Green's regularity lemma for Finite Abelian group using thus Green's zero-radius Bohr neighbourhood~\cite{green2005szemeredi} while the proof in ~\cite{green2005szemeredi} allows non-zero radius for Bohr neighbourhood. Note that the original proof of ~\cite{green2005szemeredi} holds for any Abelian groups, while for the sake of an efficient algorithm we first reproduce the proof of the Regularity lemma for Abelian groups of constant torsion (that is $\lcmG$ is a constant). The proof is presented in Appendix~\ref{section_Green_regularity}.

\begin{theorem}[Green's regularity lemma~\cite{green2005szemeredi}]\label{thm_green_regularity}

    Let $A \subseteq \G$ and $\varepsilon> 0$. There exists a normal subgroup $H \subseteq \G$ with cosets $\{H_i\}$ such that 

    \begin{itemize}
        \item $H$ can be derived on making $\Tower(\lcmG,O(\frac{1}{\varepsilon^3}))$ many fixing; that is, the pseudo-codimension (see Definition~\ref{defn_codimension}) of $H$ is at most $\Tower(\lcmG,O(\frac{1}{\varepsilon^3}))$. 
        \item for all but $\varepsilon$-fraction of the cosets of $H$, the functions $\mathbbm{1}_A|_{H_i}: H \to \{0,1\}$ is $\varepsilon$-regular , where 
        $$\mathbbm{1}_A|_{H_i}(x) = 
\begin{cases}
1, & x+a_i\in H_i\cap A \\
0, & x+a_i \notin H_i\cap A,
\end{cases}$$
where $H_i = H + a_i$ with $a_i \in \G$.
    \end{itemize}
\end{theorem}

The proof of this theorem is given in the appendix (Section~\ref{section_Green_regularity}).

Before moving to the main statement, we will also describe some preliminaries and facts about Fourier structure of the functions on Abelian group which will be main tools of our algorithm.

\subsection{Fourier structure of restriction and truncation}
\label{subsec:truncated-restrictedfunction}

To prove Theorem~\ref{th: Regularity for Abelian}, we need the relation between the Fourier coefficients of the restricted function and the truncated function, as the Algorithm~\ref{algo: efficient regularity lemma} works with the truncated function instead of the restricted function. We define the restricted function and the truncated function first.

\begin{definition}[Restricted function]\label{defn_restricted_function}
   Let $\G$ be a finite Abelian group and $A \le \G$ be a subgroup. The {\em restricted} function $f|_A : A \to \{0,1\}$, with respect to $A$, is defined as  
    \begin{align}
        f\mid_A(x) := f(x)
    \end{align}
    for all $x$ in $A$. 
\end{definition}

\begin{definition}[Truncated function]\label{defn_truncated_function}
    Let $\G$ be a finite Abelian group and $A \le \G$ be a subgroup. The {\em truncated} function, with respect to $A$, $f\|_A : \G \to \{0,1\}$ is 
    \[
    f\|_A(x) =
    \begin{cases}
    f(x) & x \in A, \\
    0 & x \notin A.
    \end{cases}
    \]
\end{definition}

\subsection*{Relation between restricted and truncated function}\label{section_relation_res_trun}
For a function $f: \G \to \{0,1\}$, we will try to understand the Fourier structure of the restricted and the truncated function.

We denote by $\widehat{\G}$ the group of {\em characters} of $\G$. For $\chi \in \widehat{\G}$, recall that the Fourier coefficient $\widehat{f}(\chi)$ with $\chi$ is given by 
\begin{align*}
    \widehat{f}(\chi) = \mathbb{E}_{x \in \G}[f(x)\overline{\chi(x)}].
\end{align*}
Similarly, for $\psi \in \widehat{A}$, we have
\begin{align*}
    \widehat{f\mid_{A}}(\psi) = \mathbb{E}_{x \in A}[f(x)\overline{\psi(x)}].
\end{align*}
For $\chi \in \widehat{\G}$, define restriction of $\chi$ to $A$ by $\chi|_A(x) := \chi(x)$, for all $x$ in $A$. 

The {\em annihilator} $A^\perp$ of $A$ is defined as 
$$
    A^\perp = \left\{ \chi \in \widehat{\G} : \chi(a) = 1, \forall a \in A \right\}.
$$

Observe that 
$$
    \mathrm{res} : \widehat{\G} \to \widehat{A},
$$
where $\chi \mapsto \chi\mid_{A}$, is a group homomorphism. Additionally, we have 
\begin{lemma}
The map $\mathrm{res}$ is surjective and $\ker(\mathrm{res}) = A^\perp$.
\end{lemma}
\begin{proof}
It is straightforward that $A^\perp \subseteq \ker(\mathrm{res})$.
Conversely, if $\chi|_A$ is trivial, then $\chi(a)=1$ for all $a \in A$,
so $\chi \in A^\perp$.
Surjectivity also follows directly.
\end{proof}
As a corollary of the above lemma, we get the following: 
\begin{corollary}
    $\widehat{A} \cong \widehat{\G} / A^\perp$.
\end{corollary}

We will need the following theorem in the sumset paper. 
\begin{theorem}\label{thm:main_result_1}
For every $\chi \in \widehat{\G}$, we have 
$$
    \widehat{f\|_A}(\chi) = \frac{|A|}{|\G|} \cdot \widehat{f|_A}(\chi|_A).
$$
\end{theorem}

\begin{proof}
By definition,
$$
    \widehat{f\|_A}(\chi) = \mathbb{E}_{x \in \G}[f\|_A(x)\overline{\chi(x)}].
$$
Since $f\|_A(x)=0$ for $x \notin A$, this becomes
$$
    \widehat{f\|_A}(\chi) = \frac{1}{|\G|} \sum_{x \in A} f(x)\overline{\chi(x)}.
$$
Rewriting,
$$
    \widehat{f\|_A}(\chi) = \frac{|A|}{|\G|} \cdot \mathbb{E}_{x \in A}[f(x)\overline{\chi(x)}].
$$
Since $\chi(x)$ restricted to $A$ is $\chi|_A(x)$,
$$
    \widehat{f\|_A}(\chi) = \frac{|A|}{|\G|} \cdot \widehat{f|_A}(\chi|_A).
$$
\end{proof}

The following theorem gives the annihilator connection. 

\begin{theorem}\label{thm:main_result_2}
The function $\widehat{f\|_A}$ is constant on cosets of $A^\perp$.
\end{theorem}

\begin{proof}
If $\chi_1, \chi_2 \in \widehat{\G}$ satisfy $\chi_1|_A = \chi_2|_A$,
then $\chi_1^{-1}\chi_2 \in A^\perp$.
Hence $\chi_1$ and $\chi_2$ lie in the same coset of $A^\perp$. Also, observe that, from Theorem~\ref{thm:main_result_1}, we have
$$
\widehat{f\|_A}(\chi_1)
= \frac{|A|}{|\G|} \cdot \widehat{f|_A}(\chi_1|_A)
= \frac{|A|}{|\G|} \cdot \widehat{f|_A}(\chi_2|_A)
= \widehat{f\|_A}(\chi_2).
$$
\end{proof}

\begin{corollary}
For any $\psi \in \widehat{A}$ and any $\chi \in \widehat{\G}$ with $\chi|_A=\psi$, we have 
$$
    \widehat{f\|_A}(\chi) = \frac{|A|}{|\G|} \cdot \widehat{f|_A}(\psi).
$$
\end{corollary}

\subsection{Algorithmizing the Regularity Lemma (Theorem~\ref{thm_green_regularity})}
\label{section:proof_of_efficient_regularity}
In this section, we present an algorithmic version of the regularity lemma (Theorem~\ref{th: Regularity for Abelian}) for functions over an Abelian group, whose both query and time complexity are independent of the order of the group. We work in the standard query model. The input is the Boolean function $\mathbbm{1}_A:\G\to\{0,1\}$, which is the indicator function of the set $A\subseteq \G$, where $\G$ is a finite Abelian group. The function $\mathbbm{1}_A$ is not known to the algorithm. Instead, it can query an oracle that, on input $x\in \G$, returns the value $\mathbbm{1}_A(x)$. The complexity of an algorithm is measured by the total number of oracle queries it makes, referred to as its query complexity. We first state and prove Lemma~\ref{Prop: prop of Green's regularity}, before proving Theorem~\ref{th: Regularity for Abelian}.

From Green's regularity lemma on Abelian group it follows that given any arbitrary subsets $A \subseteq \G$, and error parameter $\varepsilon$, one can get a partition of $\G$  of size $\Tower(\lcmG, \frac{1}{\varepsilon^3})$ such that $A$ is almost regular on all of the sets of the partition. The following lemma (Lemma~\ref{Prop: prop of Green's regularity}) gives us an approximation, $A'$ of the set $A\subseteq \G$ by removing all cosets that are not $\varepsilon$-regular, along with those which have density below a threshold $\zeta$. Consequently, every remaining coset is both dense and $\varepsilon$-regular.

\begin{lemma}[$A'\subseteq A$: A density-capturing representation of $A$]
\label{Prop: prop of Green's regularity}
Given $A\subseteq \G$ and $\varepsilon,\zeta>0$, there exists a subgroup
$H\leq \G$ of pseudo-codimension at most
\[
    \Tower\left(\tau,O\left(\frac{1}{\varepsilon^3}\right)\right)
\]
and a set $A'\subseteq A$ such that
\begin{enumerate}
    \item $\Lambda(A\setminus A')\leq \varepsilon+\zeta;$

    \item for every coset $H_i$ of $H$,
    $\Lambda_{H_i}(A')=0
        \qquad\text{or}\qquad
        \Lambda_{H_i}(A')\geq \frac{\zeta}{2};$

    \item for every coset $H_i$ for which
    $\Lambda_{H_i}(A')>0$, the restricted function $
        \mathbbm{1}_{A'}|_{H_i} $
    is $\varepsilon$-regular.
\end{enumerate}
We call $A'$ a \textbf{density-capturing representation} of $A$.
\end{lemma}

\color{black}
\begin{proof}
    Let $H \subseteq \G$ be a subgroup of pseudo--codimension (Definition~\ref{defn_codimension}) at most $\Tower(\lcmG, O(\frac{1}{\varepsilon^3}))$ that we get by applying Green's regularity Theorem~\ref{thm_green_regularity}. Also let $\{H_1, H_2, \dots, H_M\}$ be the cosets of $H$ where $M= \frac{|\G|}{|H|}$. Recall the definition of the density of a subset $A$ in $\G$, $\Density$ (Definition~\ref{defn_density}) and the density of a subset $A$ in the subgroup $H$, $\Density_{H}$ (Definition~\ref{defn_density_H}). It is also guaranteed by Theorem~\ref{thm_green_regularity} that for all but $\varepsilon$-fraction of $\{H_i\}$, the function $\mathbbm{1}_A|_{H_i}: H_i \to \zone$ is $\varepsilon$-regular. 
    Define disjoint subsets $A_1', \dots, A_M'$, where each $A_i' \subseteq A \cap H_i$, as follows:
    \begin{itemize}
        \item If $\mathbbm{1}_A|_{H_i}$ is not $\varepsilon$-regular, then $A_i'=\emptyset$; 
        \item If $\Density_{H_i}(A) \leq \zeta$, then $A_i'=\emptyset$;
        \item Otherwise, set $A_i'= A \cap H_i$.
    \end{itemize}
    Now we define $A' \subseteq \G$ as 
    \begin{align*}
        A' = \cup_i A_i'.
    \end{align*}
    Note that $\mathbbm{1}_A|_{H_i}'$ is $\varepsilon$-regular for all $i \in [M]$ and by the construction of $A'$ it follows that $\Density_{H_i}(A')=0$ or $\Density_{H_i}(A') \geq \zeta$.

\textbf{Claim:} $\Density(A\setminus A') \leq \varepsilon+ \zeta$.

Let us define two sets $S_1$ and $S_2$ in the following way:
$$S_1 = \{i: \mathbbm{1}_A|_{H_i} \text{ is not } \varepsilon \text{-regular}\}, \quad \quad \quad S_2 = \{i: \Density_{H_i}(A) \leq \zeta\}$$
Then,
$$A' = \cup_i A_i' = \cup_{i\notin S_1 \cup S_2} (A\cap H_i).$$

So, $A\setminus A' = \cup_{i\in S_1 \cup S_2} (A\cap H_i) = \biggl(\cup_{i\in S_1} (A\cap H_i)\biggr) \cup \biggl(\cup_{i\in S_2} (A\cap H_i)\biggr)$.

Therefore,

\begin{align*}
    \Density(A\setminus A') &\leq \frac{|\cup_{i \in S_1}A\cap H_i|}{|\G|} + \frac{|\cup_{i \in S_2}A\cap H_i|}{|\G|} \\
    &\leq \varepsilon \times M \max_{H_i, i \in S_1} \frac{|A\cap H_i|}{|H_i|} \cdot \frac{|H_i|}{|\G|} + \sum_{i \in S_2}\frac{|A\cap H_i|}{|H_i|} \cdot \frac{|H_i|}{|\G|} \\
    &= \varepsilon \times M \times \max_{H_i, i \in S_1}\bigl(\Density_{H_i}(A) \cdot \frac{1}{M}\bigr) + \sum_{i\in S_2}\Density_{H_i}(A) \cdot \Density(H_i) \\
    &\leq \varepsilon \times M \times 1\times \frac{1}{M} + \zeta \\
    &\leq \varepsilon + \zeta,
\end{align*}
where the inequality in the second last line follows from the following facts:
\begin{itemize}
    \item $|S_1| \leq \varepsilon M$, as at most $\varepsilon M$ many cosets are non-regular,

    \item $\Density_{H_i}(A) \leq 1$ for $H_i, i\in S_1$, $M = \frac{|\G|}{|H_i|}$, and

    \item $\Density_{H_i}(A) \leq \zeta$ for $i \in S_2$, and $\sum \Density(H_i) \leq 1$ as $H_i$'s form a partition of $\G$.
\end{itemize}
\end{proof}

Now we are ready to state the algorithmic version of Green's regularity lemma for finite Abelian groups, Theorem~\ref{th: Regularity for Abelian}.

\begin{theorem}[Algorithmic version of Green's regularity lemma for finite Abelian groups]
\label{th: Regularity for Abelian}
There is an algorithm $\mathrm{ConstDT}(A,\varepsilon,\zeta)$ that,
given query access to a subset $A\subseteq \G$, a regularity parameter
$\varepsilon>0$, and a density threshold $\zeta>0$, runs in time and
query complexity bounded by
\[
    \operatorname{poly}\left(
        \Tower\left(\tau,O(\frac{1}{\varepsilon^3})\right),
        \frac{1}{\zeta}
    \right).
\]
With probability at least $9/10$, it outputs an probabilistic oracle
$O_{A'}$ and a character decision tree $T$ for a subgroup
$H\leq \G$ of pseudo-codimension

   $ \Tower\left(\tau,O\left(\frac{1}{\varepsilon^3}\right)\right),$
together with a structural subset $A'\subseteq A$ such that:
\begin{enumerate}
    \item
        $\Lambda(A\setminus A')\leq \varepsilon+\zeta;$

    \item every coset $K$ of $H$ satisfies
    $  \Lambda_K(A')=0
        \qquad\text{or}\qquad
        \Lambda_K(A')\geq \frac{\zeta}{2};$

    \item for every coset $K$ with $\Lambda_K(A')>0$,
    the restricted indicator function $
        \mathbbm{1}_{A'}|_K $
    is $\varepsilon$-regular.
\end{enumerate}
The number of cosets of $H$ is bounded by
$\frac{|\G|}{|H|}
    \leq
    \Tower\left(\tau,O\left(\frac{1}{\varepsilon^3}\right)\right).$ The query complexity remains independent of the size of the group $|\G|$.
\end{theorem}

\color{black}

\begin{algorithm}[]
    \textbf{Input.}
    Given query access to $A \subseteq \G$, regularity parameter $\varepsilon$, density threshold $\zeta$. 
    
    \textbf{Output.} An probabilistic oracle $\Or_{A'}^A, $ a character decision tree $\mathcal{T}$ with $A'$ as described in Lemma~\ref{Prop: prop of Green's regularity}.\\ 
 \vspace{0.7cm}
 \textbf{ConstDT$(A,\varepsilon,\zeta)$}

\begin{enumerate}
    
    \item \textbf{Initialization: }  Initialize the decision tree $\mathcal{T}$ to contain one leaf labelled by $\mathbbm{1}_A: \G \to \{0,1\}$. Define, 
    \begin{align*}
        \eta = \frac{1}{10} \times \frac{1}{\Tower(\lcmG,{O(\frac{1}{\varepsilon^3})})}
    \end{align*}
    \item for 
    \begin{enumerate}
     \item Let $\{O_1,\ldots,O_k\}$, where
$k = O_{\tau}(\log M),$
denote the query oracles for the $k$ characters used by the decision tree at the current stage. 
These characters determine the $M$ cosets $H_1,\ldots,H_M$ corresponding to the leaves of the decision tree. The $i$-th leaf is labelled by the function $\mathbbm{1}_A|_{H_i}.$

\item For each coset $H_i$, call,
       \begin{align*}
           \mathcal{L}_i\gets \mathsf{LIST_{Corr}(\mathbbm{1}_A\|_{H_i}, \frac{\varepsilon}{M},\eta)},
       \end{align*}
 $\mathbbm{1}_A\|_{H_i}$ is the truncated function (Definition~\ref{defn_truncated_function}) on coset $H_i$. To simulate query access to $\mathbbm{1}_A\parallel_{H_i}$ on an input
$x\in G$, the algorithm queries the
$k=O_{\tau}(\log M)$ character oracles
$\{O_1,\ldots,O_k\}$ to determine the coset containing $x$. Discard from \(\mathcal{L}_i\) any character that is constant on \(H_i\), since such a character does not witness non-regularity.
 \item For each $\mathcal{L}_i$ and for each $\alpha \in \mathcal{L}_i$, call,
 \begin{align*}
     \{\widehat{\mathbbm{1}_A|_{H_i}}(\alpha): \alpha \in \mathcal{L}_i\} \gets \mathsf{ESTIMATE(\widehat{\mathbbm{1}_A|_{H_i}}(\alpha) , \varepsilon, \eta)}.
 \end{align*}
 
 estimate the magnitude of 
$\widehat{\mathbbm{1}_A|_{H_i}}(\alpha)$ up to additive error $\frac{\varepsilon}{4}$ with confidence $\eta$. If the estimated magnitude is less than $\frac{3\varepsilon}{4}$, then remove $\alpha$ from $\mathcal{L}_i$.
         \item If $\mathcal{L}_i = \phi$ for at least $(1-\varepsilon)$-fraction of the set list $\{\mathcal{L}_1,\dots,\mathcal{L}_M\}$ go to step-3.
         \item Let $\mathcal{O}$ be the collection of the oracles that describes of the internal nodes. For each non-empty $\mathcal{L}_i$ do the following: 
         \begin{itemize}
             \item Choose the oracle for $\alpha$ from $\mathcal{L}_i$. Check, whether
the candidate character $\alpha$ is generated by the characters represented
by $\mathcal{O}$, with failure probability at most $\eta$. 
             \item If not, then we add the oracle of $\alpha$, $O_{\alpha}$ to $\mathcal{O}$. Add $\{\alpha\}$ to $\mathcal{T}$ and split \emph{all nodes} at the current stage according to the restrictions $\alpha$.
         \end{itemize}
         \item Repeat Step-2.
    \end{enumerate}
     \item For each leaf node, corresponding to some cosets (say) $H_i$, estimate $\Density_{H_i}(A_i)$ upto an additive error of $\frac{\zeta}{4}$ with confidence $\eta$.
     \begin{itemize}
         \item If the  $\Density_{H_i}(A_i) \geq \frac{3\zeta}{4}$, set the function associated with the leaf node to be identically $1$ function.
         \item Else, set it to be the identically $0$ function. 
     \end{itemize}
     \item Define the oracle $\Or_{A'}^A=\mathcal{T}(x)\cdot A(x)$.
    \end{enumerate}
\caption{A constructive Regularity lemma for Abelian groups}
\label{algo: efficient regularity lemma}
\end{algorithm}

\subsection*{Proof Of Theorem~\ref{th: Regularity for Abelian}}\label{sec_efficient_regularity}

\begin{proof}
Algorithm~\ref{algo: efficient regularity lemma} takes blackbox access to a set $A \subseteq \G$ as input, along with regularity parameter $\varepsilon$ and density threshold $\zeta$. At any given point of the algorithm let $H$ be the subgroup generated by the nodes of the tree and the leaves are labelled by $\mathbbm{1}_A|_{H_i}: H \to \{0,1\}$ is $\varepsilon$-regular , where 
        $$\mathbbm{1}_A|_{H_i}(x) = 
\begin{cases}
1, & x+a_i\in H_i\cap A \\
0, & x+a_i \notin H_i\cap A,
\end{cases}$$
where $H_i = H + a_i$ with $a_i \in \G$.


Note that the algorithm uses the truncated function (Definition~\ref{defn_truncated_function}) in the process. Using Lemma~\ref{thm:main_result_1} (see Section~\ref{section_relation_res_trun}), we can relate the coefficients of the restricted function (Definition~\ref{defn_restricted_function}) with the truncated function (Definition~\ref{defn_truncated_function}). We will analyze the steps of Algorithm~\ref{algo: efficient regularity lemma} now.

\begin{description}
    \item[Step 2 of Algorithm~\ref{algo: efficient regularity lemma}:] In Step (2), the algorithm applies List Correction for Abelian groups, Algorithm~\ref{algo:Implicit_Sieve_generalized} repeatedly to find a subset $A$ of $\G$ that is $\varepsilon$-regular. By Green's regularity lemma for finite Abelian groups (Theorem~\ref{thm_green_regularity}), Step 2 of Algorithm~\ref{algo: efficient regularity lemma} terminates after at most $O(\frac{1}{\varepsilon^3})$ many steps. Let the total number of cosets at a certain stage be $M$. If number of survived cosets is $\leq \varepsilon M$, the algorithm moves to Step 3. 

\begin{itemize}
    \item \textbf{Step 2(b) and 2(c) of Algorithm~\ref{algo: efficient regularity lemma}:} In the first iteration, Step 2(b) of Algorithm~\ref{algo: efficient regularity lemma} runs List Correction for Abelian groups, Algorithm~\ref{algo:Implicit_Sieve_generalized} on the function $\mathbbm{1}_A\|_{H_i}$. The goal of this stage is to detect a non-regular Fourier-coefficient. Usually finding one Fourier character is quite costly, but we do not need it exactly, we can go away by learning the Fourier character implicitly also. That's why, the output of List Correction for Abelian groups, Algorithm~\ref{algo:Implicit_Sieve_generalized} gives a list $\mathcal{L}_i$ of oracles to the non-regular Fourier coefficient.  In other word, it gives the description of those Fourier character as an implicit description in terms of oracle. Among the heavy Fourier characters returned by the List Correction algorithm, we ignore those characters that are constant on the subgroup \(H\), since such characters do not witness failure of regularity.
    
    \
    
    Here we are running the Algorithm~\ref{algo:Implicit_Sieve_generalized} on the truncated function $\mathbbm{1}_A\|_{H_i}$. Note that using Theorem~\ref{thm:main_result_1} from the previous subsection~\ref{subsec:truncated-restrictedfunction}, 
    the Fourier
coefficients of the restricted and truncated functions satisfy,

\[
\left|
\widehat{\mathbbm{1}_A\parallel_{H_i}}(\alpha)
\right|
=
\frac{1}{M}
\left|
\widehat{\mathbbm{1}_A|_{H_i}}(\alpha|_H)
\right|,
\]
where $M$ is the number of cosets of $H$ at the current stage.
Consequently,
\[
\left|
\widehat{\mathbbm{1}_A|_{H_i}}(\alpha|_H)
\right|>\varepsilon
\quad\Longleftrightarrow\quad
\left|
\widehat{\mathbbm{1}_A\parallel_{H_i}}(\alpha)
\right|>
\frac{\varepsilon}{M}.
\]
Therefore, Algorithm~\ref{algo:Implicit_Sieve_generalized} is called on
$\mathbbm{1}_A\parallel_{H_i}$ with threshold $\varepsilon/M$ in
order to obtain a list of candidate characters corresponding to non-regular Fourier coefficients of the restricted function. One can also simulate the query of  $\widehat{\mathbbm{1}_A|_{H_i}}(\alpha)$ using the query access of $\widehat{\mathbbm{1}_A\|_{H_i}}(\alpha)$.
    
    \
    
    With the help of these oracles, one can query the value of those characters (at certain inputs), which corresponds to the non-regular or heavy coefficients of the function $\mathbbm{1}_A\|_{H_i}$. In Step 2(c), it estimates the corresponding Fourier coefficients of $\mathbbm{1}_A\|_{H_i}$. If the estimate is less than $\frac{3\varepsilon}{4}$, then remove it from the list, else it keeps in the list. We retain only those heavy characters that vary non-trivially on \(H\), as these correspond to the Fourier coefficients relevant to the \(\varepsilon\)-regularity condition.
    
Next it chooses one heavy character from the list in Step 2(e), and splits all the nodes of the current decision tree by fixing that character. So it splits each nodes into atmost $\lcmG$ many nodes at this stage. This step of the algorithm continues till it reaches a point where the number of survived cosets (heavy Fourier coefficient corresponding to the cosets) is at most $\varepsilon$ times the total number of cosets at that stage.

\noindent

\

So, in the first iteration, Step 2(b) of Algorithm~\ref{algo: efficient regularity lemma} takes error parameter $\varepsilon$ and confidence $\eta$. It makes $O\biggl( \frac{1}{\varepsilon^{10}} \log\biggl(\frac{1}{\varepsilon^{24}}\log \frac{1}{\eta\varepsilon^2}\biggr) \log \frac{1}{\eta\varepsilon^2} \biggr)$ many queries at this stage, by Theorem~\ref{query_implicit_sieve}. 

At the $2^{nd}$ stage of Step 2(b), if the loop continues, number of cosets $M$ that has survived is $\geq \varepsilon \lcmG^{1+\lcmG}$. So, at the $i^{th}$ stage of Step 2(b), if the loop continues, number of cosets $M$ that has survived is $\leq \Tower(\lcmG, i)$. Therefore, the error at the $i^{th}$ iteration is $\frac{\varepsilon}{\Tower(\lcmG, i)}$. Number of queries made at this stage is $$O\biggl( \frac{(\Tower(\lcmG, i))^{12}}{\varepsilon^{12}} \log \frac{(\Tower(\lcmG, i))^{2}}{\eta\varepsilon^2} \biggr) \log\biggl(\frac{(\Tower(\lcmG, i))^{10}}{\eta\varepsilon^{10}}\log \frac{(\Tower(\lcmG, i))^{2}}{\eta\varepsilon^2}\biggr).$$ Therefore, at the final stage of Step 2(b), total number of queries made is 

$O\biggl( \frac{(\Tower(\lcmG,{O(\frac{1}{\varepsilon^3})}))^{12}}{\varepsilon^{12}} \log\biggl(\frac{(\Tower(\lcmG,{O(\frac{1}{\varepsilon^3})}))^{10}}{\eta\varepsilon^{10}}\log \frac{(\Tower(\lcmG,{O(\frac{1}{\varepsilon^3})}))^{2}}{\eta\varepsilon^2}\biggr) \log \frac{(\Tower(\lcmG,{O(\frac{1}{\varepsilon^3})}))^{2}}{\eta\varepsilon^2} \biggr)$.

Now, the total number of cosets involved in the whole process in the worst case is again $\leq \Tower(\lcmG,{O(\frac{1}{\varepsilon^3})})$. 

Note that there are various other randomized procedures are going in the next steps of the algorithms and when fixing $\eta$ for the worst case scenario we have to take account of all such randomized procedures. But here is an important observation at each coset only a constant number of randomized procedures
are invoked, namely List Correction, Fourier-coefficient estimation, and
the independence test, together with the final density estimation in
Step 3. We absorb this constant number of randomized calls into the
constant hidden in the tower bound. Thus, assigning failure probability
at most $\eta$ to each such randomized call and taking a union bound over
all calls still gives an overall failure probability at most $1/10$.

Hence, taking confidence at each step $ \eta = \frac{1}{10} \times \frac{1}{\Tower(\lcmG,{O(\frac{1}{\varepsilon^3}))}}$ works.

Step 2(c), number of queries made is equal to $O(\frac{1}{\varepsilon^2}\log\frac{1}{\eta})$ on reusing queries.

\item \textbf{Step 2(d) of Algorithm~\ref{algo: efficient regularity lemma}:}
Recall the relation between the restricted and the truncated function (Section~\ref{section_relation_res_trun}). Upon termination at Step 2(d), the algorithm ensures that for at least a $(1-\varepsilon)$-fraction of the cosets $H_i$, the list $\mathcal{L}_i = \emptyset$. Since $\mathsf{LIST_{Corr}}$ (see Appendix~\ref{section_implicit_sieve}) is guaranteed to return the matrix $\mathcal{O}$, which columns represent the values of those characters (at sufficiently many points)) with Fourier coefficients larger than $\varepsilon$, $\mathcal{L}_i = \emptyset$ implies that:
$$ \max_{\alpha \neq 0} \lvert \widehat{\mathbbm{1}_A|_{H_i}}(\alpha) \rvert < \varepsilon. $$
By definition, this means the restricted subfunction $\mathbbm{1}_A|_{H_i}$ is $\varepsilon$-regular for all but an $\varepsilon$-fraction of the cosets. For the $\varepsilon$-fraction of cosets that did not satisfy this, their lists were not empty, but the loop terminated because the threshold was met; these non-regular cosets will be implicitly handled in the next step.

\item \textbf{Step 2(e) of Algorithm~\ref{algo: efficient regularity lemma}:}
Since $\mathcal{O}$ is the collection of the oracles that describe the internal nodes, we can think of $\mathcal{O}$ as a matrix, whose columns denote the value of the characters evaluated at certain points, each distributed uniformly and independently in $\mathcal{G}$, which the oracles inside $\mathcal{O}$ output on querying. We proceed along the lines of the proof of correctness of Algorithm~\ref{algo: estimation}. Let the matrix $\mathcal{O}$ be a $k_i \times N_i$ matrix at the $i^{th}$ iteration; that is, its $j^{th}$ column denotes the values of the character evaluated at $k_i$ many points, each distributed uniformly and independently in $\mathcal{G}$, which the oracle $O_j$ outputs on querying.
So, the columns are labelled by $1, \dots, N_i$, which correspond to some $M_1', \dots, M_{N_i}' \in \mathcal{G}$ (note that the $M_{n_j}$'s are not known).
The $(i,j)$th entry of $Q$ is in $\{1, \omega_\lcmG, \ldots, \omega_\lcmG^{\lcmG-1}\}$. 

Now, let us consider all the group automorphisms on $\G$ which are formed by only choosing the images of the elements in $S'=\{M_1', \dots, M_{N_i}'\}$, which are pseudo--independent by construction. Then, by the property of group automorphism, the images of any other element that are not pseudo--independent from $M_1', \dots, M_{N_i}'$ are defined automatically. That is, for an element $\alpha \notin \{M_1', \dots, M_{N_i}'\}$, let $\lambda_1, \ldots, \lambda_{|N_i|}, \lambda \in \mathbb{Z}_\lcmG$, with $\lambda$ invertible in $\mathbb{Z}_\lcmG$ be such that $\lambda \alpha = \sum_{j=1}^{N_i} \lambda_i M_{j}'$. If $\psi$ is a group automorphism on $\G$, then $\psi(\alpha) = \frac{1}{\lambda} \sum_{j=1}^{N_i} \lambda_i \psi(M_{j}')$.

Proceeding along the lines of the proof of correctness of Algorithm~\ref{algo: estimation}, the probability that an $\alpha$ is chosen which is pseudo--dependent on $S'$, is $\lcmG^{-k_i} \times \frac{\varphi(\lcmG)}{\lcmG}$. So, for a large enough choice of $k_i$, a pseudo--independent $\alpha$ is chosen with probability $1- \lcmG^{-k_i} \times \frac{\varphi(\lcmG)}{\lcmG}$, which is $\geq \frac{2}{3}$ whenever a large enough $k_i$ is chosen. We use a brute-force algorithm to do this. Note that, since the number of cosets that have survived in the $i^{th}$ iteration is $\geq \Tower(\lcmG, i)$, the error at the $i^{th}$ iteration is at most $\frac{\varepsilon}{\Tower(\lcmG, i)}$. Therefore, $k_i = O(\frac{(\Tower(\lcmG, i))^{8}}{\varepsilon^8} \log \frac{(\Tower(\lcmG, i))}{\eta\varepsilon})$, and $N_i = O(\frac{(\Tower(\lcmG, i))^{2}}{\varepsilon^2})$, from Algorithm~\ref{algo:Implicit_Sieve_generalized} and \ref{algo: estimation}.
\end{itemize}

\item[Step 3 of Algorithm~\ref{algo: efficient regularity lemma}:]
For a coset $H_i$, $$\Density_{H_i}(A_i) = \frac{|A_i\cap H_i|}{|H_i|}.$$

Also, 
\begin{align*}
    \widehat{\mathbbm{1}_{A_{i}}|_{H_i}}(\chi_0) &= \frac{1}{|H_i|} \sum_{x\in H_i} \mathbbm{1}_{A_{i}}(x) \chi_0(x) = \frac{1}{|H_i|} \sum_{x\in H_i} \mathbbm{1}_{A_{i}}(x) = \mathbb{E}_{x\in H_i} [\mathbbm{1}_{A_{i}}(x)] \\
    &= \frac{|A_i\cap H_i|}{|H_i|} = \Density_{H_i}(A_i).
\end{align*}

So, $\Density_{H_i}(A_i)$ estimation can be done in the same way as $\wt_2$-estimation (see Section~\ref{sec_wt_2_estimation}). With error $\frac{\zeta}{4}$ and confidence $\eta$, this can be done using $O(\frac{1}{\zeta^2}\log\frac{1}{\eta})$ many queries.

In this step, the algorithm computes an estimate of the density of $A$ for each leaf, that is, coset
$H_i$, up to an additive error of $\zeta/4$. We consider the following two cases:

\begin{itemize}
\item \textbf{Case 1:} If the estimate of the true density $\Lambda_{H_i}(A)$ is less than
$3\zeta/4$, then with high probability
\[
\Lambda_{H_i}(A)
<
\frac{3\zeta}{4}+\frac{\zeta}{4}
=
\zeta.
\]
The algorithm sets $T(x)=0$ on this coset, which means
$A'\cap H_i=\emptyset$. This satisfies the condition
$\Lambda_{H_i}(A')=0$.

\item \textbf{Case 2:} If the estimate of the true density $\Lambda_{H_i}(A)$ is at least
$3\zeta/4$, then with high probability
\[
\Lambda_{H_i}(A)
\geq
\frac{3\zeta}{4}-\frac{\zeta}{4}
=
\frac{\zeta}{2}.
\]
The algorithm sets $T(x)=1$ on this coset, which means
$A'\cap H_i=A\cap H_i$.
This guarantees
\[
\Lambda_{H_i}(A')
=
\Lambda_{H_i}(A)
\geq
\frac{\zeta}{2}.
\]
\end{itemize}

\item[Step 4 of Algorithm~\ref{algo: efficient regularity lemma}:] The output oracle is defined as $\Or_{A'}^A(x) = \mathcal{T}(x) \cdot A(x)$. Since $\mathcal{T}(x) \in \{0,1\}$ acts as the indicator function of the union of dense, regular cosets, the function $\Or_{A'}^A(x)$ is probabilistic oracle for the indicator function of the set $A'$ defined in Lemma~\ref{Prop: prop of Green's regularity}.

By the exact same counting logic proven in the claim of Lemma~\ref{Prop: prop of Green's regularity}, removing the non-regular cosets (at most an $\varepsilon$-fraction) and cosets discarded by the density test (density $\le \zeta$) bounds the total discarded density by $\Density(A \setminus A') \leq \varepsilon + \zeta$. Thus, the oracle is correct.
\end{description}

\subsection{Query and time complexity of Algorithm~\ref{algo: efficient regularity lemma}}
Note that Step 2(b), 2(c), and Step 3 of Algorithm~\ref{algo: efficient regularity lemma} contribute to the query complexity of the algorithm. The correctness of this algorithm follows from the proof of correctness of List Correction for Abelian groups, Algorithm~\ref{algo:Implicit_Sieve_generalized} (Section~\ref{section_implicit_sieve}), Coefficient estimate algorithm~\ref{algo: estimation} (Section~\ref{section_estimation}), and from the proof of Green's regularity Lemma~\ref{thm_green_regularity} (Section~\ref{section_Green_regularity}).

By Theorem~\ref{thm_green_regularity}, Step 2 of Algorithm~\ref{algo: efficient regularity lemma} terminates after at most $O(\frac{1}{\varepsilon^3})$ many steps. Let the total number of cosets at a certain stage be $M$. If number of survived cosets is $\leq \varepsilon M$, the algorithm moves to Step 3. 
\begin{itemize}
    \item Queries made at the final stage of Step 2(b) is $$O\biggl( \frac{(\Tower(\lcmG,{O(\frac{1}{\varepsilon^3})}))^{12}}{\varepsilon^{12}} \log\biggl(\frac{(\Tower(\lcmG,{O(\frac{1}{\varepsilon^3})}))^{10}}{\eta\varepsilon^{10}}\log \frac{(\Tower(\lcmG,{O(\frac{1}{\varepsilon^3})}))^{2}}{\eta\varepsilon^2}\biggr) \log \frac{(\Tower(\lcmG,{O(\frac{1}{\varepsilon^3})}))^{2}}{\eta\varepsilon^2} \biggr).$$

Now, the total number of cosets involved in the whole process in the worst case is again $\leq \Tower(\lcmG,{O(\frac{1}{\varepsilon^3})}$. Taking confidence at each step $ \eta = \frac{1}{10} \times \frac{1}{\Tower(\varepsilon\lcmG,{O(\frac{1}{\varepsilon^3}))}}$ works, as stated earlier.

\item Step 2(c), number of queries made is equal to $O(\frac{1}{\varepsilon^2}\log\frac{1}{\eta})$ on reusing queries, where $ \eta = \frac{1}{10} \times \frac{1}{\Tower(\lcmG,{O(\frac{1}{\varepsilon^3}))}}$.

\item In Step 3, number of queries made is equal to $O(\frac{1}{\zeta^2}\log\frac{1}{\eta})$ on reusing queries, where $ \eta = \frac{1}{10} \times \frac{1}{\Tower(\lcmG,{O(\frac{1}{\varepsilon^3}))}}$.

\end{itemize}

Therefore, total queries made is equal to $$O\biggl( \frac{(\Tower(\lcmG,{O(\frac{1}{\varepsilon^3})}))^{12}}{\varepsilon^{12}} \log\biggl(\frac{(\Tower(\lcmG,{O(\frac{1}{\varepsilon^3})}))^{10}}{\eta\varepsilon^{10}}\log \frac{(\Tower(\lcmG,{O(\frac{1}{\varepsilon^3})}))^{2}}{\eta\varepsilon^2}\biggr) \log \frac{(\Tower(\lcmG,{O(\frac{1}{\varepsilon^3})}))^{2}}{\eta\varepsilon^2} \biggr),$$ where $ \eta = \frac{1}{10} \times \frac{1}{\Tower(\varepsilon\lcmG,{O(\frac{1}{\varepsilon^3}))}}$.

\medskip
 The time complexity of Algorithm~\ref{algo: efficient regularity lemma} is also at most $\poly(\Tower(\lcmG, O(\frac
{1}{\varepsilon^3})), \frac{1}{\zeta})$, where the tower-term comes from the depth of the character decision tree constructed in Green's regularity lemma (Theorem~\ref{thm_green_regularity}), exactly like the case of query complexity. It is also completely independent of the order of the group $\G$.
\end{proof}


\section{Testing linear configuration freeness}
\label{sec:linear-removal}

The primary objective of this section is to establish a constructive framework for testing freeness of linear configurations in finite Abelian groups of constant torsion, leading to the proofs of Theorem~\ref{thm:remove-single_informal_statement} and Theorem~\ref{thm:remove-system1_informal_statement}.

The core subroutine for all our testing algorithms will be the efficient regularity algorithm \textsc{ConstDT} (Algorithm~\ref{algo: efficient regularity lemma}).
We demonstrate that the subset $A'$ of $A$, satisfying the properties of Lemma~\ref{Prop: prop of Green's regularity} and produced by \textsc{ConstDT} (Algorithm~\ref{algo: efficient regularity lemma}) is free of the target linear configuration provided its density is sufficiently small. 
This structural guarantee is first proved for a single linear configuration in  Section~\ref{ssec:single} and then extended to matrix linear configuration of complexity-$1$ systems in Section~\ref{ssec:systems}.

\subsection{A single linear equation}
\label{ssec:single}

Let $L(\mathbf{x})=\sum_{j=1}^{m}a_{j}x_{j}=0$ with coefficients $a_{j}\in\mathbb{Z}$. We define the configuration density operator $\Lambda_L$ for functions $f_1, \dots, f_m: \G \to [0,1]$ and subsets $A \subseteq \G$ as:
\[
   \Lambda_{L}(f_1,\dots,f_m)
   =\mathbb{E}_{L(\mathbf{x})=0}\;\prod_{j=1}^{m}f_j(x_j),
   \qquad
   \Lambda_{L}(A)=\Lambda_{L}(\mathbbm{1}_A,\dots,\mathbbm{1}_A).
\]
Our non-degeneracy hypothesis requires every coefficient to be a unit modulo the torsion ($\gcd(a_j, \lcmG) = 1$). This algebraic condition ensures that the relevant character maps on $\widehat{\G}$ remain bijective, even in the presence of zero divisors within $\mathbb{Z}_\lcmG$.

\begin{lemma}[Counting lemma for a single equation]
\label{lem:count-single}
Let $m \ge 3$ and suppose at least three of the coefficients $a_j$ are units in $\mathbb{Z}_\lcmG$. Let $f_1,\dots,f_m:\G\to[0,1]$, and suppose some function $f_{j_0}$, whose corresponding coefficient $a_{j_0}$ is a unit, is $\varepsilon$-regular. Then
\[
   \Bigl|\,\Lambda_{L}(f_1,\dots,f_m)-\prod_{j=1}^{m}\mathbb{E}[f_j]\,\Bigr|\le\varepsilon.
\]
The same uniform bound holds for the inhomogeneous equation $L(\mathbf{x})=b$, for any $b\in\G$ lying in the image of $L$.
\end{lemma}

\begin{proof}
Expanding the solution indicator $\mathbbm{1}_{[L(\mathbf{x})=0]}$ into a character sum over the dual group $\widehat{\G}$ yields
\[
   \sum_{\mathbf{x}:\,L(\mathbf{x})=0}\prod_{j=1}^m f_j(x_j)
   =\frac{1}{|\G|}\sum_{\chi\in\widehat{\G}}\prod_{j=1}^m \sum_{x_j\in\G}f_j(x_j)\chi(a_j x_j)
   =|\G|^{m-1}\sum_{\chi\in\widehat{\G}}\prod_{j=1}^m c_j(\chi),
\]
where $c_j(\chi):=\widehat{f_j}\bigl(\overline{\chi^{a_j}}\bigr)=\mathbb{E}_x[f_j(x)\chi(a_j x)]$, 
at least one $a_j$ is a unit, one coordinate is uniquely determined by the remaining $m-1$ variables, so the number of solutions in $\G^m$ is exactly $|\G|^{m-1}$; hence $\Lambda_L(f_1,\dots,f_m)=\sum_{\chi}\prod_j c_j(\chi)$. 

The principal term corresponding to the trivial character $\chi=\chi_0$ equals $\prod_j\mathbb{E}[f_j]$. For the remainder of the sum, we isolate the regular index $j_0$ along with two further unit indices $j_1,j_2$, bounding all other factors trivially by $\lvert c_j(\chi)\rvert\le\mathbb{E}[f_j]\le1$:
\[
   \Bigl|\sum_{\chi\ne\chi_0}\prod_{j=1}^m c_j(\chi)\Bigr|
   \le\Bigl(\max_{\chi\ne\chi_0}\lvert c_{j_0}(\chi)\rvert\Bigr)
      \sum_{\chi\ne\chi_0}\lvert c_{j_1}(\chi)\rvert\,\lvert c_{j_2}(\chi)\rvert
   \le\varepsilon \cdot \lVert c_{j_1}\rVert_2\,\lVert c_{j_2}\rVert_2 .
\]
For any unit $a_j$, the mapping $\chi\mapsto\overline{\chi^{a_j}}$ is an automorphism of $\widehat{\G}$. Thus, by Parseval's identity, $\lVert c_j\rVert_2^2=\sum_{\chi}\lvert\widehat{f_j}(\overline{\chi^{a_j}})\rvert^2=\sum_{\psi}\lvert\widehat{f_j}(\psi)\rvert^2=\lVert f_j\rVert_2^2\le1$, which matches the error bound $\varepsilon$. For the inhomogeneous equation, each non-principal term simply picks up a unimodular phase $\overline{\chi(b)}$, leaving the absolute value estimate unchanged.
\end{proof}

The above counting lemma, a standard result in additive combinatorics, can be interpreted as stating that if a set is $\varepsilon$-regular, then the number of solutions to a linear system is close to the number one would expect if the elements of a set of the same density were chosen independently at random. Since ConstDT (Algorithm~\ref{algo: efficient regularity lemma}) guarantees that the restriction of $A'$ to each regular coset is $\varepsilon$-regular, the density of solutions within such a coset is approximately $(\frac{\zeta}{2})^m$, where $\frac{\zeta}{2}$ denotes the density of the set in that coset and $m$ is the number of variables in the linear configuration. In the proof below, we choose $\varepsilon<(\zeta/2)^m$. Consequently, the error term is too small to account for the disappearance of all solutions. Thus, the single solution that we initially assume exists must, in fact, expand into a large cluster of solutions within the coset.

\begin{theorem}[Algorithmic removal for a single equation]
\label{thm:remove-single}
Fix $t\ge2$ and let $L(\mathbf{x})=\sum_{j=1}^{m}a_jx_j=0$ with $m\ge3$ and $\gcd(a_j,\lcmG)=1$. For every target tolerance $\theta>0$, there exists a threshold $\delta=\delta(\theta,m,t)>0$, \emph{independent of $\G$}, such that the following holds for every finite Abelian group $\G$ with constant torsion $\lcmG$. Given query access to $A\subseteq\G$, Algorithm~\ref{alg:remove-single} requires at most
\[
   \poly\Bigl(\Tower\bigl(\lcmG,\;O_m(\theta^{-3m})\bigr)\Bigr)
\]
queries and time, both independent of the minimum number of elements required to generate $|\G|$, as well as $|\G|$. With probability at least $9/10$, it outputs an exact oracle for a subset $A'\subseteq A$ satisfying $|A\setminus A'|\le\theta\,|\G|$ such that, \emph{whenever} $\Lambda_{L}(A)\le\delta$, the set $A'$ contains absolutely no solutions to $L(\mathbf{x})=0$. One may explicitly take
\(
\delta=
\operatorname{Tower}\!\left(\tau,O_m(\theta^{-3m})\right)^{-1}.
\)

\end{theorem}

\begin{algorithm}[H]
\caption{\textsc{ConfigRemoval}$(A,L,\theta)$ \quad(single equation)}
\label{alg:remove-single}
\begin{algorithmic}[1]
\Require query access to $A\subseteq\G$; equation $L(\mathbf{x})=\sum_{j=1}^{m}a_jx_j=0$ with $m\ge3$ and $\gcd(a_j,\lcmG)=1$ for all $j$; global error parameter $\theta\in(0,1)$.
\Ensure an exact oracle $\mathcal{O}_{A'}$ and character decision tree $T$ for a set $A'\subseteq A$.
\State Set the density threshold $\zeta\gets\theta/2$.
\State Set the target local regularity parameter $\varepsilon\gets\tfrac12(\frac{\zeta}{2})^{\,m}$.
\State $(\mathcal{O}_{A'},T)\gets\textsc{ConstDT}(A,\varepsilon,\zeta)$. \Comment{Run the constructive regularity routine}
\State \Return $(\mathcal{O}_{A'},T)$.
\end{algorithmic}
\end{algorithm}

\begin{proof}[Proof of Theorem~\ref{thm:remove-single}]
The query/time bounds and the existence of $\mathcal{O}_{A'}$ follow from Theorem~\ref{th: Regularity for Abelian} evaluated at $\varepsilon=\tfrac12(\zeta/2)^m$ and $\zeta=\theta/2$. Under these parameters, $1/\varepsilon^3=O_m(\theta^{-3m})$, and the constant torsion $\lcmG$ updates all structural dependencies to the constant $\lcmG$. By Theorem~\ref{th: Regularity for Abelian}, the discarded volume satisfies $|A\setminus A'|\le(\varepsilon+\zeta)|\G|\le\theta|\G|$.

It remains to certify solution-freeness. Suppose for contradiction that there exists a valid solution $\mathbf{x}\in (A')^{m}$ to $L(\mathbf{x})=0$, with components $x_j\in c_j+H$. By Theorem~\ref{th: Regularity for Abelian}, each underlying coset carries $\Density_{c_j+H}(A')\ge \frac{\zeta}{2}$ and the restricted indicator function $\mathbbm{1}_{A'}|_{c_j+H}$ is $\varepsilon$-regular. 

Shifting variables to the core subgroup via $y_j:=x_j-c_j\in H$, the constraint translates to the inhomogeneous equation
\[
   \sum_{j=1}^{m}a_jy_j=b\quad\text{in }H,\qquad b:=-\sum_{j=1}^m a_jc_j\in H,
\]
where $b\in H$ precisely because $\mathbf{x}$ is a solution (so $b=\sum_j a_j(x_j-c_j)\in H$, since $a_jH\subseteq H$). Let $g_j:H\to[0,1]$ be the indicator function of $A'$ on $c_j+H$ pulled back to $H$. By construction, $\mathbb{E}_H[g_j]\ge\frac{\zeta}{2}$ and $g_j$ is $\varepsilon$-regular on $H$. Because $\exp(H)\mid\lcmG$, the coefficients $a_j$ remain coprime to the internal torsion of $H$. Applying Lemma~\ref{lem:count-single} inside the subgroup $H$ yields:
\[
   \mathbb{E}_{\sum_j a_jy_j=b}\;\prod_{j=1}^m g_j(y_j)
   \;\ge\;\prod_{j=1}^m\mathbb{E}[g_j]-\varepsilon
   \;\ge\;(\frac{\zeta}{2})^{\,m}-\varepsilon > \frac{1}{2}(\frac{\zeta}{2})^m,
\]
since $\varepsilon=\tfrac12(\frac{\zeta}{2})^m$. The number of solutions counted here is at least $\tfrac12 (\frac{\zeta}{2})^m \,|H|^{\,m-1}$. Every such structural tuple $(c_j+y_j)_j$ lies in $(A')^m\subseteq A^m$ and satisfies $L(\mathbf{x})=0$. Embedding this back into the global group yields the density bound:

\[
\Lambda_L(A)
\ge
\frac{\frac12(\zeta/2)^m}{M^{m-1}}
\ge
\frac{\frac12(\zeta/2)^m}
{\operatorname{Tower}\!\left(\tau,O(1/\varepsilon^3)\right)^{m-1}}
=: \delta_0.
\]

Since \(1/\varepsilon^3=O_m(\theta^{-3m})\), we may take
\(
\delta=
\operatorname{Tower}\!\left(\tau,O_m(\theta^{-3m})\right)^{-1}.
\)
Setting our global threshold $\delta<\delta_0$ directly contradicts the hypothesis $\Lambda_L(A)\le\delta$. Therefore, $A'$ is completely solution-free.
\end{proof}

Observe that testing triangle-freeness in subsets of the vector space $\mathbb{F}_{p}^{n}$ is a special case of the above theorem, obtained by taking $m=3$, $a=(1,1,-1)$, and $\G=\mathbb{F}_{p}^{n}$. A polynomial dependence on $1/\varepsilon$ for testing triangle-freeness in this setting is already known~\cite{fox2017tight}. Our result substantially generalizes this by showing that triangle-freeness, as well as the absence of a broad class of linear configurations, can be tested in arbitrary finite Abelian groups of constant torsion.

Our contribution yields a query-efficient, constructive, and dimension-independent tester under the constant-torsion assumption. The hypothesis that $\gcd(a_j,\lcmG)=1$ for every coefficient $a_j$ is also structurally necessary. For example, the equation defining a $3$-term arithmetic progression,
\[
x-2y+z=0,
\]
contains the coefficient $-2$, which is invertible precisely when $\lcmG$ is odd. This coincides with the well-known degeneracy of $3$-term arithmetic progressions in characteristic~$2$.

\subsection{Complexity-\texorpdfstring{$1$}{1} systems}
\label{ssec:systems}

We now treat systems of equations $M\mathbf{x}=0$ where $M\in\mathbb{Z}^{r\times m}$ is an integer matrix, whose number of pseudo--independent (Definition~\ref{defn_independence}) rows is at most $r$. To proceed via basic Fourier methods, the system must satisfy two fundamental non-degeneracy conditions.

\begin{definition}[Determinantal non-degeneracy, Definition~\ref{def:det}, restated]
The $r$-th determinantal divisor $d_r(M)$ is the greatest common divisor of all $r\times r$ minors of $M$. We require the matrix to be coprime to the group torsion: $\gcd\bigl(d_r(M),\lcmG\bigr)=1$.
\end{definition}




The determinantal non-degeneracy condition
$\gcd(d_r(M),\lcmG)=1$ ensures that $M$
can be transformed into $[I_r\mid 0]$
using invertible row and column operations
over $\mathbb Z_\lcmG$.
Consequently, the solution space
$\{\mathbf{x}\in\G^m:M\mathbf{x}=0\}$
can be bijectively parameterized by
$\mathbf{t}\in\G^{m-r}$ through
integer-linear forms
$x_j=\ell_j(\mathbf{t})$.
In particular, the total number of
solutions is exactly $|\G|^{m-r}$.

To establish the counting lemma, we also
require a suitable Cauchy--Schwarz
complexity-$1$ condition on these forms.
The usual complexity condition, formulated
in terms of linear spans, does not directly
provide the surjectivity properties required
by our argument over general finite Abelian
groups. In particular, noncontainment in
the $\mathbb Q$-linear span does not imply
that
$\ell_i\left(
\bigcap_{j\in J}\ker\ell_j
\right)=\G.$
This difficulty arises because rational
linear independence need not be preserved
modulo the torsion of the group.

We therefore introduce a torsion-compatible
version of Cauchy--Schwarz complexity-$1$,
which explicitly guarantees the required
surjectivity of the coordinate forms on the
relevant common kernels. This allows the
generalized counting lemma to be established
using the Cauchy--Schwarz inequality and
first-order Fourier analysis.

We therefore use the following torsion-compatible version of
Cauchy--Schwarz complexity-$1$. The definition directly encodes the
surjectivity property required by the two Cauchy--Schwarz steps. Over
groups of prime exponent it reduces to the usual complexity-$1$
condition over the corresponding finite field.

\begin{definition}[Torsion-compatible Cauchy-Schwarz complexity-$1$]
\label{def:cs1}
Let $\ell_1,\ldots,\ell_m$ be the coordinate forms parameterizing the
solution space of $Mx=0$. We say that $M$ has Cauchy--Schwarz
complexity-$1$ if, for every $i\in[m]$, there exists a partition
\([m]\setminus\{i\}=J_1\sqcup J_2\)
such that, for $s\in\{1,2\}$,
\(\ell_i\notin\operatorname{span}_{\mathbb Q}\{\ell_j:j\in J_s\}
\text{ and, }
\ell_i\!\left(\bigcap_{j\in J_s}\ker\ell_j\right)=G.\)
\end{definition}


\begin{remark}
\label{rem:cs_complexity_surjectivity}
The coefficient-vector formulation in
Definition~\ref{def:cs1}
implies the required surjectivity directly.
Indeed, fix $s\in\{1,2\}$ and let
$v_{i,s}=(v_1,\ldots,v_d)$ be the vector guaranteed by the
definition. For any $y\in\G$, define
\(
    h_y:=(v_1y,\ldots,v_dy)\in\G^d.
\)
For every $j\in J_s$,
\(
    \ell_j(h_y)
    =
    (a_j\cdot v_{i,s})y
    =
    0,
\)
since $a_j\cdot v_{i,s}=0\pmod{\tau}$ and $\tau y=0$.
Hence
\(
    h_y\in W_{J_s}.
\)
On the other hand,
\(
    \ell_i(h_y)
    =
    (a_i\cdot v_{i,s})y
    =
    y.
\)
Therefore
\(
    \ell_i(W_{J_s})=\G.
\)

Moreover, the same argument holds verbatim for every subgroup
$H\leq\G$, since $\exp(H)\mid\tau$. Consequently, the
torsion-compatible complexity-$1$ condition is inherited when the
forms are restricted from $\G$ to any subgroup $H\leq\G$.

When $\tau=p$ is prime, this condition is equivalent to the usual
linear-algebraic condition
\(
    a_i
    \notin
    \operatorname{span}_{\mathbb F_p}
    \{a_j:j\in J_s\}.
\)
Indeed, over a field the latter condition is equivalent to the
existence of a vector $v$ annihilating all $a_j$, $j\in J_s$, while
satisfying $a_i\cdot v\neq 0$; rescaling $v$ gives
$a_i\cdot v=1$.
Thus, Definition~\ref{def:cs1} recovers the usual
Cauchy--Schwarz complexity-$1$ condition over fields.
\end{remark}

\color{black}

Using the above definition, we will show that the configuration density operator is completely controlled by the uniformity of the Fourier spectrum.

\begin{lemma}[Counting lemma for matrix systems]
\label{lem:count-system}
Let $M$ be a matrix that has number of pseudo-independent (Definition~\ref{defn_independence}) rows at most $r$, satisfies Definition~\ref{def:det}, and has complexity-$1$. If the functions $f_1,\dots,f_m:\G\to[0,1]$ are each $\varepsilon$-regular, then
\[
   \Bigl|\,\Lambda_{M}(f_1,\dots,f_m)-\prod_{j=1}^{m}\mathbb{E}[f_j]\,\Bigr|
   \;\le\;2^{m}\,\varepsilon^{1/2}.
\]
The same bound holds identically for inhomogeneous systems $M\mathbf{x}=\mathbf{b}$ provided $\mathbf{b}\in\operatorname{im}M$.
\end{lemma}
\begin{proof}
By Definition~\ref{def:det}, the parametrization $\mathbf{x}=(\ell_j(\mathbf{t}))_j$ forms a perfect bijection from $\G^{m-r}$ to the kernel $\{\mathbf{x}:M\mathbf{x}=0\}$. Thus, the configuration density simplifies to:
\[
   \Lambda_M(f_1,\dots,f_m)=\mathbb{E}_{\mathbf{t}\in\G^{m-r}} \prod_{j=1}^{m}f_j(\ell_j(\mathbf{t})).
\]
We decompose each function as $f_j=\mu_j+g_j$, where $\mu_j=\mathbb{E}[f_j]$ is the mean component and $g_j$ is a mean-zero error term bounded by $\lVert g_j\rVert_\infty\le1$. By our regularity assumption, $\max_{\chi}\lvert\widehat{g_j}(\chi)\rvert = \max_{\chi\ne\chi_0}\lvert\widehat{f_j}(\chi)\rvert\le\varepsilon$. Expanding the product generates $2^m$ terms indexed by subsets $S\subseteq[m]$, where the base case $S=\varnothing$ yields the principal term $\prod_j\mu_j$. Since $0\le\mu_j\le1$, it suffices to show that
\[
   \Lambda_S:=\mathbb{E}_{\mathbf{t}\in\G^{m-r}}\prod_{j\in S}g_j(\ell_j(\mathbf{t}))
\]
satisfies $\lvert\Lambda_S\rvert\le\varepsilon^{1/2}$ for every non-empty $S\subseteq[m]$.

Fix a non-empty subset $S$ and pick an index $i\in S$; for notational convenience set $g_j\equiv1$ for $j\notin S$, so that $\Lambda_S=\mathbb{E}_{\mathbf{t}}\prod_{j\ne i}g_j(\ell_j(\mathbf{t}))\,g_i(\ell_i(\mathbf{t}))$. Let $[m]\setminus\{i\}=J_1\sqcup J_2$ be the index partition provided by the complexity-$1$ condition (Definition~\ref{def:cs1}), and for $J\in\{J_1,J_2\}$ define the subgroup
\[
   W_J:=\bigcap_{j\in J}\ker\bigl(\ell_j:\G^{m-r}\to\G\bigr)\;\le\;\G^{m-r}.
\]

By Definition~\ref{def:cs1},
\(
\ell_i(W_{J_1})=\ell_i(W_{J_2})=G.
\)
In particular, the restrictions of $\ell_i$ to $W_{J_1}$ and
$W_{J_2}$, and hence $\ell_i$ itself, are surjective onto $G$.


\textbf{First Cauchy--Schwarz (eliminating $J_1$).} The factor $F(\mathbf{t}):=\prod_{j\in J_1}g_j(\ell_j(\mathbf{t}))$ is invariant under translation by $W_{J_1}$, so writing $H(\mathbf{t}):=\prod_{j\in J_2\cup\{i\}}g_j(\ell_j(\mathbf{t}))$ and using $\lVert F\rVert_\infty\le1$,
\[
   \lvert\Lambda_S\rvert^{2}
   =\Bigl|\mathbb{E}_{\mathbf{t}}\,F(\mathbf{t})\,\mathbb{E}_{\mathbf{w}\in W_{J_1}}H(\mathbf{t}+\mathbf{w})\Bigr|^{2}
   \le\mathbb{E}_{\mathbf{t}}\Bigl|\mathbb{E}_{\mathbf{w}\in W_{J_1}}H(\mathbf{t}+\mathbf{w})\Bigr|^{2}
   =\mathbb{E}_{\mathbf{h}\in W_{J_1}}\mathbb{E}_{\mathbf{t}}\,H(\mathbf{t})\,\overline{H(\mathbf{t}+\mathbf{h})}.
\]

\textbf{Second Cauchy--Schwarz (eliminating $J_2$).} For each fixed $\mathbf{h}\in W_{J_1}$, the factor 
$$
    \prod_{j\in J_2}g_j(\ell_j(\mathbf{t}))\,\overline{g_j(\ell_j(\mathbf{t})+\ell_j(\mathbf{h}))}
$$ 
is invariant under translation of $\mathbf{t}$ by $W_{J_2}$; applying the Cauchy--Schwarz inequality in $\mathbf{t}$ exactly as above (and averaging over $\mathbf{h}$) leaves only the $j=i$ factor:
\[
   \lvert\Lambda_S\rvert^{4}
   \le\mathbb{E}_{\mathbf{h}\in W_{J_1},\,\mathbf{h}'\in W_{J_2}}\;\mathbb{E}_{\mathbf{t}}\;
   g_i(y)\,\overline{g_i(y+a)}\,\overline{g_i(y+b)}\,g_i(y+a+b),
\]
where $y=\ell_i(\mathbf{t})$, $a=\ell_i(\mathbf{h})$, $b=\ell_i(\mathbf{h}')$. 

Since the restrictions
\(
\ell_i|_{W_{J_1}}:W_{J_1}\to G
\text{ and, }
\ell_i|_{W_{J_2}}:W_{J_2}\to G
\)
are surjective homomorphisms, $a=\ell_i(h)$ and
$b=\ell_i(h')$ are uniform on $G$ when $h$ and $h'$ are
uniform on $W_{J_1}$ and $W_{J_2}$, respectively.
Likewise, $\ell_i$ is surjective, so $y=\ell_i(t)$ is uniform on $G$.
Since $t,h,h'$ are independent, so are $y,a,b$.
By the surjectivity noted above, $a$ is uniform on $\G$ as $\mathbf{h}$ ranges over $W_{J_1}$, likewise $b$ over $W_{J_2}$, and $y$ is uniform on $\G$; moreover $(y,a,b)$ are independent. Hence
\[
   \lvert\Lambda_S\rvert^{4}
   \le\mathbb{E}_{y,a,b\in\G}\;g_i(y)\,\overline{g_i(y+a)}\,\overline{g_i(y+b)}\,g_i(y+a+b)
   =\sum_{\chi\in\widehat{\G}}\lvert\widehat{g_i}(\chi)\rvert^{4},
\]
where the identity follows by expanding each factor into its Fourier series: the averages over $a$, $b$, and $y$ annihilate, via the orthogonality of characters, every term in which the four characters do not all coincide, leaving exactly the diagonal contribution $\sum_{\chi}\lvert\widehat{g_i}(\chi)\rvert^{4}$.

Finally, since $g_i$ is mean-zero, $\widehat{g_i}(\chi_0)=0$, and by Parseval's identity,(Theorem~\ref{theorem_Parseval_groups}),
\[
   \sum_{\chi\in\widehat{\G}}\lvert\widehat{g_i}(\chi)\rvert^{4}
   =\sum_{\chi\ne\chi_0}\lvert\widehat{g_i}(\chi)\rvert^{4}
   \le\Bigl(\max_{\chi\ne\chi_0}\lvert\widehat{f_i}(\chi)\rvert\Bigr)^{2}\,\lVert g_i\rVert_2^{2}
   \le\varepsilon^{2},
\]
so that $\lvert\Lambda_S\rvert\le\varepsilon^{1/2}$. Summing over all $2^m$ terms yields our uniform error bound. Because these steps rely purely on the bijective variable substitution, the orthogonality of characters, and Parseval's identity, the inequality holds over any finite Abelian group. The inhomogeneous case follows identically: shifting each variable by a particular solution multiplies each Fourier coefficient by a unimodular phase, leaving every absolute-value estimate unchanged.
\end{proof}

\begin{algorithm}[H]
\caption{\textsc{ConfigRemoval}$(A,M,\theta)$ \quad(complexity-$1$ system)}
\label{alg:remove-system}
\begin{algorithmic}[1]
\Require query access to $A\subseteq\G$; matrix $M\in\mathbb{Z}^{r\times m}$, whose number of pseudo--independent (Definition~\ref{defn_independence}) rows is at most $r$, with $m-r\ge1$, complexity-$1$, and $\gcd(d_r(M),\lcmG)=1$; global error parameter $\theta\in(0,1)$.
\Ensure an exact oracle $\mathcal{O}_{A'}$ and character decision tree $T$ for a set $A'\subseteq A$.
\State Set the density threshold $\zeta\gets\theta/2$.
\State Set the target local regularity parameter $\varepsilon\gets \zeta^{\,2m}/(4\cdot 4\cdot 4^{m}) = \zeta^{2m}/4^{m+2}$.
\State $(\mathcal{O}_{A'},T)\gets\textsc{ConstDT}(A,\varepsilon,\zeta)$.
\State \Return $(\mathcal{O}_{A'},T)$.
\end{algorithmic}
\end{algorithm}

\begin{theorem}[Algorithmic removal for complexity-$1$ systems]
\label{thm:remove-system}
Fix $m>r\ge1$. Let $M\in\mathbb{Z}^{r\times m}$ be a matrix, whose number of pseudo--independent (Definition~\ref{defn_independence}) rows is at most $r$, is of complexity-$1$, and $\gcd(d_r(M),\lcmG)=1$. For every target tolerance $\theta>0$, there exists a threshold $\delta=\delta(\theta,m,r,\lcmG)>0$, independent of $|\G|$, such that for every finite Abelian group $\G$ with constant torsion $\lcmG$, Algorithm~\ref{alg:remove-system} makes at most
\[
   \poly\Bigl(\Tower\bigl(\lcmG,\;O_m(\theta^{-6m})\bigr)\Bigr)
\]
queries and time, both independent of $n$ and $|\G|$. With probability at least $9/10$, it outputs an exact oracle for a subset $A'\subseteq A$ satisfying $|A\setminus A'|\le\theta\,|\G|$ such that, whenever $\Lambda_{M}(A)\le\delta$, the set $A'$ contains absolutely no solutions to $M\mathbf{x}=0$. One may explicitly take \(
\delta=
\operatorname{Tower}\!\left(\tau,O_m(\theta^{-6m})\right)^{-1}.
\)

\end{theorem}

\begin{proof}
The structural framework matches the proof of Theorem~\ref{thm:remove-single}, substituting Lemma~\ref{lem:count-system} in place of Lemma~\ref{lem:count-single}. Our choice of regularity parameter $\varepsilon=\zeta^{2m}/(4^{m+2})$ ensures that the spectral error component evaluates to exactly half the main density term ($2^{m}\varepsilon^{1/2} = \tfrac12 (\frac{\zeta}{2})^{m}$), satisfying the strict positivity lower bound required by our counting lemma. Invoking \textsc{ConstDT} at parameters $(\varepsilon,\zeta)$ limits our discarded set size to $|A\setminus A'|\le(\varepsilon+\zeta)|\G|\le\theta|\G|$ with the required tower complexity torsion scaling as $1/\varepsilon^3=O_m(\theta^{-6m})$.

Suppose for contradiction that a solution $\mathbf{x}\in(A')^m$ to $M\mathbf{x}=0$ exists with components falling into cosets $x_j\in c_j+H$. Restricting the system to these cosets isolates an inhomogeneous system $M\mathbf{y}=\mathbf{b}$ within the subgroup $H$, where the shift vector is $\mathbf{b}=-M\mathbf{c}\in H^{r}$. This restricted system inherits our complexity-1 structure and satisfies $\gcd(d_r(M),\exp H)=1$ since $\exp(H)\mid\lcmG$. Because our pulled-back indicator functions $g_j$ are $\varepsilon$-regular with $\mathbb{E}_H[g_j]\ge\frac{\zeta}{2}$, Lemma~\ref{lem:count-system} guarantees that the system contains at least
\[
   ((\frac{\zeta}{2})^{m}-2^{m}\varepsilon^{1/2})\,|H|^{m-r} = ((\frac{\zeta}{2})^m - \tfrac12 (\frac{\zeta}{2})^m)\,|H|^{m-r} = \tfrac12 (\frac{\zeta}{2})^m|H|^{m-r}
\]
solutions in $(A')^m\subseteq A^m$. Embedding this count back into the global group gives the density lower bound:

\[
\Lambda_M(A)
\ge
\frac{(\zeta/2)^m-2^m\varepsilon^{1/2}}
{M^{m-r}}
\ge
\frac{(\zeta/2)^m-2^m\varepsilon^{1/2}}
{\operatorname{Tower}\!\left(\tau,O(1/\varepsilon^3)\right)^{m-r}}
=: \delta_0.
\]
Since \(1/\varepsilon^3=O_m(\theta^{-6m})\), we may take
\(
\delta=
\operatorname{Tower}\!\left(\tau,O_m(\theta^{-6m})\right)^{-1}.
\) Choosing our threshold $\delta<\delta_0$ directly contradicts the structural hypothesis $\Lambda_M(A)\le\delta$.
\end{proof}

Theorem~\ref{thm:remove-system} covers genuine multi-equation systems ($r\ge2$), and both non-degeneracy conditions are essential. The co-primality condition $\gcd(d_r(M),\lcmG)=1$ ensures that the solution set is well-behaved: without it, the count of solutions is distorted by the interaction between the matrix and the group structure, and the main density term in the counting lemma is no longer correct. The same condition appears in the work of Kr\'al'--Serra--Vena~\cite{KSV2012}. The complexity-1 condition marks the natural limit of our method. Our arguments rely only on Fourier coefficients, and complexity-1 systems are precisely those whose solution counts are controlled by them. For systems of higher complexity, such as 4-term arithmetic progressions, small Fourier coefficients no longer guaranty the expected number of solutions, a fact already observed by Green~\cite{green2005szemeredi}. Removal lemmas for such systems are known via hypergraph methods~\cite{Shapira2009, KSV2012}, but these proofs are purely existential and give no query-efficient or constructive guarantees. In this sense, our result covers the full range of systems that can be handled constructively by basic Fourier analysis.

%
%

\section{Approximate image counting for a general linear form}
\label{sec:lin-image}

In this section, we ask how well one can estimate the size of the image of a set under a linear map. Given query access to a subset $A \subseteq \G$, we want to estimate the density of its image under a single linear form, $L(A^m) = a_1 A + \dots + a_m A$, and more generally under a matrix system, $M(A^m) \subseteq \G^r$, where $M$ has at most $r$ pseudo--independent (Definition~\ref{defn_independence}) rows. As before, we want the number of queries to be independent of the group's size $|\G|$. As before, we want the number of queries to be independent of the group's
size $|G|$. A known lower bound rules out estimating the size of the genuine
image with a bounded number of queries, so instead we slightly perturb $A$
into a dense, structured subset $A'$. The regularity partition produced by
Algorithm~\ref{algo: efficient regularity lemma} then allows us to estimate the image density
inside the small quotient $G/H$. It also yields an approximate membership
oracle in the single-form case and an exact membership oracle for the
matrix-system case.

\color{black}
Along the way, a natural contrast emerges with the removal results of Section~\ref{sec:linear-removal}: single forms are handled by a second-moment (variance) argument, while matrix systems require a first-moment argument and hence the complexity-1 hypothesis.

\subsection{A single general linear form}\label{ssec:lin-image}

In this subsection, we study the image sets generated by evaluating a single general linear form across our subset. Let us first define a linear image set.

\begin{definition}[Linear image set]
For integers $a_1,\dots,a_m\in\mathbb Z$ and the form $L(\mathbf x)=\sum_{j=1}^m a_jx_j$,
the \emph{image set} of $A\subseteq \G$ under $L$ is
\[
  L(A^m)\;=\;a_1A+\cdots+a_mA\;=\;\Bigl\{\,\textstyle\sum_{j=1}^m a_jx_j \;:\; x_j\in A\,\Bigr\}\subseteq \G.
\]
For $m=2$, $a=(1,1)$ this is the sumset $A+A$.
\end{definition}

The genuine image $L(A^m)$ cannot be estimated to any
nontrivial accuracy with $|\G|$-independent query complexity , so we estimate the density of the
image of a \emph{density-capturing} perturbation $A'\subseteq A$, exactly as in
Lemma~\ref{Prop: prop of Green's regularity}.


\begin{remark}[Running hypothesis for this subsection]
We assume every coefficient is a unit modulo the torsion,
$\gcd(a_j,\lcmG)=1$ for all $j$. Hence each map
$\chi\mapsto\chi^{a_j}$ an automorphism of $\widehat\G$, which is what the
Parseval steps below require in the presence of zero divisors. No
complexity-$1$ assumption is needed: at $r=1$ the second moment closes for all
$m\ge2$.
\end{remark}

\subsubsection*{The covering lemma}

We now show that if the set has sufficiently large density on each relevant
regular coset, then the image of the linear form covers almost all of the
corresponding output coset. Its proof is
the second-moment computation: the mean of the representation count over an
output coset is $\prod_j\mu_j$, the variance is at most $\varepsilon^2$, and
Chebyshev finishes.

\begin{lemma}[Image covering for a single form]
\label{lem:image-cover}
Let $m\ge2$ and $\gcd(a_j,\lcmG)=1$ for all $j$. Let $c_1,\dots,c_m\in \G$ be such
that, for every $j$, the restriction $\mathbbm{1}_A|_{c_j+H}$ is $\varepsilon$-regular
and $\Density_{c_j+H}(A)\ge\zeta$. Then, writing $w_0:=\sum_{j}a_jc_j$,
\[
  \Density_{\,w_0+H}\bigl(a_1A+\cdots+a_mA\bigr)\;\ge\;1-\frac{\varepsilon^{2}}{\zeta^{2m}} .
\]
\end{lemma}

\begin{proof}
Let $g_j:H\to\{0,1\}$ be the restriction translated to $H$, $g_j(h)=\mathbbm{1}_A(c_j+h)$,
so $\mathbb E_H[g_j]=\mu_j\ge\zeta$ and $g_j$ is $\varepsilon$-regular on $H$. For
$b\in H$ let
\[
  N(b)=\#\Bigl\{\mathbf h\in H^m : \textstyle\sum_j a_jh_j=b,\ g_j(h_j)=1\ \forall j\Bigr\}
\]
count the representations landing in the output coset $w_0+H$ at offset $b$, and
set $\widetilde N(b)=N(b)/|H|^{m-1}$. Using
$\mathbbm{1}_{[g=0]}=\tfrac1{|H|}\sum_{\chi}\chi(g)$ and
$\sum_{h}g_j(h)\chi^{a_j}(h)=|H|\,c_j(\chi)$ with
$c_j(\chi):=\widehat{g_j}(\overline{\chi^{a_j}})$,
\[
  \widetilde N(b)=\sum_{\chi\in\widehat H}\overline{\chi(b)}\prod_{j=1}^m c_j(\chi).
\]
Thus $\widetilde N$, as a function of $b$, has Fourier coefficients
$\prod_j c_j(\chi)$. Its mean is the $\chi_0$ term,
$\mathbb E_b[\widetilde N]=\prod_j c_j(\chi_0)=\prod_j\mu_j\ge\zeta^{m}$, and by
Parseval (Theorem~\ref{theorem_Parseval_groups}),
\[
  \operatorname{Var}_b[\widetilde N]
  =\sum_{\chi\ne\chi_0}\Bigl|\prod_{j}c_j(\chi)\Bigr|^{2}
  =\sum_{\chi\ne\chi_0}\prod_{j}|c_j(\chi)|^{2}.
\]
Fix $\chi\ne\chi_0$. Since $a_1$ is a unit, $\overline{\chi^{a_1}}\ne\chi_0$, so
$|c_1(\chi)|=|\widehat{g_1}(\overline{\chi^{a_1}})|\le\varepsilon$ by regularity;
and $|c_j(\chi)|\le\mathbb E|g_j|=\mu_j\le1$ for every $j$. Hence, bounding all
factors past the second by $1$,
\[
  \prod_{j}|c_j(\chi)|^{2}\le\varepsilon^{2}\,|c_2(\chi)|^{2},
  \qquad
  \operatorname{Var}_b[\widetilde N]\le\varepsilon^{2}\sum_{\chi}|c_2(\chi)|^{2}
  =\varepsilon^{2}\|g_2\|_2^{2}\le\varepsilon^{2},
\]
where $\sum_\chi|c_2(\chi)|^2=\|g_2\|_2^2$ because $a_2$ is a unit, so
$\chi\mapsto\overline{\chi^{a_2}}$ is an automorphism of $\widehat H$. By
Chebyshev,
\[
  \Pr_{b}\bigl[N(b)=0\bigr]
  \le\Pr_{b}\Bigl[\,\bigl|\widetilde N(b)-\textstyle\prod_j\mu_j\bigr|\ge\textstyle\prod_j\mu_j\Bigr]
  \le\frac{\operatorname{Var}_b[\widetilde N]}{(\prod_j\mu_j)^{2}}
  \le\frac{\varepsilon^{2}}{\zeta^{2m}} ,
\]
and $\Density_{w_0+H}(L(A^m))=\Pr_b[N(b)\ge1]\ge1-\varepsilon^{2}/\zeta^{2m}$.
\end{proof}

\subsubsection*{The induced form on the quotient}

Let $\bar\G=\G/H$ and let $\bar L:\bar\G^{\,m}\to\bar\G$ be the form induced by
$L$, $\bar L(\bar{\mathbf x})=\sum_j a_j\bar x_j$ (well-defined as $a_jH\subseteq H$).
Because each $a_j$ is a unit, $L(H^m)=H$ and the fibres of $L$ over $w_0+H$ have
size exactly $|H|^{m-1}$, so output cosets are precisely cosets of $H$ indexed by
$\bar\G$. 

Writing $S\subseteq\bar\G$ for the surviving leaves of
\textsc{ConstDT}, each surviving coset is
$\varepsilon$-regular and has density at least $\zeta/2$.
Applying Lemma~\ref{lem:image-cover} with density
parameter $\zeta/2$, we obtain that $L(A'^m)$ covers
all but an $\varepsilon^2/(\zeta/2)^{2m}$ fraction of
every coset in
\[
\bar L(S^m)=a_1S+\cdots+a_mS\subseteq\bar\G,
\]
and lands in no other coset. Hence,
\begin{equation}
\label{eq:vol-identity}
\Density\bigl(L(A'^m)\bigr)
=
\frac{|\bar L(S^m)|}{|\bar\G|}
\;\pm\;
\frac{\varepsilon^2}{(\zeta/2)^{2m}}.
\end{equation}
\color{black}
The right-hand side is a finite computation in the quotient $\bar\G$, whose size $|\bar\G|\le\Tower\bigl(\lcmG,O(1/\eta^3)\bigr)$
\color{black}
is independent of $|\G|$ and of $n$; the iterated
sum $a_1S+\cdots+a_mS$ is formed in $\bar\G$ in $O(m\,|\bar\G|^2)$ time.

\subsubsection*{The algorithm}

\begin{algorithm}[H]
\caption{\textsc{ApproxLinImage}$(A,L,\theta)$ \quad (single linear form)}
\label{alg:lin-image}
\begin{algorithmic}[1]
\Require query access to $A\subseteq\G$; form $L(\mathbf x)=\sum_{j=1}^m a_jx_j$
         with $\gcd(a_j,\lcmG)=1$ for all $j$; accuracy $\theta\in(0,1)$.
\Ensure  a value $v$ and a character-decision-tree oracle $\mathcal O_{L(A'^m)}$.

\State Set $\zeta\gets\theta/2$ and
$\eta\gets\tfrac12\,\theta^{1/2}(\zeta/2)^m$.
\Comment{So $\eta^2/(\zeta/2)^{2m}\le\theta$
and $\eta+\zeta\le\theta$}
\color{black}
\State $(\mathcal O_{A'},T)\gets\textsc{ConstDT}(A,\eta,\zeta)$; let
       $S\subseteq\bar\G=\G/H$ be the surviving (regular, density $\ge\frac{\zeta}{2}$)
       leaves of $T$.
\State Compute $\bar L(S^m)=a_1S+\cdots+a_mS\subseteq\bar\G$ by iterated addition
       in the quotient.
\State Build $\mathcal O_{L(A'^m)}$: on input $w\in \G$, use $T$ to read the coset
       $\bar w\in\bar\G$ and output $1$ iff $\bar w\in\bar L(S^m)$.
\State \Return $v\gets |\bar L(S^m)|/|\bar\G|$ and $\mathcal O_{L(A'^m)}$.
\end{algorithmic}
\end{algorithm}

\begin{theorem}[Approximate image counting for a single form]
\label{thm:lin-image}
Fix $\tau \ge2$, an integer $m\geq 2$, and a form $L(\mathbf x)=\sum_{j=1}^m a_jx_j$ with
$\gcd(a_j,\lcmG)=1$ for all $j$. For every $\theta>0$ and every finite
Abelian $\G$ with constant torsion $\lcmG$, given query access to $A\subseteq\G$,
Algorithm~\ref{alg:lin-image} makes
\[
   \poly\Bigl(\Tower\bigl(\lcmG,\,O_m(\theta^{-O(m)})\bigr)\Bigr)
\]
queries and time, both independent of $n$ and $|\G|$, and with probability
at least $9/10$ outputs a value $v$ and an oracle $\mathcal O_{L(A'^m)}$ such that,
for a density-capturing perturbation $A'\subseteq A$ with
$\Density(A\setminus A')\le\theta$,



\[
\left|v-\Lambda\bigl(L(A'^m)\bigr)\right|\leq\theta,
\]
and $O_{L(A'^m)}$ is a $(\theta,0,O(n))$-oracle
for the image set $L(A'^m)=a_1A'+\cdots+a_mA'$.

\color{black}

\end{theorem}

\begin{proof}
Theorem~\ref{th: Regularity for Abelian} (\textsc{ConstDT}) returns, at parameters $(\eta,\zeta)$, the oracle
$\mathcal O_{A'}$, the tree $\mathcal{T}$, and the surviving leaves $S$, within the stated complexity (here $1/\eta^3=O_m(\theta^{-O(m)})$); the
density-capturing set $A'$ satisfies $\Density(A\setminus A')\le\eta+\zeta\le\theta$
by Lemma~\ref{Prop: prop of Green's regularity}. 
By Lemma~\ref{lem:image-cover}, applied with density
parameter $\zeta/2$, every coset of $\bar L(S^m)$
is covered by $L(A'^m)$ up to an
$\eta^2/(\zeta/2)^{2m}\le\theta$ fraction, and
$L(A'^m)$ meets no coset outside $\bar L(S^m)$.
This gives
\(
\left|v-\Density(L(A'^m))\right|\le\theta,
\)
where
\(
v=\frac{|\bar L(S^m)|}{|\bar\G|}.
\)

Moreover, the oracle $\mathcal O_{L(A'^m)}$ outputs
$1$ precisely on the cosets indexed by $\bar L(S^m)$.
Therefore,
\(
\Pr_{w\in\G}
\left[
\mathcal O_{L(A'^m)}(w)
\neq \mathbbm{1}_{L(A'^m)}(w)
\right]
\le\frac{\eta^2}{(\zeta/2)^{2m}}
\le\theta.
\)
Thus, $\mathcal O_{L(A'^m)}$ is an approximate
membership oracle with error at most $\theta$.
\color{black}
\end{proof}

\subsection{Approximate image counting for a linear system}
\label{ssec:matrix-image}

We now generalise Section~\ref{ssec:lin-image} from a single form to a matrix
$M\in\mathbb Z^{r\times m}$, whose number of pseudo--independent (Definition~\ref{defn_independence}) rows is at most $r$, estimating the density of the image
\[
   M(A^m)=\{M\mathbf x:\mathbf x\in A^m\}\subseteq \G^{r}.
\]
Unlike the single-form case, this requires genuine restrictions on $M$, and for
a structural reason: at $r\ge2$ the second-moment argument of
Lemma~\ref{lem:image-cover} is unavailable, because the Fourier variance over
$\G^{r}$ is $\sum_{\boldsymbol\eta\ne\mathbf 0}\prod_j|\widehat g_j((M^{\top}\boldsymbol\eta)_j)|^{2}$
and the fibres of $\boldsymbol\eta\mapsto(M^{\top}\boldsymbol\eta)_j$ have size
$\approx|\G|^{r-1}$, so Chebyshev blows up. We therefore route through the
\emph{first} moment (a generalised Von Neumann inequality), which forces the
Cauchy--Schwarz complexity-$1$ hypothesis but, being an $\ell^{\infty}$ bound
over targets, yields \emph{exact} covering of each reachable output coset.

\textbf{Hypotheses we will use for the rest of the section.}
Let $\G$ be a finite Abelian group with constant torsion $\lcmG$. Also let $m>r$, and $d_r(M)$ (the $r$-th determinantal
divisor, i.e.\ the gcd of all $r\times r$ minors of $M$) be coprime to $\lcmG$.
Write $\Lambda_{M,\mathbf b}(f_1,\dots,f_m)=\mathbb E_{M\mathbf x=\mathbf b}\prod_j f_j(x_j)$
for the counting operator with right-hand side $\mathbf b\in \G^{r}$.



\begin{definition}[Complexity-$1$ over $\mathbb Z_\lcmG$]
\label{def:cs1-matrix}
Let $\ell_1,\ldots,\ell_m:\G^{m-r}\to\G$
be the coordinate forms of Lemma~\ref{lem:variety}.
We say that $M$ has Cauchy--Schwarz complexity-$1$
over $\mathbb Z_\lcmG$ if these forms satisfy the
torsion-compatible complexity-$1$ condition of
Definition~\ref{def:cs1}.
\end{definition}
\color{black}

\noindent

This condition guarantees the surjectivity properties
required by the two Cauchy--Schwarz steps and is
preserved when the forms are restricted to any
subgroup $H\leq\G$, as established in
Remark~\ref{rem:cs_complexity_surjectivity}.
\color{black}

\subsubsection*{The solution space}

\begin{lemma}[Parametrisation]
\label{lem:variety}
There are integer-linear forms $\ell_1,\dots,\ell_m:\G^{m-r}\to \G$ such that
$\Phi(\mathbf t)=(\ell_1(\mathbf t),\dots,\ell_m(\mathbf t))$ is a bijection from
$\G^{m-r}$ onto $\{\mathbf x:M\mathbf x=0\}$. Consequently $M$ is surjective,
$|\{\mathbf x:M\mathbf x=\mathbf b\}|=|\G|^{m-r}$ for every $\mathbf b\in \G^{r}$, and
\[
   \Lambda_{M,\mathbf b}(f_1,\dots,f_m)
   =\mathbb E_{\mathbf t\in \G^{m-r}}\ \prod_{j=1}^{m}f_j\bigl(\ell_j(\mathbf t)+x_{0,j}(\mathbf b)\bigr),
\]
where $\mathbf x_0(\mathbf b)$ is any particular solution, depending $\mathbb Z$-linearly on $\mathbf b$.
\end{lemma}

\begin{proof}
By the Smith normal form there are $U\in GL_m(\mathbb Z)$, $V\in GL_r(\mathbb Z)$
with $VMU=[\,\mathrm{diag}(d_1,\dots,d_r)\mid 0\,]$, and $d_1\cdots d_r=d_r(M)$ is
coprime to $\lcmG$, so each $d_i$ is a unit in $\mathbb Z_\lcmG$. Replacing $V$ by
$\mathrm{diag}(d_i)^{-1}V\in GL_r(\mathbb Z_\lcmG)$ gives $VMU\equiv[I_r\mid 0]\pmod\lcmG$.
As integer matrices invertible mod $\lcmG$, $U$ and $V$ act as automorphisms of
$\G^{m}$ and $\G^{r}$. Writing $\mathbf x=U\mathbf y$, $M\mathbf x=0$ becomes
$[I_r\mid0]\mathbf y=0$, i.e.\ $y_1=\dots=y_r=0$; the free variables
$\mathbf t=(y_{r+1},\dots,y_m)\in \G^{m-r}$ then give $\mathbf x=\Phi(\mathbf t)$ with
$\ell_j(\mathbf t)=\sum_{k}U_{j,\,r+k}\,t_k$. Surjectivity of $M$ and the fibre
count follow from $VMU\equiv[I_r\mid0]$, and the inhomogeneous formula is the
coset version.
\end{proof}


\subsubsection*{The generalised covering}

For a function $g:K\to\mathbb C$ on a finite Abelian group $K$, we will use the fourth moment of its Fourier coefficients, $\sum_{\chi\in\widehat K}|\widehat g(\chi)|^{4}$, which admits the physical-side expression
\[
   \sum_{\chi\in\widehat K}|\widehat g(\chi)|^{4}
   =\mathbb E_{y,a,b\in K}\ g(y)\,\overline{g(y+a)}\,\overline{g(y+b)}\,g(y+a+b),
\]
as one checks by expanding each factor into characters and using orthogonality. In particular, for mean-zero $g$ with $\|g\|_\infty\le1$,
\(
   \sum_{\chi\in\widehat K}|\widehat g(\chi)|^{4}
   \le\bigl(\max_{\chi\ne\chi_0}|\widehat g(\chi)|\bigr)^{2}\|g\|_2^{2}
   \le\bigl(\max_{\chi\ne\chi_0}|\widehat g(\chi)|\bigr)^{2}.
\)

\begin{lemma}
\label{lem:gvn}
Suppose $M$ has complexity $1$ over $\mathbb Z_\lcmG$. Then for all
$g_1,\dots,g_m:\G\to\mathbb C$ with $\|g_j\|_\infty\le1$ and every
$\mathbf b\in \G^{r}$,
\[
   \bigl|\Lambda_{M,\mathbf b}(g_1,\dots,g_m)\bigr|^{4}\;\le\;\min_{i\in[m]}\ \sum_{\chi\in\widehat\G}|\widehat{g_i}(\chi)|^{4} .
\]
\end{lemma}

\begin{proof}
Shifting each $g_j$ by $x_{0,j}(\mathbf b)$ reduces to $\mathbf b=\mathbf 0$ (each shift multiplies the Fourier coefficients of $g_j$ by unimodular phases, so the quantities below are unchanged), so we bound
$\Lambda:=\mathbb E_{\mathbf t\in \G^{m-r}}\prod_j g_j(\ell_j(\mathbf t))$. Fix $i$;
relabel so that $i=m$, and let $[m-1]=J_1\sqcup J_2$ be its partition. For
$J\in\{J_1,J_2\}$ put
\[
   W_J:=\bigcap_{j\in J}\ker\bigl(\ell_j:\G^{m-r}\to \G\bigr)\le \G^{m-r}.
\]

\emph{We claim that : 
$\ell_m(W_J)=\G$.}
By Definition~\ref{def:cs1-matrix}, there exists
$v=(v_1,\ldots,v_{m-r})\in\mathbb Z_\lcmG^{m-r}$
such that
\(
a_j\cdot v=0\pmod{\lcmG} \text{ for every }j\in J,
a_m\cdot v=1\pmod{\lcmG}.
\)
For any $y\in\G$, set
\(
h_y=(v_1y,\ldots,v_{m-r}y).
\)
Then $h_y\in W_J$ and $\ell_m(h_y)=y$.
Consequently,
\(
\ell_m(W_J)=\G.
\)
In particular, $\ell_m$ is surjective onto $\G$.


\textbf{First Cauchy--Schwarz (eliminate $J_1$).} The factor
$F(\mathbf t)=\prod_{j\in J_1}g_j(\ell_j(\mathbf t))$ is $W_1$-invariant, so with
$H(\mathbf t)=\prod_{j\in J_2\cup\{m\}}g_j(\ell_j(\mathbf t))$,
\[
   \Lambda=\mathbb E_{\mathbf t}\,F(\mathbf t)\,\mathbb E_{\mathbf w\in W_1}H(\mathbf t+\mathbf w),
   \qquad
   |\Lambda|^{2}\le\mathbb E_{\mathbf t}\Bigl|\mathbb E_{\mathbf w\in W_1}H(\mathbf t+\mathbf w)\Bigr|^{2}
   =\mathbb E_{\mathbf h\in W_1}\,\mathbb E_{\mathbf t}\,H(\mathbf t)\,\overline{H(\mathbf t+\mathbf h)} .
\]
For $\mathbf h\in W_1$,
$H(\mathbf t)\overline{H(\mathbf t+\mathbf h)}=\prod_{j\in J_2\cup\{m\}}
g_j(\ell_j(\mathbf t))\,\overline{g_j(\ell_j(\mathbf t)+\ell_j(\mathbf h))}$.

\textbf{Second Cauchy--Schwarz (eliminate $J_2$).} For fixed
$\mathbf h$ the factor $\prod_{j\in J_2}g_j(\ell_j(\mathbf t))\overline{g_j(\ell_j(\mathbf t)+\ell_j(\mathbf h))}$
is $W_2$-invariant; applying Cauchy--Schwarz in $\mathbf t$ (and over
$\mathbf h\in W_1$) as above leaves only the $j=m$ factor:
\[
   |\Lambda|^{4}\le\mathbb E_{\mathbf h\in W_1,\,\mathbf h'\in W_2}\ \mathbb E_{\mathbf t}\
   \Delta_{\mathbf h}g_m(\ell_m(\mathbf t))\,\overline{\Delta_{\mathbf h}g_m(\ell_m(\mathbf t)+\ell_m(\mathbf h'))},
\]
where $\Delta_{\mathbf h}g_m(y)=g_m(y)\overline{g_m(y+\ell_m(\mathbf h))}$.
Expanding,
\[
   |\Lambda|^{4}\le\mathbb E_{\mathbf h\in W_1,\,\mathbf h'\in W_2,\,\mathbf t}\
   g_m(y)\,\overline{g_m(y+a)}\,\overline{g_m(y+b)}\,g_m(y+a+b),
   \quad y=\ell_m(\mathbf t),\ a=\ell_m(\mathbf h),\ b=\ell_m(\mathbf h').
\]
By the Claim, $a=\ell_m(\mathbf h)$ is uniform on $\G$ as $\mathbf h$ ranges over
$W_1$, and likewise $b$ over $W_2$; and $y=\ell_m(\mathbf t)$ is uniform on $\G$
since $\ell_m$ is surjective. Hence the right-hand side is
$$
    \mathbb E_{y,a,b\in \G}\,g_m(y)\overline{g_m(y+a)}\,\overline{g_m(y+b)}\,g_m(y+a+b) =\sum_{\chi\in\widehat\G}|\widehat{g_m}(\chi)|^{4},
$$
by the identity stated at
the start of this subsection. As $i=m$ was an arbitrary index, the minimum
over $i$ follows.
\end{proof}

\begin{corollary}[Generalized counting lemma]
\label{cor:count-matrix}
Suppose $M$ has complexity $1$ over $\mathbb Z_\lcmG$. If $f_1,\dots,f_m:\G\to[0,1]$
are each $\varepsilon$-regular, then for \emph{every} $\mathbf b\in \G^{r}$,
\[
   \Bigl|\Lambda_{M,\mathbf b}(f_1,\dots,f_m)-\prod_{j=1}^{m}\mathbb E[f_j]\Bigr|\le 2^{m}\varepsilon^{1/2}.
\]
\end{corollary}

\begin{proof}
Write $f_j=\mu_j+g_j$ with $\mu_j=\mathbb E[f_j]$ and $g_j$ mean-zero,
$\|g_j\|_\infty\le1$, $\max_{\chi\ne\chi_0}|\widehat{g_j}(\chi)|\le\varepsilon$.
Expand $\Lambda_{M,\mathbf b}(f_1,\dots,f_m)$ over the $2^m$ subsets
$S\subseteq[m]$. The term $S=\varnothing$ is $\prod_j\mu_j$ (the constant
functions give the full average, and $|\{M\mathbf x=\mathbf b\}|=|\G|^{m-r}$ by
Lemma~\ref{lem:variety}). Each $S\ne\varnothing$ has a factor $g_{i}$
($i\in S$); setting the other factors to $1$ and applying
Lemma~\ref{lem:gvn} together with the moment bound
$\sum_{\chi}|\widehat{g_i}(\chi)|^{4}\le\varepsilon^{2}$ bounds that term by
$\varepsilon^{1/2}$. Summing the at most $2^m$ terms gives the claim.
\end{proof}

\subsubsection*{Covering, the quotient image, and the algorithm}

Let $H$ be the subgroup returned by \textsc{ConstDT}, $\bar\G=\G/H$, and let
$\bar M:\bar\G^{\,m}\to\bar\G^{\,r}$ be the induced map (well-defined as
$M(H^m)=H^r$, which holds because $d_r(M)$ is coprime to $\exp(H)\mid\lcmG$, where $\exp(H)$ denotes the torsion of $H$).
Output cosets are cosets of $H^{r}$ in $\G^{r}$, indexed by $\bar\G^{\,r}$. Let
$S\subseteq\bar\G$ be the regular, density $\ge\frac{\zeta}{2}$ leaves of \textsc{ConstDT}.

\begin{lemma}[Exact covering]
\label{lem:cover-matrix}
Suppose $M$ has complexity $1$ over $\mathbb Z_\lcmG$ and the regularity parameter
satisfies $\varepsilon<(\zeta/4)^{2m}$. Then for every $\bar{\mathbf c}\in S^{m}$,
the entire output coset $M({\mathbf c})+H^{r}$ is contained in
$M(A'^m)$; and $M(A'^m)$ meets no output coset outside $\bar M(S^m)$. Hence
\begin{equation}
\label{eq:vol-matrix}
   \Density\bigl(M(A'^m)\bigr)=\frac{\bigl|\bar M(S^m)\bigr|}{|\bar\G|^{\,r}} .
\end{equation}
\end{lemma}

\begin{proof}
Fix $\bar{\mathbf c}\in S^m$ and a target $\mathbf b$ in the coset
$ M({\mathbf c})+H^{r}$. Restricting the input coordinates to the dense regular cosets $c_j+H$ and pulling back to $H$ gives functions
$g_j=\mathbbm{1}_{A'}|_{c_j+H}$ on $H$ with $\mathbb E_H[g_j]\ge\frac{\zeta}{2}$, each
$\varepsilon$-regular; the restricted system $M\mathbf y=\mathbf b'$ in $H$ inherits
complexity $1$ over $\mathbb Z_{\exp(H)}$, where $\exp(H)$ denotes the torsion of $H$. By Corollary~\ref{cor:count-matrix}
applied over $H$, uniformly in $\mathbf b'$,

\[
\begin{aligned}
\Lambda_{M,\mathbf b'}(g_1,\dots,g_m)
&\ge
\prod_j\mathbb E_H[g_j]
-2^m\varepsilon^{1/2}\\
&\ge
(\zeta/2)^m-2^m\varepsilon^{1/2}
>0,
\end{aligned}
\]
since $\varepsilon<(\zeta/4)^{2m}$.
\color{black}
Thus $\mathbf b$ has a representation in
$A'^m$, proving full covering; and $M(A'^m)$ lands only in cosets $\bar M(\bar{\mathbf c})$
with $\bar{\mathbf c}\in S^m$. Counting cosets times $|H^{r}|$ and dividing by
$|\G|^{r}$ gives \eqref{eq:vol-matrix}.
\end{proof}

\begin{algorithm}[H]
\caption{\textsc{ApproxMatrixImage}$(A,M,\theta)$}
\label{alg:matrix-image}
\begin{algorithmic}[1]
\Require query access to $A\subseteq\G$; matrix $M\in\mathbb Z^{r\times m}$, whose number of pseudo--independent (Definition~\ref{defn_independence}) rows is at most $r$, with $m>r$, $\gcd(d_r(M),\lcmG)=1$, and complexity $1$ over
         $\mathbb Z_\lcmG$; accuracy $\theta\in(0,1)$.
\Ensure A value $v$ and a $(\theta,0,O(n))$-oracle
$O_{M(A'^m)}$ for the image set $M(A'^m)$.

\State Set $\zeta\gets\theta/2$ and
$\eta\gets\tfrac12(\zeta/4)^{2m}$.
\Comment{So $\eta<(\zeta/4)^{2m}$
and $\eta+\zeta\le\theta$}
\color{black}
       \Comment{so $2^m\eta^{1/2}<\zeta^m$ (exact covering) and $\eta+\zeta\le\theta$}
\State $(\mathcal O_{A'},T)\gets\textsc{ConstDT}(A,\eta,\zeta)$; let
       $S\subseteq\bar\G=\G/H$ be the surviving leaves of $T$.
\State Compute $\bar M(S^m)\subseteq\bar\G^{\,r}$ by evaluating $\bar M$ on $S^m$
       in the quotient.
\State Build $\mathcal O_{M(A'^m)}$: on input $\mathbf w\in \G^{r}$, read its coset
       $\bar{\mathbf w}\in\bar\G^{\,r}$ from $T$ and output $1$ iff
       $\bar{\mathbf w}\in\bar M(S^m)$.
\State \Return $v\gets|\bar M(S^m)|/|\bar\G|^{\,r}$ and $\mathcal O_{M(A'^m)}$.
\end{algorithmic}
\end{algorithm}

\begin{theorem}[Approximate image counting for matrix linear system, Theorem~\ref{thm:matrix-image}, restated]
Fix $M\in\mathbb Z^{r\times m}$, whose number of pseudo--independent (Definition~\ref{defn_independence}) rows is at most $r$, with $m>r$,
$\gcd(d_r(M),\lcmG)=1$, and complexity $1$ over $\mathbb Z_\lcmG$. For every
$\theta>0$ and every finite Abelian $\G$ with constant torsion $\lcmG$, given query
access to $A\subseteq\G$, Algorithm~\ref{alg:matrix-image} makes
\[
   \poly\Bigl(\Tower\bigl(\lcmG,\,O_m(\theta^{-O(m)})\bigr)\Bigr)
\]
queries and time, both independent of $n$ and $|\G|$, and with probability
at least $9/10$ outputs a value $v$ and an oracle $\mathcal O_{M(A'^m)}$ such that,
for a density-capturing perturbation $A'\subseteq A$ with
$\Density(A\setminus A')\le\theta$,
\[
   \bigl|\,v-\Density\bigl(M(A'^m)\bigr)\,\bigr|\le\theta,
\]
and $\mathcal O_{M(A'^m)}$ is an oracle for the image set $M(A'^m)\subseteq \G^{r}$.
\end{theorem}

\begin{proof}
Computing $\bar M(S^m)$ touches at most
\[
|S|^m\le|\bar\G|^m
\le\Bigl(\Tower(\lcmG,O(\eta^{-3}))\Bigr)^m
\]
tuples, independently of $|\G|$.
\color{black}
Theorem~\ref{th: Regularity for Abelian} (\textsc{ConstDT}) supplies $\mathcal O_{A'}$, $T$, and $S$ within the
stated complexity ($1/\eta^3=O_m(\theta^{-6m})$), with
$\Density(A\setminus A')\le\eta+\zeta\le\theta$ by Lemma~\ref{Prop: prop of Green's regularity}. The parameter
choice gives $\varepsilon=\eta<(\zeta/4)^{2m}$, so Lemma~\ref{lem:cover-matrix}
applies: $\Density(M(A'^m))=|\bar M(S^m)|/|\bar\G|^{\,r}=v$ \emph{exactly} (the only
losses are the perturbation $A\setminus A'$ and the $\le1/10$ failure of
\textsc{ConstDT}). Oracle correctness is the same coset description:
$\mathbf w\in M(A'^m)$ iff $\bar{\mathbf w}\in\bar M(S^m)$.
\end{proof}


Theorem~\ref{thm:matrix-image} is the matrix
analogue of Theorem~\ref{thm:lin-image} for systems
satisfying the complexity-$1$ hypothesis.
In particular, it applies to suitable single equations
with $m\ge3$. The case $r=1$, $m=2$, and
$a=(1,1)$ is handled separately by
Theorem~\ref{thm:lin-image}, since the corresponding
two-variable system does not satisfy the
complexity-$1$ condition.
\color{black}
A few points are worth noting. First, the covering here is \emph{exact} (Lemma~\ref{lem:cover-matrix}), with no analogue of the $1-\varepsilon^2/\zeta^{2m}$ loss in the single-form case; this is because Corollary~\ref{cor:count-matrix} holds uniformly over all targets $\mathbf b$, rather than only on average. Second, complexity $1$ is the natural limit of our method: for systems of higher complexity, small Fourier coefficients no longer control the number of solutions, as observed by Green~\cite{green2005szemeredi}, and the known removal and counting results for such systems go through hypergraph methods~\cite{Shapira2009, KSV2012}, which give no constructive guarantee. Third, the coprimality condition $\gcd(d_r(M),\lcmG)=1$ plays the same role for the matrix $M$ that pseudo-independence (Definition~\ref{defn_independence}) plays for a single form; without it, the solution count, and with it the main term of Corollary~\ref{cor:count-matrix}, is no longer correct. Finally, as in Section~\ref{ssec:lin-image}, it is the image of the perturbation $A'$ that we estimate, not the genuine image $M(A^m)$: the latter cannot be estimated to any nontrivial accuracy with a number of queries independent of $|\G|$, already at $r=1$, by the lower bound of~\cite{de2022approximating}.

\bigskip

\textbf{AI use statement.}
All mathematical ideas and proofs presented in this work are the authors' own. AI tools were used solely for proofreading and editing to improve the readability of the paper. The authors take full responsibility for the content of the submitted manuscript.

\bibliography{ref}
\appendix

\section{Preliminaries}
\label{sec_prelims}
Throughout this paper, $\G$ denotes a finite Abelian group $\mathbb{Z}_{p_1^{m_1}}\times \cdots \times \mathbb{Z}_{p_T^{m_T}}$, where $p_i$ are primes for all $i \in [T]$ and may not be distinct, and $\omega_\lcmG$ denotes a primitive $\lcmG^{th}$ root of unity, where $\lcmG = \LCM \{p_1^{m_1}, \ldots, p_T^{m_T}\}$.

\subsection{Basics on Fourier Analysis over \texorpdfstring{$\G$}{}}

Let us define the characters of $\G$ first.

\begin{definition}\label{defn_character_1_groups}
    \textbf{(Character)} For $\G = \mathbb{Z}_{p_1^{m_1}} \times \cdots \times \mathbb{Z}_{p_T^{m_T}}$, where $p_i, i \in [T]$ are primes, a character of the group $\G$ is a homomorphism $\chi: \G \to \mathbb{C}^\times$ of $\G$, that is, $\chi$ satisfies the following: for all $x,y \in \G$ we have $\chi(x+y)=\chi(x) \chi(y)$.

    \noindent Equivalently, a character of $\G$ is of the form $$\chi(x^{(1)},\dots, x^{(T)}) = \chi_{r^{(1)}}(x^{(1)}) \cdots \chi_{r^{(T)}}(x^{(T)}),$$ where $r^{(i)} \in \mathbb{Z}_{p_i^{m_i}}$ and $\chi_{r^{(i)}}$ is a character of $\mathbb{Z}_{p_i^{m_i}}$ and is defined by $$\chi_{r^{(i)}}(x^{(i)}) = \omega_{p_i^{m_i}}^{r^{(i)} x^{(i)}}.$$ 
    Thus for any $(r^{(1)}, \dots, r^{(T)})\in \G$ we can define a corresponding character of $\G$ as $\chi_{r^{(1)}, \dots, r^{(T)}}$ that on input  $x = (x^{(1)}, \dots, x^{(T)})$ is defined as
    \begin{align}\label{eq:character}
    \chi_{(r^{(1)}, \ldots, r^{(n)})}(x^{(1)}, \ldots, x^{(T)}) 
    &= \omega_{p_1^{m_1}}^{r^{(1)}\cdot x^{(1)}} \cdots \omega_{p_T^{m_T}}^{r^{(T)}\cdot x^{(T)}} \\
    &= \omega_\lcmG^{\sum_{i=1}^T r^{(i)}\cdot x^{(i)} \frac{\lcmG}{p_i^{m_i}} \pmod{\lcmG}},
\end{align}
where $\lcmG = \LCM \{ p_1^{m_1}, \ldots, p_T^{m_T} \}$.
\end{definition}

Now let us look at some properties of characters.

\begin{lemma}\label{lemma_character_groups}
Let $\chi$ be a character of $\G$. Then,
\begin{enumerate}
    \item $\chi_0(x)=1$ for all $x\in \G$. 
    \item $\chi(-x)= \chi(x)^{-1}=\overline{\chi(x)}$ for all $x \in \G$. 
    \item For any character $\chi$ of $\G$, where $\chi\neq \chi_0$, $\sum_{x\in \G} \chi(x)=0$.
    \item $|\chi(x)|=1$ for all $x \in \G$.
\end{enumerate}
\end{lemma}

Now let us define the dual group of $\G$.

\begin{definition}\label{defn_dual_group}
    \textbf{(Dual group)} The set of characters of $\G$ forms a group under the operation $(\chi\psi) (x) = \chi(x) \psi(x)$ and is denoted by $\widehat{\G}$, where $\chi$ and $\psi$ are characters of $\G$. $\widehat{\G}$ is called the dual group of $\G$.
\end{definition}

The following theorem states that $\G$ is isomorphic to its dual group.

\begin{theorem}
    $\widehat{\G} \cong \G$, that is $\G$ and $\widehat{\G}$ are isomorphic to each other.
\end{theorem}

Let us look at the definition of the Fourier transform for functions on $\G$.

\begin{definition} \label{def:foruier-transform_groups}
    \textbf{(Fourier transform)} For any $\G = \mathbb{Z}_{p_1^{m_1}} \times \cdots \times\mathbb{Z}_{p_T^{m_T}}$, with $p_i, i \in [T]$ primes, and any function $f: \G \to \mathbb{C}$, the Fourier transform $\widehat{f}: \widehat{\G}\to \mathbb{C}$ is 
    $$\widehat{f}(\chi_{r^{(1)},\ldots,r^{(T)}}) =\frac{1}{|\G|} \sum_{x\in \G} f(x) \omega_{p_1^{m_1}}^{-r^{(1)}\cdot x^{(1)}} \cdots \omega_{p_T^{m_T}}^{-r^{(T)}\cdot x^{(T)}},$$ where $x = (x^{(1)}, \dots, x^{(T)})$. 

\end{definition}

\begin{remark}
    The Fourier transform of a function $f:\G \to \mathbb{C}$ is defined by $$\widehat{f}(\chi)= \frac{1}{|\G|} \sum_{x\in \G} f(x) \overline{\chi(x)},$$ where $\overline{\chi(x)}$ is the conjugate of $\chi(x)$. The Definition~\ref{def:foruier-transform_groups} follows from this, as $\chi=\chi_{r^{(1)},\ldots,r^{(T)}}$ for some $r^{(i)} \in \mathbb{Z}_{p_i^{m_i}}, \ i\in \{1,\ldots,T\}$.
\end{remark}

The following theorem states that any function from $\G$ to $\mathbb{C}$ can be written as a linear combination of characters of $\G$.

\begin{theorem}\label{Fourier_inversion_groups}
    \textbf{(Fourier inversion formula)} Any function $f: \G \to \mathbb{C}$ can be uniquely written as a linear combination of characters of $\G$, that is,
    \begin{equation} \label{eq:fourierex}
    f(x)= \sum_{\chi_{r^{(1)},\ldots,r^{(T)}} \in \widehat{\G}} \widehat{f}(\chi_{r^{(1)},\ldots,r^{(T)}}) \chi_{r^{(1)}, \dots, r^{(T)}}(x),\end{equation} 
    where $x = (x^{(1)}, \dots, x^{(T)})$. 
\end{theorem}

\begin{theorem}\label{theorem_Parseval_groups}
    \textbf{(Parseval)} For any two functions $f,g: \G \to \mathbb{C}$, $$\mathbb{E}_{x\in \G} [f(x)\overline{g(x)}] =\sum_{\chi\in \widehat{\G}} \widehat{f}(\chi) \overline{\widehat{g}(\chi)}.$$
    More specifically, if $f:\G \to \{0,1\}$ is a Boolean-valued function then
    $$
        \sum_{\chi \in \widehat{\G}} |\widehat{f}(\chi)|^2 \leq 1.
    $$
\end{theorem}

Now let us define the Fourier sparsity $s_f$ of a function $f$ on $\G$. 

\begin{definition}
\label{defi: Fourier spectral norm}
\textbf{(Fourier Spectral Norm)}
    The Fourier Spectral Norm or the Spectral norm of $f : \G \to \mathbb{C}$, denoted by $\|f\|_1$ is defined as the sum of the absolute value of its Fourier coefficients. That is,
    \begin{align*}
        \left\| \widehat{f} \right\|_1 = \sum_{r \in \G} \left| \widehat{f}(\chi_r) \right|.
    \end{align*}
\end{definition}

\begin{definition}
\textbf{(Sparsity and Fourier Support)} 
\begin{itemize}
\item 
        Fourier support $\supp(\widehat{f})$ of a function $f: \G \to \mathbb{C}$ denotes the set $\left\{ \chi \,\mid\, \widehat{f}(\chi) \neq 0\right\}$. 
    \item
        The Fourier sparsity $s_f$ of a function $f: \G \to \mathbb{C}$ is defined to be the number of non-zero Fourier coefficients in the Fourier expansion of $f$ (Theorem~\ref{Fourier_inversion_groups}).  Alternately, Fourier sparsity can be defined as the size of the Fourier support. In this paper, by sparsity of a function, we mean the Fourier sparsity of the function. Moreover, by $s$-sparse function we mean functions with Fourier Sparsity $s$. 
\end{itemize}
\end{definition}

\subsection{Density of a set}
\begin{definition}[Indicator function]\label{defn_indicator}
    The indicator function on a set $A\subseteq \G$ is defined by the following:
    \begin{align*}
        \mathbbm{1}_A(x) = 
        \begin{cases}
            1, & x\in A \\
            0, &\text{otherwise.}
        \end{cases}
    \end{align*}
\end{definition}

Let us now define the density of a set.

\begin{definition}\label{defn_density}
    Let $A$ be a subset of $\G$. Then the density of $A$ in $\G$ is defined in the following way:
    $$\Density (A) = \frac{|A|}{|\G|},$$
    where $|.|$ denotes the size of the set.
\end{definition}

\begin{definition}\label{defn_density_H}
Given arbitrary subsets $A, B \subseteq \G$ of some abelian group $\G$, we denote the following:
\begin{align}
    \Density(A) := \frac{|A|}{|\G|}, 
\end{align}
 and 
 \begin{align}
     \Density_B(A) := \frac{|A \cap B|}{|B|}.
 \end{align}
 \end{definition}


\subsection{Regularity}
Let us fix some notations first. For any arbitrary subset $A \subseteq \G$ and a subgroup $H \leq \G$,  
\begin{definition}[$\epsilon$-regular]\label{defn_regular}
    A function $f: \G \to \mathbb{C}$ is $\epsilon$-regular if 
    \begin{align*}
        \textit{ sup}_{\chi \neq \chi_0} | \widehat{f}(\chi)| \leq \epsilon.
    \end{align*}
    Let $H$ be a normal subgroup of $\G$ and $H'$ be a coset of $H$. A function $f$ restricted to $H'$, $f_{H'}: H' \to \mathbb{C}$ is $\epsilon$-regular if 
     \begin{align*}
        \textit{ sup}_{\chi_0 \neq \chi \in H}| \widehat{f_{H'}}(\chi)| \leq \epsilon,
    \end{align*}
    where $\widehat{f_{H'}}(\chi)$ is defined in Definition~\ref{defn_restricted_function}. 
\end{definition}





\subsection{The Pseudo Inner Product Framework}
We know that since $\mathbb{Z}_2^n$ is also a vector space over the field $\ftwo$, it has a linear algebraic structure. The characters are in one-to-one correspondence with $\ftwo$-linear forms. However, since a general Abelian group is not a vector space, these do not hold here. But one can define something along similar lines using the `pseudo inner product framework'. We borrow this framework from Chakraborty et al.~\cite{chakraborty2024fourier} and Datta et al.~\cite{DBLP:conf/approx/DattaGKPR25}.

An element $r$ of $\G$ is 
of the form $(r^{(1)}, \dots, r^{(T)})$, where $r^{(i)}$ is an element of $\mathbb{Z}_{p_i^{m_i}}$.

\begin{definition}
\label{defi:pseudo-inner product}
    \textbf{(Pseudo inner product)} For $x, r \in \G$ we denote by $*$ the following pseudo inner product. $$r * x := \biggl( \sum_{i=1}^{T} \frac{\lcmG}{p_i^{m_i}} r^{(i)} \cdot x^{(i)} \biggr) \pmod{\lcmG},$$ where $\lcmG = \LCM \{ p_1^{m_1}, \ldots, p_T^{m_T} \}$.
\end{definition}

\begin{observation}\label{obs_chi_r_x}
    $r* x \neq 0 \Rightarrow \chi_r(x)\neq 1$. 
\end{observation}

\begin{definition}\label{defn_independence}
    Let $r_1, \ldots, r_k$ be elements in $\G$. We call $r_1, \ldots, r_k$ to be \textbf{pseudo--dependent} if there exists $\lambda_1,\ldots, \lambda_k \in \mathbb{Z}_\lcmG$ with at least one $\lambda_j$ invertible in $\mathbb{Z}_\lcmG$ such that $\sum_{i=1}^k \lambda_i r_i = 0$.

    If there does not exist any such $\lambda_1,\ldots, \lambda_k \in \mathbb{Z}_\lcmG$ with at least one $\lambda_j$ invertible in $\mathbb{Z}_\lcmG$ such that $\sum_{i=1}^k \lambda_i r_i = 0$, then we call $r_1, \ldots, r_k$ to be \textbf{pseudo--independent}. 
\end{definition}

For example, consider the group $\mathbb{Z}_4$. The element $3$ is dependent on the element $1$, but it is independent of the element $2$. Whereas, $2$ is dependent on $3$, as $2 = 2 \cdot 3 (\mod 4)$. So the set $\{2, 3\}$ is dependent.

\begin{definition}[Collapsing or identification of characters]
\label{defn_collapsing_groups}
Let $f:\G \to \pmo$ be a Boolean function. For any two characters $\chi_\alpha, \chi_\beta$ of $\G$ in $\supp(\widehat{f})$, we say $\chi_\alpha$ collapses or identifies with $\chi_\beta$ if $\alpha,\beta$ belong to the same coset $C$ in $\mathcal{C}_{r_1, \ldots, r_k}(f)$, (that is, $(\alpha-\beta)$ belongs to $S_{\{r_1,\ldots,r_k\}}$), where $\mathcal{C}_{r_1, \ldots, r_k}(f)$ is as in Definition~\ref{defn_bucket_complexity_groups}.
\end{definition}

\begin{definition}\label{defn_bucket_complexity_groups}
\textbf{(Bucket complexity)} Let $f:\G \to \{\-1,+1\}$ be a function and let $\Gamma=\{r_1,\ldots,r_t\}$ be a set of elements in $\G$.
Let $\mathcal{C}_{r_1,\cdots,r_t}$ be the set of cosets of $S_{\Gamma}$, that is necessary to partition the Fourier support of $f$ with respect to $S_\Gamma$. That is, 
$$\mathcal{C}_{r_1, \ldots, r_t}(f) = \{r+ S_\Gamma | (r+S_\Gamma) \cap \supp(\widehat{f}) \neq \emptyset, r \in \G\}.$$
The cardinality of the set $\mathcal{C}_{r_1, \ldots, r_t}(f)$ is called the bucket complexity of $f$ with respect to $\Gamma$ and is denoted by $\mathfrak{b}_{\Gamma}(f)$. 
\end{definition}

Now, we partition $\G$ with the help of a normal subgroup.

\begin{definition}\label{defn_normal_subgroup}
\label{defn_cosets_groups}
    Let $r_1,\ldots,r_k \in \G$ such that $r_j= (r_j^{(1)},\ldots, r_j^{(T)})$, where $r_{j}^{(i)} \in \mathbb{Z}_{p_i^{m_i}}$ for all $i \in [T]$ and $j \in [k]$. Let $b= (b_1, \ldots, b_k) \in \mathbb{Z}_\lcmG^k$, where $\lcmG = \LCM\{ p_1^{m_1}, \ldots, p_T^{m_T}\}$. 
    We define the set $V_{b,r_1,\ldots,r_k}$ by
    \begin{align}\label{eqn_cosets_groups}
        V_{b,r_1,\ldots,r_k} = \{ x\in \G : r_j * x = b_j \pmod{\lcmG} \ \forall j \in [k]\}.
    \end{align}
     That is, 
    \begin{align*}
        V_{b,r_1,\ldots,r_k} 
        &= \biggl\{x \in \G : \biggl( \sum_{i=1}^{T} \frac{\lcmG}{p_i^{m_i}} r_j^{(i)} \cdot x^{(i)} \biggr) = b_j \pmod{\lcmG} \ \forall j \in [k] \biggr\}.
    \end{align*}

    If $b_1 = b_2 = \dots = 0$ then we use $V_{0,r_1,\ldots,r_k}$ to denote $V_{b,r_1,\ldots,r_k}$.
\end{definition}

\begin{remark}\label{remark_cosets_groups}
    When $k=1$, then $b \in \mathbb{Z}_\lcmG$, and we have only one element $r \in \G$, so this set becomes $V_{b,r} = \{ r * x =b\pmod{\lcmG}\}$.
\end{remark}

\begin{lemma}\label{lemma_cosets_of_V_0}
    $V_{0,r_1,\ldots,r_k}$ is a normal subgroup of $\G$ and for any $b, r_1, \dots, r_k$ either $V_{b,r_1,\ldots,r_k}$ is a coset of $V_{0,r_1,\ldots,r_k}$ or $V_{b,r_1,\ldots,r_k} = \emptyset$.
\end{lemma}

\begin{proof}
    Clearly, $0\in V_{0,r_1,\ldots,r_k}$, where $0$ is the identity element of $\G$. Let $x,y \in V_{0,r_1,\ldots,r_k}$. Then, since $r_j * (x-y) = \sum_{i=1}^{T} \frac{\lcmG}{p_i^{m_i}} r_j^{(i)} \cdot x^{(i)} - \sum_{i=1}^{T} \frac{\lcmG}{p_i^{m_i}} r_j^{(i)} \cdot y^{(i)} =0 \pmod{\lcmG}$ for all $j\in [k]$, so $x-y \in V_{0,r_1,\ldots,r_k}$. So $V_{0,r_1,\ldots,r_k}$ is a subgroup of $\G$, and hence a normal subgroup of $\G$ since $\G$ is Abelian.

    Since $V_{0,r_1,\ldots,r_k}$ is a normal subgroup of $\G$ we can now consider its cosets. 
    For each coset $B$ of $V_{0,r_1, \ldots,r_k}$, let $a(B)$ be a fixed coset representative (we choose a coset representative and fix it all through the proof). Then, if $y \in B$, then $y = a(B)+x$, where $x \in V_{0,r_1,\ldots,r_k}$. Now, for each $r_j$, 
    $$r_j * y = r_j * (a(B)+x) = r_j * a(B),$$ 
    since $r_j * x =0$ for all $r_j, \ j \in [k]$. 
    So, $r_j * y$ is fixed for each coset $B$. If we assume $r_j * a(B) = b_j \pmod{\lcmG}$ for each $j \in [k]$, then we have $V_{b,r_1,\ldots,r_k}$ as a coset of $V_{0,r_1,\ldots,r_k}$, where $b=(b_1, \ldots ,b_k)$.

    If for a $b = (b_1, \ldots,b_k) \in \mathbb{Z}_\lcmG^k$, there does not exist any $a(B)$ such that $r_j * a(B) = b_j \ \forall j \in [k]$, $V_{b, r_1, \ldots, r_k} = \emptyset$. 
\end{proof}

\begin{remark}\label{remark_subgroup_implies_normal}
    Note that since $\G$ is an Abelian group, any subgroup of $\G$ is a normal subgroup of $\G$.
\end{remark}

\begin{lemma}\label{lemma_restriction}
Let us consider $r_1,\ldots,r_t \in \G$. Then, for any $b=(b_1,\ldots,b_t) \in \mathbb{Z}_\lcmG^t$, if we fix $r_j*x = b_j \pmod{\lcmG}$ for each $j \in [t]$, we get the following:
\begin{enumerate}
\item The characters $r(C)+r_1,\ldots, r(C)+r_t$ collapse (see Definition~\ref{defn_collapsing_groups}) to the same coset $\mathcal{C}_{r_1,\ldots,r_t}(f)$ of $S_{\{r_1,\ldots,r_t\}}$, where $r(C)$ is the coset representative of the coset $\mathcal{C}_{r_1,\ldots,r_t}(f)$. 

\item This fixing gives us the subfunctions $f|_{V_{b,r_1,\ldots,r_t}}$, where $V_{b,r_1,\ldots,r_t}$ is as defined in Definition~\ref{defn_cosets_groups}, whose characters are $\chi_{r(C)}$, and the coefficients are given by
    \begin{align*}
        \widehat{f|_{V_{b,r_1,\ldots,r_t}}} (\chi_{r(C)}) &= P_C(b_1, \ldots, b_t)
        = \sum_{r\in S_\Gamma} \widehat{f}(r(C)+r) \chi_r(x).
    \end{align*}
\end{enumerate}
\end{lemma}

\begin{proof}[Proof of Lemma~\ref{lemma_restriction}]

    Let $\mathcal{C}_{r_1,\ldots,r_t}(f)$ be the set defined in Definition~\ref{defn_bucket_complexity_groups}.
    We can write the Fourier expansion of $f$ in the following way.

\begin{align}\label{eqn_fixing}
f(x)
&=\sum_{r'\in \G}\widehat{f}(\chi_{r'})\chi_{r'}(x)\nonumber\\
&=\sum_{C\in \mathcal{C}_{r_1,\ldots,r_t}}\ \sum_{r'\in C}\widehat{f}(\chi_{r'})\chi_{r'}(x)\nonumber\\
&=\sum_{C\in \mathcal{C}_{r_1,\ldots,r_t}}\ \sum_{r\in S_{\{r_1,\ldots,r_t\}}}\widehat{f}(\chi_{r(C)+r})\chi_{r(C)+r}(x)\nonumber\\
&=\sum_{C\in \mathcal{C}_{r_1,\ldots,r_t}}\biggl(\sum_{r\in S_{\{r_1,\ldots,r_t\}}}\widehat{f}(\chi_{r(C)+r})\chi_{r}(x)\biggr)\chi_{r(C)}(x) \\
&= \sum_{C \in \mathcal{C}_{r_1,\ldots,r_t}(f)} \bigg ( P_C(\chi_{r_1}(x), \ldots, \chi_{r_t}(x)) \bigg ) \chi_{r(C)}(x),
\end{align}

    where $r(C)$ is the representative of the coset $C\in \mathcal{C}_{r_1,\ldots,r_t}(f)$ and $$P_C(\chi_{r_1}(x), \ldots, \chi_{r_t}(x))=\sum_{r\in S_\Gamma} \widehat{f}(\chi_{r(C)+r}) \chi_r(x).$$ If we fix $r_j * x = b_j \pmod{\lcmG}$ for all $j\in [t]$, where $x \in \G$, then $f$ becomes $f|_{V_{b,r_1,\ldots,r_t}}$, which follows from the definition of $V_{b,r_1,\ldots,r_t}$ and $f|_{V_{b,r_1,\ldots,r_t}}$. Its characters are $\chi_{r(C)}, \ C \in \mathcal{C}_{r_1,\ldots,r_t}(f)$, and the coefficients are given by
    \begin{align*}
        \widehat{f|_{V_{b,r_1,\ldots,r_t}}} (\chi_{r(C)}) &= P_C(b_1, \ldots, b_t)
        = \sum_{r\in S_\Gamma} \widehat{f}(\chi_{r(C)+r}) \chi_r(x).
    \end{align*}
\end{proof}

\begin{definition}\label{defn_codimension}
    \textbf{(Codimension)} Let $V_{b,r_1,\ldots,r_k}$ be as defined in Definition~\ref{defn_cosets_groups}. Then the codimension of $V_{b,r_1,\ldots,r_k}$ is given by 
    \begin{align*}
        \Codim(&V_{b,r_1,\ldots,r_k}) =  \min_{k'} \biggl\{k: \exists r_1',\ldots, r_{k'}' \in \{r_1, \ldots, r_k\} \text{ such that } \\
        &V_{b,r_1,\ldots,r_k} = \{x\in \mathcal{G}: r_j' * x =b_j \pmod{\lcmG} \ \forall j \in [k]\} \biggr\},
    \end{align*}
    where $b =(b_1, \ldots, b_k) \in \mathbb{Z}_\lcmG^k$, and $\lcmG = \LCM\{p_1^{m_1}, \ldots, p_T^{m_T}\}$. 
\end{definition}

\begin{definition}\label{defn_random_coset_structure}
    \textbf{(Random coset structure)} Let us consider $\beta_1, \ldots, \beta_k \in \mathcal{G}$, which are chosen independently and uniformly at random from $\mathcal{G}$. Also let $H= \{\alpha \in \mathcal{G} :  \alpha * \beta_i = 0 \ \forall i \in [k]\}$. So $H$ is a subgroup, hence a normal subgroup of $\mathcal{G}$. Then, we define the sets $C(b)$ by $C(b)= \{\alpha \in \mathcal{G} : \alpha * \beta_i = b_i \ \forall i \in [k]\}$, for all $b=(b_1,\ldots, b_k) \in \mathbb{Z}_{\lcmG}^k$. 
\end{definition}

\begin{remark}\label{remark_random_coset_structure}
    Observe that the nonempty $C(b)$ are cosets of $H$. We will often refer to them as buckets.
\end{remark}

\subsection{The subgroup \texorpdfstring{$H^\perp$}{Hperp}}\label{section_H_perp}

Throughout this section, $\G$ denotes the finite Abelian group $\mathbb{Z}_{p_1^{m_1}}\times \cdots \times \mathbb{Z}_{p_T^{m_T}}$, where $p_i$ are primes for all $i \in [T]$ and not necessarily distinct, $\langle x \rangle$ denotes the subgroup generated by $x$, and $\omega_\lcmG$ denotes a primitive $\lcmG^{th}$ root of unity and $\lcmG = \LCM\{p_1^{m_1},\cdots,p_T^{m_T}\}$. Let $H$ be a subgroup of $\G$.

Let us look at the definition of $H^\perp$.

\begin{definition}\label{defn_1_new}
    \textbf{(Definition 1)} Let $H$ be a subgroup of $\G$. Then $H^\perp$ is the subgroup given by
    \begin{align*}
        H^\perp := \{ x \in \G : x * y =0 \ \forall y \in H \},
    \end{align*}
    where $$x * y := \biggl( \sum_{i=1}^{T} \frac{\lcmG}{p_i^{m_i}} x^{(i)} \cdot y^{(i)} \biggr) \pmod{\lcmG},$$ and $\lcmG = \LCM \{p_1^{m_1} \cdots p_n^{m_n}\}$.
\end{definition}

\begin{remark}
    Observe that $H^\perp$ is also a subgroup of $\G$.
\end{remark}

\begin{claim}\label{claim_H_perp_new}
    For $z \in H^\perp$, 
    \begin{align*}
        \sum_{\beta \in H} \omega_\lcmG^{\beta * z} =|H|,
    \end{align*}
    where $\omega_\lcmG$ is a primitive $\lcmG^{th}$ root of unity of order $\lcmG$, and $|H|$ denotes the order of the subgroup $H$. 
    
    Also, for $z \notin H^\perp$, 
    \begin{align*}
        \sum_{\beta \in H} \omega_\lcmG^{\beta * z} =0.
    \end{align*}
\end{claim}

\begin{proof}
    \begin{description}
        \item [Case 1: When $z\in H^\perp$.] Then $z * x =0$ for all $x \in H$. Hence 
    \begin{align*}
        \sum_{\beta \in H} \omega_\lcmG^{\beta * z} =|H|.
    \end{align*}

        \item [Case 2: When $z\notin H^\perp$.] Let $\sum_{\beta \in H} \chi_\beta (z) =A$. Since $z \notin H^\perp$, so, by Definition~\ref{defn_1_new}, there exists $\gamma \in H$ such that $\gamma * z \neq 0$, that is, $\chi_\gamma(z) \neq 1$. (see Observation~\ref{obs_chi_r_x})
        Then, 
        \begin{align*}
            \chi_\gamma(z) \times A &= \chi_\gamma(z) \sum_{\beta \in H} \chi_\beta(z) \\
            &= \sum_{\beta \in H} \chi_{\beta+\gamma}(z) \\
            &= \sum_{\gamma' \in H} \chi_{\gamma'}(z), &\gamma' = \beta+\gamma \\
            &= A,
        \end{align*}
        which implies that 
        \begin{align*}
            A (\chi_\gamma(z)-1) =0 \Rightarrow A=0,
        \end{align*}
        since $\chi_\gamma(z) \neq 1$.
    \end{description}
\end{proof}

We need the following isomorphism result:

\begin{lemma}\label{lemma_iso_H_perp_G/H}
    $H^\perp$ is isomorphic to the quotient group $\G/H$.
\end{lemma}

\color{black}
\begin{proof}
    We will show that $\widehat{H^\perp}$ is isomorphic to $\widehat{\G/H}$, which implies that $H^\perp$ is isomorphic to the quotient group $\G/H$, as $G \equiv \widehat{G}$ for any group $G$.

    The set of characters of $H^\perp$ is given by $$\Ann_{\G}(H) = \{\chi \in \widehat{\G} : \chi(y) = 1 \ \forall y \in H\},$$ where $\Ann_{\G}(H)$ is known as the annihilator of the subgroup $H$ in $\G$. From Definition~\ref{defn_1_new}, observe that for $x \in \G$, $\chi_x(y) = \omega_\lcmG^{x*y} = 1$  if and only if $x* y =0 \Leftrightarrow x \in H^\perp$.

    Let us define a group homomorphism $\mathcal{F} : \widehat{\G/H} \to \Ann_{\G}(H)$ by $\mathcal{F}(\zeta) = \zeta \circ q$, where $q: \G \to \G/H$ is the quotient is the quotient group homomorphism defined by $q(r)=r+H$. 
    \begin{itemize}
        \item \textbf{($\mathcal{F}$ is injective)} Let $\zeta \in \widehat{\G/H}$ such that $\zeta \circ q = \Tilde{0}$, where $\Tilde{0}=(0, \ldots, 0) [T \text{ times}]$ is the identity element of $\G$. So, $\zeta(r + H) = \zeta \circ q(r) = 1$ for all $r \in \G$, which implies $\zeta$ is the identity element of $\G/H$. Therefore, $\mathcal{F}$ is injective.

        \item \textbf{($\mathcal{F}$ is surjective)} Let $\psi \in \Ann_{\G}(H)$. Let $\zeta : \G/H \to \mathbb{C}$ by $\zeta(r+H) = \psi(r)$ for all $r \in \G$. Since $\chi_{r+H}(x) = \omega_\lcmG^{r *x + H*x}$, so any character $\chi_{r+H}$ of $\G$ is a character of $\G/H$ if $H * x=0$ for all $x\in \G$ (as then, the value of $\omega_\lcmG^{r *x + H*x}$ will be determined only by the coset representatives). Since $\psi(r+H) = \psi(r)\psi(H) = \psi(r)$ and $H$ is the identity element of $G/H$, so $\zeta$ is a character of $\G/H$. Also, $\psi = \zeta \circ q$. Therefore, $\mathcal{F}(\zeta) = \zeta \circ q =\psi$. Hence, $\mathcal{F}$ is surjective.
    \end{itemize}

    Therefore, $\mathcal{F}$ is an isomorphism, which implies that $H^\perp$ is isomorphic to the quotient group $\G/H$.
\end{proof}

\begin{corollary}\label{lemma_H_perp_new}
    $|H| \times |H^\perp| = |\G|$.
\end{corollary}

\begin{proof}
    Follows from the fact that $|H^\perp|= \frac{|\G|}{|H|}$, since $H^\perp$ is isomorphic to the quotient group $\G/H$ by Lemma~\ref{lemma_iso_H_perp_G/H}, and $|\G/H| = \frac{|\G|}{|H|}$ by Lagrange's theorem.
\end{proof}

\subsection{Character Decision Trees}\label{sec_character_decision_tree}

For a general finite Abelian group $\G$, we use character decision trees instead of parity decision trees.

\begin{definition}\label{defn_char_tree}
    \textbf{(Character decision tree)} Let $\G$ be an Abelian group. A character decision tree is a rooted tree $\mathcal{T}$. Each internal node $v$ of $\mathcal{T}$ is labelled by a non-trivial character $r\in \G$. $v$ has exactly as many children as the size of the range of $r$, and each edge connecting $v$ to a child of $v$ is labelled by a distinct element of the range of $r$. In particular, each internal node has at most $|\G|$ children. Each leaf is labelled by $-1$ or $+1$.

    Given an input $x \in \G$, $\mathcal{T}$ can be simulated on $x$ as follows. Start at the root. At each step, if the current node $v$ is a leaf, output its label. Else, evaluate the character $r$ that labels $v$ on $x$. We often refer to doing this as querying $r$ on $x$, and obtaining an answer to the query. Then, navigate to a child of $v$ by traversing the edge connecting $v$ to a child of $v$ that is labelled by $r(x)$, and continue. Let $\mathcal{T}(x)$ denote the value that the above simulation outputs. $\mathcal{T}$ is said to compute a function $f$ if for each $x \in \G$, $\mathcal{T}(x)=f(x)$.

    $\mathcal{T}$ is said to be non-adaptive if all its leaves have the same depth, and all its internal nodes at the same depth are labelled by the same character. In other words, the character queried at a node is independent of the answers to the previous queries. Alternatively, a non-adaptive decision tree of depth $k$ may be thought of as querying $k$ characters simultaneously, and producing a value in $\{-1, +1\}$ as  an output based on the answers to the queries. A character decision tree that is not non-adaptive is called an adaptive character decision tree.
\end{definition}

Since any Boolean valued function $f:\G \to\pmo$ can be written as a linear combination of the characters of $\G$, one can evaluate the function $f$ with the help of its character decision tree. Let $x$ be an element in $\G$. To compute the function $f$, we start from the root (marked with a character $r$) of the character decision tree and compute $r * x$ in the first step, where $r * x = \biggl( \sum_{i=1}^{T} \frac{\lcmG}{p_i^{n_i}} r^{(i)} \cdot x^{(i)} \biggr) \pmod{\lcmG},$ and $\lcmG = \LCM \{ p_1^{n_1}, \ldots, p_T^{n_T} \}$. Depending on the value, we move to an outgoing vertex in the next step and compute $r' * x$, where $r'$ is the character that determines that specific node. We continue the process until we reach a leaf. The value at that leaf determines the value of $f$ at $x$.

\begin{definition}\label{def:decisiontreecomplexity}
    \textbf{(Character decision tree complexity and non-adaptive character decision tree complexity)} The character decision tree complexity of a Boolean-valued function $f:\G \to \pmo$ is defined as the smallest $d$ such that there is a  character decision tree of depth $d$ that computes $f$. The non-adaptive character decision tree complexity of $f$ is defined as the smallest $d$ such that there is a  non-adaptive character decision tree of depth $d$ that computes $f$.
\end{definition}

\begin{definition}\label{defn_char_decision_size}
    Let $f:\G \to \{-1,+1\}$ be a Boolean valued function and $\mathcal{T}$ be its character decision tree. Then the size of $\mathcal{T}$, denoted by $\size_{\oplus \lcmG} f$, is the total number of leaves in the tree.
\end{definition}

\subsection{Concentration inequalities}

\begin{lemma}[Hoeffding's Inequality] \label{Hoeffding}
Let $X_1, \cdots, X_k$ be real independent random variables, each taking value in [-1, 1]. Then 
\[\Pr \left[\left|\sum_{i=1}^k X_i - \mathbb{E}\sum_{i=1}^k X_i\right| \geq \epsilon\right] \le 2\exp\left(-\frac{\epsilon^2}{2k}\right). 
\]
\end{lemma}

\subsection{\texorpdfstring{$\wt_2$}{wt2} estimation}\label{sec_wt_2_estimation}

In this section, we estimate $\wt_2$ of a bucket (see Definition~\ref{defn_random_coset_structure} and Remark~\ref{remark_random_coset_structure}).

\begin{definition}\label{defn_wt_2}
    Let $f:\G \to \{0,1\}$ be a Boolean valued function. Then $\wt_2$ of a bucket $C$ is defined in the following way.
    \begin{align*}
        \wt_2(C) = \sum_{\beta\in \mathcal{C}} |\widehat{f}(\chi_\beta)|^2.
    \end{align*}
\end{definition}

Now let us define the projection operator 
${P}_{r+H}$, $r \in \G$ on a coset $r+H$ of $H$ by
\begin{align*}
    \widehat{P_{r+H}}(f) =
    \begin{cases}
        \widehat{f}(\chi_\gamma) & \gamma \in r+H \\
        0 & \text{otherwise.}
    \end{cases}
\end{align*}

\begin{claim}\label{claim_projection_expectation_new}
\begin{align*}
     P_{r+H} f(x) = \mathbb{E}_{z \in H^\perp} [f(x-z)\chi_r(z)].
\end{align*}
\end{claim}

\begin{proof}
\begin{align*}
    {P}_{r+H} f(x)
    &= \sum_{\gamma \in r+H} \hat{f}(\gamma) \chi_\gamma(x) & \\
    &= \sum_{\beta \in H} \hat{f}(r+\beta) \chi_{r+\beta} (x) &\text{ where } \gamma = r+\beta, \beta \in H \\
    &= \sum_{\beta \in H} \hat{f}(r+\beta) \chi_r(x) \chi_\beta(x) \\
    &= \chi_r(x) \sum_{\beta \in H} \hat{f}(r+\beta) \chi_\beta(x) \\
    &= \chi_r(x) \sum_{\beta\in H} \chi_\beta (x) \frac{1}{|\G|} \sum_{y\in \G} f(y) \overline{\chi_{r+\beta} (y)} \\
    &= \frac{1}{|\G|} \chi_r(x) \sum_{y\in \G} f(y) \sum_{\beta \in H} \chi_\beta(x) \chi_{-r-\beta} (y) \\
    &= \frac{1}{|\G|} \chi_r(x) \sum_{y\in \G} f(y) \sum_{\beta \in H} \chi_\beta(x) \chi_{r+\beta} (-y) \\
    &= \frac{1}{|\G|} \chi_r(x) \sum_{y\in \G} f(y) \chi_r(-y) \sum_{\beta \in H} \chi_\beta (x-y) \\
    &= \frac{|H|}{|\G|} \sum_{y \in H^\perp} f(y) \chi_r (x-y) &\text{by Claim} ~\ref{claim_H_perp_new} \\
    &= \frac{1}{|H^\perp|} \sum_{z \in H^\perp} f(x-z) \chi_r (z) \\
    &= \mathbb{E}_{z \in H^\perp} [f(x-z)\chi_r(z)],
\end{align*}
where $z = x-y$.
\end{proof}

Now let us estimate the weights of the buckets at the $k^{th}$ step.

\begin{lemma}\label{lemma_wt_of_bucket_new}
    The weight $\wt_2$ of a bucket $\mathcal{C}_{\alpha}$, where $\alpha$ is the dominating character of the bucket, and $\mathcal{C}_{\alpha}$ corresponds to the coset $r+H$, is given by
    \begin{align*}
        \wt_2 \bigl(\mathcal{C}_{\alpha}\bigr) = \mathbb{E}_{x \in \G, z \in H^\perp} \biggl[ f(x) f(x+z) \chi_r(z) \biggr].
    \end{align*}
\end{lemma}

\begin{proof}
    \begin{align*}\wt_2 \bigl(\mathcal{C}_{\alpha}\bigr) 
        &= \sum_{\gamma\in r+H} |\widehat{f}(\chi_\gamma)|^2 \\
        &= \sum_{\gamma\in \G} |\widehat{{P}_{r+H}}(\chi_\gamma)|^2 \\
        &= \mathbb{E}_{x' \in \G} \biggl[ {P}_{r+H}f(x') \overline{{P}_{r+H}f(x')} \biggr] &\text{by Parseval} \\
        &= \mathbb{E}_{x' \in \G, z_1,z_2 \in H^\perp} [ \chi_r(z_1)f(x'-z_1) \overline{\chi_r(z_2)} f(x'-z_2)] &\text{by Claim} ~\ref{claim_projection_expectation_new} \\
        &= \mathbb{E}_{x' \in \G, z_1,z_2 \in H^\perp} \biggl[ f(x'-z_1) f(x'-z_2) \chi_r(z_1-z_2) \biggr] \\
        &= \mathbb{E}_{x \in \G, z \in H^\perp} \biggl[ f(x) f(x+z) \chi_r(z) \biggr],
    \end{align*}
    taking $z=z_1-z_2$ and $x=x'-z_1$.
\end{proof}

\begin{lemma}\label{lemma_wt_estimate_new}
    Given query access to a Boolean valued function $f:\mathcal{G} \to \{0,1\}$, one can estimate the weight $\wt_2$ of each bucket $\mathcal{C}_{\alpha}$ with accuracy $\pm \varepsilon$ except with probability $\delta'$ using $O(4 \ln(\frac{4}{\delta}) /\varepsilon^2)$ many samples.
\end{lemma}

\begin{proof}
    By Lemma~\ref{lemma_wt_of_bucket_new}, 
    \begin{align*}
        \wt_2 \bigl(\mathcal{C}_{\alpha}\bigr) = \mathbb{E}_{x \in \G, z \in H^\perp} \biggl[ f(x) f(x+z) \chi_r(z) \biggr],
    \end{align*}
    where $r+H$ is the coset corresponding to the bucket $\mathcal{C}_{\alpha}$.
    
    Here, $f(x), \ f(x+z)$ are $0,1$ valued and $\chi_r(z)$ is a complex number with absolute value equal to $1$. So, $\Re(\chi_r(z)), \ \Im(\chi_r(z))$ both are in $[-1,+1]$, where $\Re(\chi_r(z))$ and $\Im(\chi_r(z))$ are the real and imaginary parts of $\chi_r(z)$ respectively. Let $E_i= f(x) f(x+z) \chi_r(z) = \Re(E_i) + \Im(E_i)$, where $\Re(E_i) = \Re(\chi_r(z))$ and $\Im(E_i) = \Im(\chi_r(z))$. Then $\Re(E_i)$ and $\Im(E_i)$ take values in $[-1,+1]$. Let the total number of samples required be $N$. Also let $F_1=\frac{\sum_{i=1}^N \Re(E_i)}{N}$, $F_2=\frac{\sum_{i=1}^N \Im(E_i)}{N}$ and $F=F_1+F_2$. Clearly, $\Re(E_i)$ are independent for all $i\in [N]$, since the samples $x_i,z_i$ are drawn uniformly and independently at random (and the same holds for $\Im(E_i)$, $i \in [N]$). Now let us apply Hoeffding's inequality (Lemma~\ref{Hoeffding}) on $\Pr[|F_1 - \mathbb{E}(F_1)|\geq \frac{\varepsilon}{\sqrt{2}}]$ and $\Pr[|F_2 - \mathbb{E}(F_2)|\geq \frac{\varepsilon}{\sqrt{2}}]$ and upper bound both the bounds by $\dfrac{\delta'}{2}$, as then
    \begin{align*}
        \Pr[|F - \mathbb{E}(F)|\geq \varepsilon] &= \Pr[|F_1 - \mathbb{E}(F_1)|^2 + |F_2 - \mathbb{E}(F_2)|^2 \geq \varepsilon^2] \\
        &\leq \Pr\left[|F_1 - \mathbb{E}(F_1)|\geq \frac{\varepsilon}{\sqrt{2}}\right] + \Pr\left[|F_2 - \mathbb{E}(F_2)|\geq \frac{\varepsilon}{\sqrt{2}}\right] \\
        &\leq \frac{\delta'}{2} + \frac{\delta'}{2} \\
        &= \delta'.
    \end{align*}
    So,
    \begin{align*}
        \Pr\left[\left|F_1 - \mathbb{E}(F_1)\right|\geq \frac{\varepsilon}{\sqrt{2}}\right] &= \Pr \left[  \left| \frac{\sum_{i=1}^N \Re(E_i)}{N} - \frac{\mathbb{E} \left( \sum_{i=1}^N \Re(E_i) \right) }{N} \right| \geq \frac{\varepsilon}{\sqrt{2}} \right] \\
        &= \Pr \left[ \left| \sum_{i=1}^N \Re(E_i) - \mathbb{E} \left( \sum_{i=1}^N \Re(E_i) \right) \right| \geq \frac{N\varepsilon}{\sqrt{2}} \right] \\
        &\leq 2 e^{-\frac{N^2\varepsilon^2}{2\times 2N}} \\
        &= 2 e^{-\frac{N\varepsilon^2}{4}} \\
        &\leq \frac{\delta'}{2},
    \end{align*}
    The last inequality follows from the fact that 
    which implies that 
    \begin{align}\label{eqn_wt_2_estimate}
        N \geq \frac{4 \ln \left( {4}/{\delta'} \right)}{\varepsilon^{2}}.
    \end{align}
\end{proof}

\section{A New proof of Green's regularity lemma for finite Abelian groups}\label{section_Green_regularity}

The regularity Lemma for Abelian groups was already known from~\cite{green2005szemeredi}. In this section, we start by reproducing Green's original proof without using Bohr sets, for finite Abelian groups with constant torsion. 
Let us first fix some notation that we will be using. For a character $r$, let $
d_r:=|\operatorname{range}(r)|.$
Since $G$ has exponent $\tau$, we have $d_r\leq \tau$.
Thus, fixing a character $r$ splits a node into exactly $d_r$
non-empty children, and in particular into at most $\tau$ non-empty
children. Whenever only an upper bound is required, we use
$d_r\leq\tau$.
\begin{theorem}[Green's regularity lemma~\cite{green2005szemeredi}, Theorem~\ref{thm_green_regularity}, restated]

    Let $A \subseteq \G$ and $\varepsilon> 0$. There exists a normal subgroup $H \subseteq \G$ with cosets $\{H_i\}$ such that 

    \begin{itemize}
        \item $H$ can be derived on making $\Tower(\lcmG,O(\frac{1}{\varepsilon^3}))$ many fixing; that is, the pseudo-codimension (see Definition~\ref{defn_codimension}) of $H$ is at most $\Tower(\lcmG,O(\frac{1}{\varepsilon^3}))$. 
        \item for all but $\varepsilon$-fraction of the cosets of $H$, the functions $\mathbbm{1}_A|_{H_i}: H \to \{0,1\}$ is $\varepsilon$-regular , where 
        $$\mathbbm{1}_A|_{H_i}(x) = 
\begin{cases}
1, & x+a_i\in H_i\cap A \\
0, & x+a_i \notin H_i\cap A,
\end{cases}$$
where $H_i = H + a_i$ with $a_i \in \G$.
    \end{itemize}
\end{theorem}

\begin{proof}

We first construct a character decision tree (see Definition~\ref{defn_char_tree}) that recursively partitions the group $\widehat{G}$ (see Definition~\ref{defn_dual_group}) according to the values of carefully chosen characters (in this case, the largest one). At each internal node of the tree, we split the node with respect to the largest Fourier coefficient on the restriction of the function to the corresponding coset, thereby refining the partition only when the current restriction is far from being $\varepsilon$-regular. The leaves of the resulting tree, therefore, correspond to cosets on which no large nontrivial Fourier coefficient remains, implying that the associated restricted functions are $\varepsilon$-regular. Let us now construct the character decision tree in the following way:

\begin{itemize}
    \item At the first stage, the decision tree has only one node, labelled by the function $\mathbbm{1}_A$. 
    
    If $\mathbbm{1}_A$ is $\varepsilon$-regular, we are done.
    
    If not, we consider the largest non-trivial Fourier coefficient of $\mathbbm{1}_A$, say $\widehat{\mathbbm{1}_A}(\chi_r)$. We then split the node with respect to the character whose corresponding Fourier coefficient is the largest, into atmost $\lcmG$ many subfunctions $g_b = \mathbbm{1}_A|_{V_{b,r}}$, where $V_{b,r}, \ b \in \{0,1, \ldots, \lcmG-1\}$ are the cosets (see Definition~\ref{defn_cosets_groups}) $$V_{b,r} = \{x \in \G : r * x = b\}.$$

    \item At the second stage, if at least a $(1-\varepsilon)$-fraction of the
subfunctions $g_b$ are $\varepsilon$-regular, we are done. Otherwise,
at least an $\varepsilon$-fraction of the subfunctions are not
$\varepsilon$-regular.

For each such non-$\varepsilon$-regular subfunction $g_b$, consider
a largest nontrivial Fourier coefficient of $g_b$. Let
$\{r_1',\ldots,r_t'\}$
be the collection of characters obtained in this way from all the
non-$\varepsilon$-regular subfunctions at the current stage, after
discarding those characters which are already generated by the
characters used in the previous stages.

We then refine the current partition simultaneously with respect to
all the characters
$r_1',\ldots,r_t'$.
Thus, every current coset is refined using the same collection of
new characters, and the resulting non-empty sets are cosets of the subgroup
$V_{0,r,r_1',\ldots,r_t'}.$
In particular, the resulting subfunctions are restrictions of
$\mathbbm{1}_A$ to sets of the form
\[
V_{b,b_1',\ldots,b_t',r,r_1',\ldots,r_t'}
=
\left\{
x\in\G:
r*x=b,\;
r_j'*x=b_j' \text{ for every }j\in[t]
\right\}.
\]

\color{black}

    \item At the third stage, if there exists a subfunction $g_{b,b'}$ which is $\varepsilon$-regular, we are done!

    If not, we again consider each subfunction $g_{b,b'}$ that are not $\varepsilon$-regular. We again consider the largest Fourier coefficients for each such $g_{b,b'}$, and also check whether the corresponding character is generated by $\{r,r_1',\ldots,r_t'\}$. If yes, we continue the aforementioned node-splitting process.
\end{itemize}

Now, suppose that at a certain stage the number of non-empty cosets is
$N$. We define the progress measure by
\begin{align}\label{progress_measure}
\ExpImb
=
\frac{1}{N}
\sum_{\text{all non-empty cosets }
V_{b,r_1,\ldots,r_t}}
\left|
\widehat{
\mathbbm{1}_A|_{V_{b,r_1,\ldots,r_t}}
}(\chi_0)
\right|^2.
\end{align}
At the first stage,
\[
\ExpImb
=
\left|
\widehat{\mathbbm{1}_A}(\chi_0)
\right|^2.
\]

At the next stage, suppose that the node is split with respect to the
character $r$. Recall that
$d_r:=|\operatorname{range}(r)|.$
Hence there are exactly $d_r$ non-empty cosets $V_{b,r}$, with
$b\in\operatorname{range}(r)$. Therefore, the progress measure after
this refinement is
\begin{align}\label{eqn_ExpImb_second_stage}
\ExpImb
=
\frac{1}{d_r}
\sum_{b\in\operatorname{range}(r)}
\left|
\widehat{g_b}(\chi_0)
\right|^2.
\end{align}

Since the character $\chi_r$ has order $d_r$, the characters
$\chi_{ir},\ i=0,\ldots,d_r-1,$
are distinct. From Lemma~\ref{lemma_restriction}, for every
$b\in\operatorname{range}(r)$,
\begin{align*}
\widehat{g_b}(\chi_0)
=
\sum_{i=0}^{d_r-1}
\widehat{\mathbbm{1}_A}(\chi_{ir})
\omega_{\lcmG}^{ib}.
\end{align*}
Consequently,
\begin{align*}
\left|
\widehat{g_b}(\chi_0)
\right|^2
&=
\sum_{i=0}^{d_r-1}
\left|
\widehat{\mathbbm{1}_A}(\chi_{ir})
\right|^2 
+
\sum_{\substack{j,k\in\{0,\ldots,d_r-1\}\\j\neq k}}
\widehat{\mathbbm{1}_A}(\chi_{jr})
\overline{
\widehat{\mathbbm{1}_A}(\chi_{kr})
}
\omega_{\lcmG}^{b(j-k)}.
\end{align*}

Averaging over $b\in\operatorname{range}(r)$ and using character
orthogonality on $\operatorname{range}(r)$ gives
\begin{align*}
\frac{1}{d_r}
\sum_{b\in\operatorname{range}(r)}
\left|
\widehat{g_b}(\chi_0)
\right|^2
=
\sum_{i=0}^{d_r-1}
\left|
\widehat{\mathbbm{1}_A}(\chi_{ir})
\right|^2.
\end{align*}
Therefore,
\begin{align*}
\ExpImb
=
\sum_{i=0}^{d_r-1}
\left|
\widehat{\mathbbm{1}_A}(\chi_{ir})
\right|^2.
\end{align*}

In particular, the increase in the progress measure caused by
refining a non-$\varepsilon$-regular node using the character $r$ is
at least
\begin{align*}
\sum_{i=1}^{d_r-1}
\left|
\widehat{\mathbbm{1}_A}(\chi_{ir})
\right|^2
\geq
\left|
\widehat{\mathbbm{1}_A}(\chi_r)
\right|^2
\geq
\varepsilon^2.
\end{align*}

The same calculation applies at every subsequent stage after replacing
$\G$ by the current subgroup
$V_{0,r_1,\ldots,r_t}.$
More precisely, for a heavy nontrivial character $r$ of the current
restricted function, $d_r$ denotes the size of the range of that
character on the current subgroup. Since the exponent of every
subgroup of $\G$ divides $\lcmG$, we always have
$d_r\leq \lcmG.$
Thus the above energy-increment calculation applies identically at
every stage.

Now suppose that at the current stage there are $N$ cosets and at
least an $\varepsilon$-fraction of them are not
$\varepsilon$-regular. For each such coset we choose a nontrivial
Fourier coefficient of absolute value at least $\varepsilon$, and
refine the partition as described above. Since refinement by any
additional characters can only increase the progress measure, the
total increase satisfies
\begin{align}
\ExpImb_{\mathrm{new}}
-
\ExpImb_{\mathrm{old}}
&\geq
\frac{N\varepsilon}{N}\varepsilon^2 =
\varepsilon^3.
\end{align}

Proceeding in this manner with the subsequent stages, as long as more
than an $\varepsilon$-fraction of the current cosets are not
$\varepsilon$-regular, the progress measure increases by at least
$\varepsilon^3$ at every stage. Since
\[
0\leq \ExpImb\leq 1,
\]
the process terminates after at most
$O\left(\frac{1}{\varepsilon^3}\right)$
stages. Therefore, after
$O(1/\varepsilon^3)$ stages, at least a
$(1-\varepsilon)$-fraction of the cosets
$H_i=V_{b,r_1,\ldots,r_t}$
of the normal subgroup
$H=V_{0,r_1,\ldots,r_t}$
are such that
$\mathbbm{1}_A|_{H_i}$
is $\varepsilon$-regular.

We now calculate the pseudo--codimension
(Definition~\ref{defn_codimension}) of the final subgroup.
Suppose that at a certain stage the character decision tree uses
$d_i$ characters. Then the number of non-empty leaves is at most
$\lcmG^{d_i}.$
At the next stage, in the worst case, one new heavy character may be
obtained from each non-$\varepsilon$-regular leaf. Hence the total
number of new characters introduced at this stage is at most
$\lcmG^{d_i},$
and therefore
\[
d_{i+1}
\leq
d_i+\lcmG^{d_i}.
\]

Thus,
\begin{itemize}
    \item \textbf{Stage 1:} $d_1\leq 1$.

    \item \textbf{Stage 2:}
    $d_2\leq 1+\lcmG.$

    \item \textbf{Stage 3:}
    $d_3
    \leq
    1+\lcmG+
    \lcmG^{\,1+\lcmG}.$
\end{itemize}

Continuing this recurrence for
$O(1/\varepsilon^3)$ stages gives
\[
d_i
\leq
\Tower\left(
\lcmG,
O\left(\frac{1}{\varepsilon^3}\right)
\right).
\]
Since the final subgroup is
$H=V_{0,r_1,\ldots,r_{d_i}},$
its pseudo--codimension is at most
$\Tower\left(
\lcmG,
O\left(\frac{1}{\varepsilon^3}\right)
\right).$
Moreover, since every character collected at a stage is subsequently
used to refine all the current cosets, the resulting character
decision tree is non-adaptive.
\end{proof}

\section{List Correction for Finite Abelian Groups}\label{section_implicit_sieve}

In this section, we give the proof of the list correction for finite Abelian groups. This was previously appeared in the literature by the work~\cite{DBLP:conf/approx/DattaGKPR25} (Generalized Implicit Sieve Algorithm), although we need a bit different set-up since the input-output demand is a bit different. To the core both the algorithm does the same job, for completeness we give a proof here with the parameter set-up that would be helpful for our purposes. 

\begin{definition}[List Correction for $\G$]\label{defn_list_correction}
Given oracle access to $g$, a parameter $\theta > 0$, a coset $a_g + H_g$, and a point $x \in \mathcal{G}$, the procedure $\textsc{ListCorrection}(g, \theta, a_g + H_g, x)$ returns $\chi_\alpha(x)$ for some $\alpha \in a_g + H_g$, assuming there exists $\alpha$ such that $|\widehat{g}(\alpha)| \ge \theta$.
\end{definition}

\subsection{List Correction for Finite Abelian Groups}



\begin{theorem}
    \label{th: Implicit Goldreich Levin}
      Let $\G$ be a finite Abelian group and $\theta > 0$ be a threshold and $\mathcal{M} = \left\{x_i\right\}_{i=1}^{\Tilde{m}} \subseteq \mathcal{G}$ be a set of independently and uniformly chosen $\Tilde{m}$ random points from $\mathcal{G}$ where $m = O(\log(\frac{1}{\delta}))$. 
    There exists an algorithm, List Correction for Abelian groups, (Algorithm~\ref{algo:Implicit_Sieve_generalized}), that takes $\theta$ and $\mathcal{M}$ as input, makes $O\left(
\frac{1}{\theta^{12}}
\log\frac{1}{\delta\theta^2}
\log\left(
\frac{1}{\delta\theta^{10}}
\log\frac{1}{\delta\theta^2}
\right)
\right)$ queries to the truth table of a Boolean valued function $f: \mathcal{G} \to \zone$, and returns, with probability at least $1-\delta$, and $T = O(\frac{1}{\theta^2})$ it returns $T$ oracle machine $\{\mathcal{O}_{\chi_{\alpha_1}}, \dots, \mathcal{O}_{\chi_{\alpha_T}}\}$, where each $\mathcal{O}_{\chi_{\alpha_i}}\}$ (a labeled set of $\Tilde{m}$ examples, and the value of the function $f$ at $x_i$ for each $i \in [\Tilde{m}]$, of the form $\left\{ \chi_{\alpha_1}\left(x_i\right), \dots, \chi_{\alpha_T}\left(x_i\right), f\left(x_i\right) \mid x_i \in \mathcal{M} \right\}$, where $\mathcal{S} = \left\{\alpha_1, \dots, \alpha_{T} \right\}$) that satisfies the following two properties:
    \begin{enumerate}
        \item For each $1 \leq i \leq T$, there is a distinct ${\alpha_i} \in \G$ such that $\mathcal{O}_{\chi_{\alpha_i}}$ is a probabilistic oracle machine
for the function $\chi_{\alpha_i}$. The query complexity of each oracle is $\poly(\frac{1}{\theta})$. 
        \item
            $\forall \alpha \in \mathcal{G}$ with $\left|\widehat{f}(\chi_\alpha)\right| \geq \theta$ then $\alpha \in \mathcal{S}$, and
        \item 
            $\forall \alpha \in \mathcal{S}$, $\left| \widehat{f}(\chi_\alpha) \right| \geq {\theta}/{2}$.
    \end{enumerate}
    The output of this algorithm can be seen as a $\Tilde{m} \times (T+1)$ matrix $Q$ with entries in powers of $\omega_\lcmG$, where $\omega_{\lcmG}$ is a primitive $\lcmG^{th}$ root of unity, $\lcmG= \LCM\{p_1^{m_1}, \ldots, p_n^{m_n}\}$, and $\lcmG$ is a constant.
\end{theorem}

\begin{remark}
\label{remark: k small in implicit sieve algo}
    Algorithm~\ref{algo:Implicit_Sieve_generalized}, however, does not return $\mathcal{S}$. Also observe that by Parseval's identity,  $|\mathcal{S}| = O(1/\theta^2)$.
\end{remark}

Let
\[
\Delta_{\lcmG}
:=
\min_{\substack{0\leq a,b<\lcmG\\a\neq b}}
\left|
\omega_{\lcmG}^a-\omega_{\lcmG}^b
\right|
=
2\sin\left(\frac{\pi}{\lcmG}\right).
\]
Since $\lcmG$ is a constant, $\Delta_{\lcmG}>0$ is a constant.

\begin{algorithm}[H]
\caption{List Correction for Abelian Groups}
\label{algo:Implicit_Sieve_generalized}
\begin{algorithmic}[1]

\Statex \textbf{Input:}
A function $f:\G\rightarrow\{0,1\}$, a confidence parameter $\delta$,
a threshold $\theta$, and a set
$M=\{x_1,\ldots,x_{\widetilde m}\}$ of uniformly and independently
chosen group elements, where
\[
\widetilde m=
\Theta\!\left(
\frac1{\theta^8}
\log\frac1{\delta\theta^2}
\right).
\]

\Statex \textbf{Output:}
A collection
$\{\mathcal O_{\chi_{\alpha_1}},\ldots,\mathcal O_{\chi_{\alpha_T}}\}$,
where each
$\mathcal O_{\chi_{\alpha_i}}$
is the $\widetilde m\times1$ column vector whose entries are
$\chi_{\alpha_i}(x_j)$,
$j\in[\widetilde m]$.
Also output the vector
\[
\mathcal F=
(f(x_1),\ldots,f(x_{\widetilde m}))^T,
\]
where $T=O(1/\theta^2)$.

\State

\State Set
$t
=
\left\lceil
\log_{\lcmG_{\min}}
\left(
\frac{96}{\delta\theta^4\eta}
\right)
\right\rceil,$ where $\eta
:=\frac{
\delta^2\Delta_{\lcmG}^2\theta^{16}
}{
128^2\cdot192^2\widetilde m^{\,2}
}.$

\State Construct a randomly permuted coset structure
$(H,\mathcal C)$
such that
$\operatorname{codim}(H)\le t$.

\State Initialize
$\mathcal L\gets\emptyset$.

\For{each bucket $C\in\mathcal C$}

    \State Estimate
    $\mathrm{wt}_2(C)$
    within additive error
    $\pm\theta^2/4$
    and confidence
    $1-\dfrac{\delta}{6\lcmG^t}$.

    \State Let the estimate be
    $\widetilde{\mathrm{wt}}_2(C)$.

    \If{$\widetilde{\mathrm{wt}}_2(C)\ge
    \dfrac{3\theta^2}{4}$}

        \State Keep bucket $C$.

        \State Sample
        $y_1,\ldots,y_{\widetilde m}$
        independently and uniformly from $\G$.














\State Estimate $\wt_2(C)$ within additive error
$\frac{\Delta_{\lcmG}\theta^2}{1000}$
and confidence
$1-\frac{\delta}{12T}.$
Let this estimate be denoted by
$\widetilde{\wt_2}'(C).$

\For{$i=1,\ldots,\widetilde m$}

\State Estimate
$P_Cf(y_i)$
within additive error
$\frac{\Delta_{\lcmG}\theta^2}{1000}$
and confidence
$1-\frac{\delta}{24T\widetilde m}.$
Let the estimate be denoted by
$\widetilde{P_Cf}(y_i)$.

\State Estimate
$P_Cf(y_i-x_i)$
within additive error
$\frac{\Delta_{\lcmG}\theta^2}{1000}$
and confidence
$1-\frac{\delta}{24T\widetilde m}.$
Let the estimate be denoted by
$\widetilde{P_Cf}(y_i-x_i)$.

\State Compute
\[
Z_{ij}
=
\frac{
\widetilde{P_Cf}(y_i)\,
\overline{\widetilde{P_Cf}(y_i-x_i)}
}{
\widetilde{\wt_2}'(C)
}.
\]

\State Let $Q_{ij}$ be the element of
$\left\{
1,\omega_{\lcmG},\ldots,
\omega_{\lcmG}^{\lcmG-1}
\right\}$
which is closest to $Z_{ij}$.

\EndFor

\EndIf

\EndFor
\end{algorithmic}
\end{algorithm}

\subsection{Proof of Theorem~\ref{th: Implicit Goldreich Levin}}

\begin{definition}[Random shift]\label{def:random-shift}
        Let $u\in \mathbb{Z}_{\lcmG}^{t}\setminus\{0\}$ be chosen uniformly at random. For every $b\in\mathbb{Z}_{\lcmG}^{t}$, define $C_u(b):=C(b+u).$
    We call this a random shift of the coset structure.

    Observe that the random shift only permutes the labels of the buckets. In particular, for any $\alpha_1,\alpha_2\in\mathcal{G}$, whether $\alpha_1$ and $\alpha_2$ belong to the same bucket is independent of the choice of $u$.
\end{definition}

\begin{definition}
    Let $\mathcal{C}$ be the set of all $C(b)$'s. We call $(H, \mathcal{C})$ a random permuted coset structure.
\end{definition}

\begin{lemma}\label{lemma_big_coeff_diff_buckets}
    Fix $\alpha\in\mathcal{G}$ and
    $b\in\mathbb{Z}_{\lcmG}^{t}$. Let $h(\alpha):=\bigl(\langle \alpha,\beta_1\rangle,\ldots,\langle \alpha,\beta_t\rangle\bigr).$ Then
        $$\Pr_u[\alpha\in C_u(b)] =
        \dfrac{1}{\lcmG^t-1},
        \qquad \text{if } h(\alpha)-b\neq 0,$$ and
        $$\Pr_u[\alpha\in C_u(b)] = 0,
        \qquad \text{if } h(\alpha)-b=0.$$
\end{lemma}

\begin{proof}
    By definition, $\alpha\in C_u(b)$
    if and only if $\langle\alpha,\beta_i\rangle=b_i+u_i \text{ for every }i\in[t].$ Equivalently, $u=h(\alpha)-b.$
    
    If $h(\alpha)-b\neq0$, then,
    \[
        \Pr_u[\alpha\in C_u(b)]
        = \frac{1}{\lcmG^t-1}.
    \]
    If $h(\alpha)-b=0$, the required shift is $u=0$, which is excluded.
\end{proof}

\begin{lemma}\label{lemma_large_char_diff_buckets}
    Let $B=\left\{\alpha\in\mathcal{G}:|\widehat{f}(\chi_\alpha)|\geq\frac{\theta^2}{4}\right\}.$
    If
    $t\geq2\log_{\lcmG_{\min}}\left(\frac{16}{\theta^4}\right)+\log_{\lcmG_{\min}}\left(\frac{1}{\delta}\right),$
    then, except with probability at most $\delta$, all the characters in $B$ belong to distinct buckets, where $\lcmG_{\min}$ denotes the smallest prime that divides $|\G|$.
\end{lemma}

\begin{proof}
    By Parseval's identity,
    $|B|\leq\frac{16}{\theta^4}.$ Consider two distinct characters $\alpha_1,\alpha_2\in B$. The random shift does not affect whether they belong to the same bucket. Indeed,
    \begin{align*}
        \alpha_1,\alpha_2\text{ belong to the same bucket}
        &\Longleftrightarrow
        \langle\alpha_1,\beta_i\rangle
        =
        \langle\alpha_2,\beta_i\rangle
        \qquad\forall i\in[t]\\
        &\Longleftrightarrow
        \langle\alpha_1-\alpha_2,\beta_i\rangle
        =
        0
        \qquad\forall i\in[t].
    \end{align*}

    Since $\alpha_1-\alpha_2\neq0$, we have
    \[
        \Pr_{\beta_i}
        \left[
        \langle\alpha_1-\alpha_2,\beta_i\rangle=0
        \right]
        \leq
        \frac{1}{\lcmG_{\min}}.
    \]
    Since $\beta_1,\ldots,\beta_t$ are chosen independently,
    \[
        \Pr[
        \alpha_1,\alpha_2
        \text{ belong to the same bucket}
        ]
        \leq
        \frac{1}{\lcmG_{\min}^{\,t}}.
    \]

    Therefore, by a union bound,
    \begin{align*}
        &\Pr[
        \text{some two distinct characters in $B$
        belong to the same bucket}
        ]\\
        &\leq
        \binom{|B|}{2}
        \frac{1}{\lcmG_{\min}^{\,t}}\leq
        \left(\frac{16}{\theta^4}\right)^2
        \frac{1}{\lcmG_{\min}^{\,t}}\leq
        \delta.
    \end{align*}
    Hence, with probability at least $1-\delta$, all the characters
    in $B$ belong to distinct buckets.
\end{proof}

From Lemma~\ref{lemma_large_char_diff_buckets}, with high probability,
the characters in
    $B=\left\{\alpha\in\mathcal{G}:|\widehat{f}(\chi_\alpha)|\geq\frac{\theta^2}{4}\right\}$ belong to distinct buckets. Let $S=\left\{\gamma\in\mathcal{G}:|\widehat{f}(\chi_\gamma)|<\frac{\theta^2}{4}\right\}.$ For each $\alpha\in B$, let $C_\alpha$ denote the bucket containing
$\alpha$. Observe that the random shift does not affect whether another character $\gamma$ belongs to $C_\alpha$. Indeed, $\gamma\in C_\alpha\Longleftrightarrow\langle\gamma-\alpha,\beta_i\rangle=0\qquad\forall i\in[t].$ Thus the small Fourier weight contained in $C_\alpha$ is controlled
entirely by the random choice of $\beta_1,\ldots,\beta_t$.

\begin{lemma}\label{lemma_small_char_uniform_dist}
    Let 
    $\eta
:=\frac{
\delta^2\Delta_{\lcmG}^2\theta^{16}
}{
128^2\cdot192^2\widetilde m^{\,2}
}.$ Suppose that the characters in $B$ belong to distinct buckets.
    If 
    $\lcmG_{\min}^{\,t}
\geq
\frac{96}{\delta\theta^4\eta},$
or equivalently,
$t\geq\log_{\lcmG_{\min}}
\left(
\frac{96}{\delta\theta^4\eta}
\right),$ then, 
with probability at least $1-\frac{\delta}{6}$
    simultaneously for every $\alpha\in B$,

    \[
X_{C_\alpha}
:=
\sum_{\gamma\in S\cap C_\alpha}
|\widehat f(\chi_\gamma)|^2
\leq
\eta.
\]
Consequently, on the additional event that all the characters in
$B$ belong to distinct buckets,
\[
\wt_2(C_\alpha)
\leq
|\widehat f(\chi_\alpha)|^2+\eta
\qquad
\text{for every }\alpha\in B.
\]

\end{lemma}

\begin{proof}
    Fix $\alpha\in B$. For every $\gamma\in S$, define
    \[
        Y_{\gamma,\alpha}
        :=
        \begin{cases}
        1,
        & \text{if }\gamma\in C_\alpha,\\
        0,
        & \text{otherwise}.
        \end{cases}
    \]
    Since $\gamma\in C_\alpha \Longleftrightarrow \langle\gamma-\alpha,\beta_i\rangle=0 \ \forall i\in[t],$
    and since $\gamma\neq\alpha$, we have
    \[
        \Pr[\gamma\in C_\alpha]
        \leq
        \frac{1}{\lcmG_{\min}^{\,t}}.
    \]

    Define
    \[
        X_{C_\alpha}
        :=
        \sum_{\gamma\in S}
        |\widehat{f}(\chi_\gamma)|^2
        Y_{\gamma,\alpha}.
    \]
    Thus $X_{C_\alpha}$ is precisely the contribution of the small
    Fourier coefficients to the weight of $C_\alpha$.

    By linearity of expectation,
    \begin{align*}
        \mathbb{E}[X_{C_\alpha}]=
        \sum_{\gamma\in S}
        |\widehat{f}(\chi_\gamma)|^2
        \Pr[\gamma\in C_\alpha]\leq
        \frac{1}{\lcmG_{\min}^{\,t}}
        \sum_{\gamma\in S}
        |\widehat{f}(\chi_\gamma)|^2 \leq
        \frac{1}{\lcmG_{\min}^{\,t}},
    \end{align*}
    where the last inequality follows from Parseval's identity.

    Therefore, by Markov's inequality, $$\Pr[X_{C_\alpha}\geq\eta]\leq\frac{1}{\lcmG_{\min}^{\,t}\eta}.$$




Since $|B|\leq16/\theta^4$, by union bound,
\begin{align*}
\Pr\left[
\exists\,\alpha\in B:
X_{C_\alpha}\geq\eta
\right]
&\leq
\frac{16}
{\theta^4\lcmG_{\min}^{\,t}\eta} \leq
\frac{\delta}{6}.
\end{align*}

Thus, with probability at least $1-\frac{\delta}{6}$,
simultaneously for every $\alpha\in B$,
\[
X_{C_\alpha}\leq\eta.
\]
If, in addition, the characters in $B$ belong to distinct buckets,
then
\[
\wt_2(C_\alpha)
=
|\widehat f(\chi_\alpha)|^2+X_{C_\alpha}
\leq
|\widehat f(\chi_\alpha)|^2+\eta.
\]
\end{proof}

The following lemma shows that a bucket survives this procedure only when it contains a large Fourier coefficient.

\begin{lemma}\label{lemma_no_small_only_bucket}
With probability at least $1-\frac{\delta}{6}$, every bucket $C$
such that
$C\cap B=\emptyset$
satisfies
\[
\wt_2(C)<\frac{\theta^2}{2}.
\]
\end{lemma}

\begin{proof}
For $\gamma\in S$, let
$w_\gamma
:=
|\widehat f(\chi_\gamma)|^2.$
Since $\gamma\in S$,
$w_\gamma<\frac{\theta^4}{16}.$

Suppose that a bucket $C$ containing only elements of $S$ satisfies
$\sum_{\gamma\in C}w_\gamma
\geq
\frac{\theta^2}{2}.$
Then
\begin{align*}
\sum_{\substack{\gamma,\gamma'\in C\\\gamma\neq\gamma'}}
w_\gamma w_{\gamma'}
&=
\left(\sum_{\gamma\in C}w_\gamma\right)^2
-
\sum_{\gamma\in C}w_\gamma^2\\
&\geq
\left(\sum_{\gamma\in C}w_\gamma\right)^2
-
\frac{\theta^4}{16}
\sum_{\gamma\in C}w_\gamma\\
&\geq
\frac{7\theta^4}{32}.
\end{align*}

For any two distinct $\gamma,\gamma'\in S$,
\[
\Pr[
\gamma,\gamma'
\text{ belong to the same bucket}
]
\leq
\frac{1}{\lcmG_{\min}^{\,t}}.
\]
Therefore,
\begin{align*}
\mathbb E
\left[
\sum_C
\sum_{\substack{\gamma,\gamma'\in C\\\gamma\neq\gamma'}}
w_\gamma w_{\gamma'}
\right]
&\leq
\frac{1}{\lcmG_{\min}^{\,t}}
\left(
\sum_{\gamma\in S}w_\gamma
\right)^2 \leq
\frac{1}{\lcmG_{\min}^{\,t}},
\end{align*}
where the last inequality follows from Parseval's identity.

Hence, by Markov's inequality,
\[
\Pr\left[
\exists\,C:
C\cap B=\emptyset
\text{ and }
\wt_2(C)\geq\frac{\theta^2}{2}
\right]
\leq
\frac{32}
{7\theta^4\lcmG_{\min}^{\,t}}
\leq
\frac{\delta}{6}.
\]
\end{proof}

\color{black}

\begin{lemma}
Algorithm~\ref{algo:Implicit_Sieve_generalized} does the following with probability $\geq 1-\frac{\delta}{6}$:
    \begin{enumerate}
        \item It does not discard any bucket $C$ with $|\widehat{f}(\chi_\gamma)| \geq \theta$.
        \item $|\widehat{f}(\chi_{\alpha(C)})| \geq \frac{\theta}{2}$ for every bucket $C=C_\alpha$ (that is, $\alpha(C)$ is the dominating element of $C_\alpha$) in the set of buckets that the algorithm has not discarded.
    \end{enumerate}
\end{lemma}

\begin{proof}
\begin{description}
    \item[\textbf{Proof of part (1).}] From Algorithm~\ref{algo:Implicit_Sieve_generalized} we know that $$\Tilde{\wt_2}(C) \geq \wt_2(C) - \frac{\theta^2}{4} \geq \theta^2 -\frac{\theta^2}{4} = \frac{3\theta^2}{4}.$$ Hence the bucket $C$ is not discarded.

    \item[\textbf{Proof of part (2).}]
Let $C$ be a bucket which has not been discarded. Then
$\widetilde{\wt_2}(C)
\geq
\frac{3\theta^2}{4}.$
Since
\[
\left|
\widetilde{\wt_2}(C)-\wt_2(C)
\right|
\leq
\frac{\theta^2}{4},
\]
we have
$\wt_2(C)
\geq
\frac{\theta^2}{2}.$
By Lemma~\ref{lemma_no_small_only_bucket}, $C$ must contain some
$\alpha(C)\in B$. Moreover, by
Lemma~\ref{lemma_large_char_diff_buckets}, this element is unique.

Suppose, for contradiction, that
\[
|\widehat f(\chi_{\alpha(C)})|
<
\frac{\theta}{2}.
\]
By Lemma~\ref{lemma_small_char_uniform_dist},
\[
\wt_2(C)
\leq
|\widehat f(\chi_{\alpha(C)})|^2+\eta
<
\frac{\theta^2}{4}+\eta.
\]
Since $\eta<\theta^2/4$,
\[
\wt_2(C)<\frac{\theta^2}{2}.
\]
Hence
\[
\widetilde{\wt_2}(C)
\leq
\wt_2(C)+\frac{\theta^2}{4}
<
\frac{3\theta^2}{4},
\]
contradicting the fact that $C$ survived. Therefore
\[
|\widehat f(\chi_{\alpha(C)})|
\geq
\frac{\theta}{2}.
\]
\end{description}
\end{proof}

\begin{lemma}\label{lemma_small_char_bound}
The small weight coefficients do not contribute much to the weight
of a bucket. That is,
\begin{align*}
\Pr_x
\left[
\left|
\sum_{r\in C_\alpha\setminus\{\alpha\}}
\widehat f(\chi_r)\chi_r(x)
\right|
\geq
\frac{\Delta_{\lcmG}\theta^4}{128}
\right]
\leq
\frac{\delta\theta^4}{192\widetilde m}.
\end{align*}
\end{lemma}

\begin{proof}
By Lemma~\ref{lemma_small_char_uniform_dist},
\[
\sum_{r\in C_\alpha\setminus\{\alpha\}}
|\widehat f(\chi_r)|^2
\leq
\eta
=
\frac{
\delta^2\Delta_{\lcmG}^2\theta^{16}
}{
128^2\cdot192^2\widetilde m^{\,2}
}.
\]
Hence, by Cauchy--Schwarz,
\begin{align*}
\mathbb E_x
\left[
\left|
\sum_{r\in C_\alpha\setminus\{\alpha\}}
\widehat f(\chi_r)\chi_r(x)
\right|
\right]
&\leq
\sqrt{\eta} =
\frac{
\delta\Delta_{\lcmG}\theta^8
}{
128\cdot192\widetilde m
}.
\end{align*}
Therefore, by Markov's inequality,
\begin{align*}
&\Pr_x
\left[
\left|
\sum_{r\in C_\alpha\setminus\{\alpha\}}
\widehat f(\chi_r)\chi_r(x)
\right|
\geq
\frac{\Delta_{\lcmG}\theta^4}{128}
\right] \leq
\frac{
\delta\Delta_{\lcmG}\theta^8/
(128\cdot192\widetilde m)
}{
\Delta_{\lcmG}\theta^4/128
}
=
\frac{\delta\theta^4}
{192\widetilde m}.
\end{align*}
\end{proof}




\begin{lemma}
Let $y_i$, $i\in[\widetilde m]$, be chosen uniformly and independently
at random. Then, except with probability at most
$\frac{\delta\theta^4}
{96\widetilde m}$,
\[
\left|
P_{C_\alpha}f(y_i)
\overline{P_{C_\alpha}f(y_i-x_i)}
-
|\widehat f(\chi_{\alpha(C)})|^2
\chi_{\alpha(C)}(x_i)
\right|
\leq
\frac{\Delta_{\lcmG}\theta^4}{60}.
\]
\end{lemma}

\color{black}
\begin{proof}
For each bucket $C_\alpha$, let $\alpha(C)$ be the dominating element, $\alpha(C) \in B$. By Lemma~\ref{lemma_large_char_diff_buckets}, with high probability over the random choice of $\beta_1,\ldots,\beta_t$, the elements of $B$ belong to distinct buckets. Notice that the random shift does not affect this event, since it only permutes the labels of the buckets. Observe that
    \begin{align*}
        &P_{C_\alpha}(y_i) \overline{P_{C_\alpha}(y_i-x_i)} \\
        &= \biggl( \widehat{f}(\chi_{\alpha(C)}) \chi_{\alpha(C)}(y_i) + \sum_{r \in C_\alpha \setminus \{ \alpha(C) \}} \widehat{f}(\chi_r) \chi_r(y_i) \biggr) \\
        &\times \overline{\biggl( \widehat{f}(\chi_{\alpha(C)}) \chi_{\alpha(C)}(y_i-x_i) + \sum_{r \in C_\alpha\setminus \{ \alpha(C) \}} \widehat{f}(\chi_r) \chi_r(y_i-x_i) \biggr)} \\
        &= \biggl( \widehat{f}(\chi_{\alpha(C)}) \chi_{\alpha(C)}(y_i) + \sum_{r \in C_\alpha \setminus \{ \alpha(C) \}} \widehat{f}(\chi_r) \chi_r(y_i) \biggr) \\
        &\times \biggl( \overline{\widehat{f}(\chi_{\alpha(C)})} \chi_{\alpha(C)}(x_i-y_i) + \sum_{r \in C_\alpha \setminus \{ \alpha(C) \}} \overline{\widehat{f}(\chi_r)} \chi_r(x_i-y_i) \biggr) \\
        &= |\widehat{f}(\chi_{\alpha(C)})|^2 \chi_{\alpha(C)}(x_i) + \biggl(\sum_{r \in C_\alpha \setminus \{ \alpha(C) \}} \widehat{f}(\chi_r) \chi_r(y_i) \biggr) \overline{\widehat{f}(\chi_{\alpha(C)})} \chi_{\alpha(C)}(x_i-y_i) \\
        &+ \biggl( \sum_{r \in C_\alpha \setminus \{ \alpha(C) \}} \overline{\widehat{f}(\chi_r)} \chi_r(x_i-y_i) \biggr) \widehat{f}(\chi_{\alpha(C)}) \chi_{\alpha(C)}(y_i) \\
        &+ \biggl( \sum_{r \in C_\alpha \setminus \{ \alpha(C) \}} \widehat{f}(\chi_r) \chi_r(y_i) \biggr) \biggl( \overline{\sum_{r \in C_\alpha \setminus \{ \alpha(C) \}} \widehat{f}(\chi_r) \chi_r(y_i-x_i)} \biggr)
    \end{align*}

    Now, 
    $$
        \left| \sum_{r \in C_\alpha \setminus \{ \alpha(C) \}} \widehat{f}(\chi_r) \chi_r(y_i) \right| \leq
\frac{\Delta_{\lcmG}\theta^4}{128},$$ 
    and 
    $$
        \left| \sum_{r \in C_\alpha \setminus \{ \alpha(C) \}} \widehat{f}(\chi_r) \chi_r(x_i - y_i) \right| \leq
\frac{\Delta_{\lcmG}\theta^4}{128}.$$ 
    with probability $\geq 1- \dfrac{4\delta\theta^4}{128\times 6\Tilde{m}}$ by Lemma~\ref{lemma_small_char_bound}. 



Therefore,
\begin{align*}
&\left|
P_{C_\alpha}f(y_i)
\overline{P_{C_\alpha}f(y_i-x_i)}
-
|\widehat f(\chi_{\alpha(C)})|^2
\chi_{\alpha(C)}(x_i)
\right| \leq
2\cdot
\frac{\Delta_{\lcmG}\theta^4}{128}
+
\left(
\frac{\Delta_{\lcmG}\theta^4}{128}
\right)^2 <
\frac{\Delta_{\lcmG}\theta^4}{60},
\end{align*}

with failure probability at most
$2\cdot
\frac{\delta\theta^4}
{192\widetilde m}
=
\frac{\delta\theta^4}
{96\widetilde m}.$

\end{proof}

So, by taking union bound over $\Tilde{m}$-elements, the above result holds for all $x_i$'s and $y_i$'s, $i \in [\Tilde{m}]$ with probability $\geq 1- \dfrac{\delta\theta^4}{96}$. Therefore, by taking union bound over $T \leq \dfrac{16}{\theta^4}$ buckets, the above holds with probability $\geq 1- \dfrac{\delta}{6}$.

Hence, Algorithm~\ref{algo:Implicit_Sieve_generalized} works with probability $\geq 1 - (\frac{\delta}{6} + \frac{\delta}{3} + \frac{\delta}{6} + \frac{\delta}{6} + \frac{\delta}{6}) = 1-\delta$.

Let $T$ denote the number of survived buckets. For every survived
bucket $C_j$, let
\[
a_j
:=
\widehat f(\chi_{\alpha(C_j)}).
\]
By the previous lemma,
\[
\left|
P_{C_j}f(y_i)
\overline{P_{C_j}f(y_i-x_i)}
-
|a_j|^2
\chi_{\alpha(C_j)}(x_i)
\right|
\leq
\frac{\Delta_{\lcmG}\theta^4}{60}.
\]

The estimates of $P_{C_j}f(y_i)$ and
$P_{C_j}f(y_i-x_i)$ each have additive error at most
$\frac{\Delta_{\lcmG}\theta^2}{1000}.$
Hence
\[
\left|
\widetilde{P_{C_j}f}(y_i)
\overline{\widetilde{P_{C_j}f}(y_i-x_i)}
-
|a_j|^2
\chi_{\alpha(C_j)}(x_i)
\right|
\leq
\frac{\Delta_{\lcmG}\theta^2}{40}.
\]

Moreover,
\[
\wt_2(C_j)
=
|a_j|^2
+
\sum_{r\in C_j\setminus\{\alpha(C_j)\}}
|\widehat f(\chi_r)|^2,
\]
and therefore
\[
\left|
\widetilde{\wt_2}'(C_j)-|a_j|^2
\right|
\leq
\frac{\Delta_{\lcmG}\theta^2}{1000}
+\eta
\leq
\frac{\Delta_{\lcmG}\theta^2}{500}.
\]

Since
\[
|a_j|\geq\frac{\theta}{2},
\]
we have
\[
|a_j|^2\geq\frac{\theta^2}{4},
\]
and consequently
\[
\widetilde{\wt_2}'(C_j)
\geq
\frac{\theta^2}{5}.
\]

Thus
\begin{align*}
\left|
Z_{ij}
-
\chi_{\alpha(C_j)}(x_i)
\right|
&\leq
\frac{
\Delta_{\lcmG}\theta^2/40
}{
\theta^2/5
}
+
\frac{
\Delta_{\lcmG}\theta^2/500
}{
\theta^2/5
} <
\frac{\Delta_{\lcmG}}{2}.
\end{align*}

By the definition of $\Delta_{\lcmG}$, distinct
$\lcmG$-th roots of unity are separated by at least
$\Delta_{\lcmG}$. Therefore the unique closest
$\lcmG$-th root of unity to $Z_{ij}$ is
$\chi_{\alpha(C_j)}(x_i).$
Hence
$Q_{ij}
=
\chi_{\alpha(C_j)}(x_i),$
and therefore
\[
Q=
\begin{bmatrix}
\chi_{\alpha(C_1)}(x_1) &
\cdots &
\chi_{\alpha(C_T)}(x_1) &
\mid & f(x_1)\\
\chi_{\alpha(C_1)}(x_2) &
\cdots &
\chi_{\alpha(C_T)}(x_2) &
\mid & f(x_2)\\
\vdots & & \vdots & \mid & \vdots\\
\chi_{\alpha(C_1)}(x_{\widetilde m}) &
\cdots &
\chi_{\alpha(C_T)}(x_{\widetilde m}) &
\mid & f(x_{\widetilde m})
\end{bmatrix}.
\]

\subsection{Query complexity of List Correction for Abelian groups,
Algorithm~\ref{algo:Implicit_Sieve_generalized}}

\begin{lemma}\label{lemma_wt__2_estimate}
Given query access to a Boolean valued function
$f:\mathcal{G}\to\{-1,+1\}$, the weights $\wt_2(C)$ of all the buckets
can be estimated simultaneously within additive error
$\pm\frac{\theta^2}{4}$, except with probability at most $\delta'$,
using $O\left(
\frac{1}{\theta^4}
\log\frac{\lcmG^t}{\delta'}
\right)$ queries.
\end{lemma}

\begin{proof}
The proof follows exactly along the lines of
Lemma~\ref{lemma_wt_estimate_new}. For a fixed bucket, taking
$\epsilon=\frac{\theta^2}{4}$
and failure probability
$\frac{\delta'}{\lcmG^t}$
gives
$N
=
O\left(
\frac{1}{\theta^4}
\log\frac{\lcmG^t}{\delta'}
\right)$ samples.

The same queried values of $f$ can be reused for all the buckets.
For a fixed sample $(x,z)$, the values $f(x)$ and $f(x+z)$ are
independent of the bucket, and only the corresponding character
factor changes. Hence no multiplicative factor $\lcmG^t$ is incurred
in the query complexity.

Taking a union bound over all the buckets proves the claim.
\end{proof}

\begin{lemma}\label{lemma_queries_proj}
For any fixed $i\in[\Tilde{m}]$, 
$P_Cf(y_i)
\text{ and }
P_Cf(y_i-x_i)$
for all surviving buckets $C$ can be estimated simultaneously within
additive error
$\frac{\Delta_{\lcmG}\theta^2}{1000},$
except with probability at most $\delta'$, using
$O\left(
\frac{1}{\theta^4}
\log\frac{T}{\delta'}
\right)$
queries.
\end{lemma}

\begin{proof}
The proof again follows exactly along the lines of
Lemma~\ref{lemma_wt_estimate_new}. For a fixed surviving bucket,
taking additive error
$\frac{\Delta_{\lcmG}\theta^2}{1000}$
and failure probability $\delta'/T$ gives
$O\left(
\frac{1}{\theta^4}
\log\frac{T}{\delta'}
\right)$
queries, since $\lcmG$ is a constant and hence
$\Delta_{\lcmG}>0$ is a constant.

For a fixed $y_i$, the same queried values of $f(y_i-z)$ can be used
for all the surviving buckets; only the corresponding character
factor depends on the bucket. The same observation applies to
$P_Cf(y_i-x_i)$. Thus no multiplicative factor $T$ is incurred.

Taking a union bound over all surviving buckets proves the claim.
\end{proof}

\begin{theorem}\label{query_implicit_sieve}
The query complexity of the List Correction for Abelian groups
(Algorithm~\ref{algo:Implicit_Sieve_generalized}) is
\[
O\left(
\frac{1}{\theta^{12}}
\log\frac{1}{\delta\theta^2}
\log\left(
\frac{1}{\delta\theta^{10}}
\log\frac{1}{\delta\theta^2}
\right)
\right),
\]
taking $\lcmG$ to be a constant.
\end{theorem}

\begin{proof}
We calculate the query complexity of the different steps of
Algorithm~\ref{algo:Implicit_Sieve_generalized} separately.

First, the algorithm estimates the weight $\wt_2(C)$ of every bucket
within additive error $\theta^2/4$. Taking the failure probability
for these estimates to be at most $\delta/6$, by
Lemma~\ref{lemma_wt__2_estimate} this requires
$O\left(
\frac{1}{\theta^4}
\log\frac{\lcmG^t}{\delta}
\right)$
queries.

Next, for every surviving bucket, the algorithm estimates
$\wt_2(C)$ more accurately, within additive error
$\frac{\Delta_{\lcmG}\theta^2}{1000}.$
Since the same samples can be reused for all surviving buckets, these
estimates require
$O\left(
\frac{1}{\theta^4}
\log\frac{T}{\delta}
\right)$
queries.

We next estimate
$P_Cf(y_i)
\text{ and }
P_Cf(y_i-x_i),
\qquad
i\in[\Tilde{m}],$
within additive error
$\frac{\Delta_{\lcmG}\theta^2}{1000}.$
For each fixed $i$, by Lemma~\ref{lemma_queries_proj}, the estimates
for all the surviving buckets can be obtained using
$O\left(
\frac{1}{\theta^4}
\log\frac{T\Tilde{m}}{\delta}
\right)$
queries, where we have chosen the failure probability sufficiently
small so that a union bound can be taken over all
$i\in[\Tilde{m}]$.

Therefore, over all $i\in[\Tilde{m}]$, the total query complexity of
the projection estimates is
$O\left(
\frac{\Tilde{m}}{\theta^4}
\log\frac{T\Tilde{m}}{\delta}
\right).$

Combining the three, the total number of queries is
\[
O\left(
\frac{1}{\theta^4}
\log\frac{\lcmG^t}{\delta}
+
\frac{1}{\theta^4}
\log\frac{T}{\delta}
+
\frac{\Tilde{m}}{\theta^4}
\log\frac{T\Tilde{m}}{\delta}
\right).
\]

Now,
$T
=
O\left(\frac{1}{\theta^2}\right)$
and
$\Tilde{m}
=
\Theta\left(
\frac{1}{\theta^8}
\log\frac{1}{\delta\theta^2}
\right).$
Hence
\begin{align*}
\frac{\Tilde{m}}{\theta^4}
\log\frac{T\Tilde{m}}{\delta}
&=
O\left(
\frac{1}{\theta^{12}}
\log\frac{1}{\delta\theta^2}
\log\left(
\frac{1}{\delta\theta^{10}}
\log\frac{1}{\delta\theta^2}
\right)
\right).
\end{align*}

Moreover, by the choice of $t$ in
Algorithm~\ref{algo:Implicit_Sieve_generalized}, and since $\lcmG$
and $\lcmG_{\min}$ are constants,
\[
\log(\lcmG^t)
=
O\left(
\log\frac{1}{\delta}
+
\log\frac{1}{\theta}
+
\log\Tilde{m}
\right).
\]

Therefore, the total query complexity is
\[
O\left(
\frac{1}{\theta^{12}}
\log\frac{1}{\delta\theta^2}
\log\left(
\frac{1}{\delta\theta^{10}}
\log\frac{1}{\delta\theta^2}
\right)
\right).
\]
\end{proof}

\section{Estimating Fourier Coefficients over Finite Abelian Groups}\label{section_estimation}

In this section, we give the proof of the correctness of Algorithm~\ref{algo: estimation}. This was appeared in the work~\cite{DBLP:conf/approx/DattaGKPR25}, we have given the same proof on changing the parameters.

\begin{algorithm}[]
    \textbf{Input.}
    Given query access to $f: \mathcal{G} \to \zone$ and $\theta'$ (error parameter), $\delta'$ (confidence parameter). 
    
    \textbf{Output.} $\{\widehat{\Tilde{f}}(\chi_{S_j})\}, \ i \in [T]$.

\begin{enumerate}
    \item Let $\theta = \varepsilon \cdot 10\pi/\lcmG$.
    Let $\mathcal{M}$ be a set of $\Tilde{m}$ elements chosen uniformly and independently from $\mathcal{G}$, where $\Tilde{m} = \Theta(\frac{1}{\theta^{8}}\log\frac{1}{\delta\theta^2})$. Run the \textsf{List Correction for Abelian Groups Algorithm}~\ref{algo:Implicit_Sieve_generalized} with $\mathcal{M}$ and $\theta, \delta$ as input to obtain a $(\Tilde{m} \times (T+1))$ matrix $Q$, where $T = O(1/\theta^2)$.\label{line: query algo: call Implicit Sieve}

    \item \label{line: query algo: find B}  Label the rows of $Q$ with $1, \dots ,\Tilde{m}$ (which corresponds to the set $\mathcal{M}$) and columns with $1, \dots, T$ (which corresponds to the set of large Fourier coefficients of $f$), leaving the last column. Let us denote the set of columns by $B$.
        Find the smallest $\Tilde{B} = \{r_1, \ldots, r_{|\Tilde{B}|}\} \subseteq B$ such that, 
        for any column $r \in B\setminus \Tilde{B}$, there exists $\lambda, \lambda_1, \ldots, \lambda_{|\Tilde{B}|} \in \mathbb{N}\cup\{0\}$, with $\lambda$ an unit in $\mathbb{Z}_\lcmG$ (that is, $\lambda$ is invertible in $\mathbb{Z}_\lcmG$) such that $\lambda r +\sum_{i\in [|\Tilde{B}|]} \lambda_i r_i = 0$.

        \begin{itemize}
            \item If such a $\Tilde{B}$ does not exist, then return \textsf{Fail}.
        \end{itemize}

    \item \label{line: query algo: relabel columns}
        Relabel the columns of $Q$ by $\{S_1, \dots ,S_{T}\}$ where $S_1, \dots S_{|\Tilde{B}|}$ such that for all $j \in \{{|\Tilde{B}|}+1, \dots, T\}$, $S_j$ can be written as $\frac{1}{\lambda} \biggl( \sum_{i\in [|\Tilde{B}|]} \lambda_i S_i \biggr)$, $\lambda, \lambda_1, \ldots, \lambda_{|\Tilde{B}|} \in \mathbb{N}\cup\{0\}$, and $\lambda$ is an unit in $\mathbb{Z}_\lcmG$. Note that the columns $S_1, \dots, S_{|\Tilde{B}|}$ correspond to those in $\Tilde{B}$.

    \item \label{line: query algo: estimate ^f}
        Query $f(x^{(1)}), \dots, f(x^{(\Tilde{m})})$. 
        For all $j \in [T]$, estimate $\widehat{f}(\chi_{S_j})$ to error $\theta'$ and confidence $\delta'$ by $r_j$ where
        \begin{align*}
            r_j = \widehat{\Tilde{f}}(\chi_{S_j}) = \frac{\sum_{i=1}^{\Tilde{m}} f(x^{(i)})Q[i,j]}{\Tilde{m}}.
        \end{align*}
        \textbf{Sample complexity $= O(\frac{1}{\theta'^2} \log \frac{1}{\delta'})$}
        \item \textbf{Return:} $\{\widehat{\Tilde{f}}(\chi_{S_j})\}$.
    \end{enumerate}
\caption{Coefficient estimation Algorithm}
\label{algo: estimation}
\end{algorithm}

\begin{theorem}
    \label{th:Coefficient-estimation}
    Algorithm~\ref{algo: estimation} outputs $\{\widehat{\Tilde{f}}(\chi_{S_j})\}$, such that 
    \begin{align*}
        |\{\widehat{\Tilde{f}}(\chi_{S_j})\} - {\Tilde{f}}(\chi_{S_j}) | \leq \theta' 
    \end{align*}
    with probability $\delta'$ with query complexity $\poly(\theta', \delta', \theta, \delta)$. 
\end{theorem}

\begin{proof}
Consider Algorithm~\ref{algo: estimation}. Observe that the only queries made by this algorithm are in Step~\ref{line: query algo: call Implicit Sieve}. Thus by Theorem~\ref{th: Implicit Goldreich Levin}, the query complexity of the algorithm is upper bounded by $O\biggl( \frac{1}{\theta^{10}} \log\biggl(\frac{1}{\theta^{24}}\log \frac{1}{\delta\theta^2}\biggr) \log \frac{1}{\delta\theta^2} \biggr)$. 

There are three bad events for this algorithm:
\begin{enumerate}
    \item The output of the \textsf{List Correction for Abelian groups Algorithm}~\ref{algo:Implicit_Sieve_generalized} is in Step~\ref{line: query algo: call Implicit Sieve}.
    
    \item The algorithm fails to find $B$ in Step~\ref{line: query algo: find B}, in which case it returns \textsf{Fail}.

    \item The estimates in Step~\ref{line: query algo: estimate ^f} are incorrect.
\end{enumerate}

By Theorem~\ref{th: Implicit Goldreich Levin}, the error probability of the first bad event is upper bounded by $\delta$.
The probability of the third bad event can also be bounded by $\delta$ be a Chernoff bound for the choice of $\Tilde{m}$.

Next, we upper bound the probability of the second bad event. Consider the output $Q$ of the \textsf{List Correction for Abelian Groups Algorithm}~\ref{algo:Implicit_Sieve_generalized} in Step~\ref{line: query algo: call Implicit Sieve}. 
The rows are labelled by $1, \dots, \Tilde{m}$, which corresponds to $x^{(1)}, \dots, x^{(\Tilde{m})}$, each distributed uniformly and independently in $\mathcal{G}$. 
Also the columns are labelled by $1, \dots, T$, where $T = O(1/\theta^2)$, which correspond to some $M_1, \dots, M_{T} \in \mathcal{G}$ (the $M_i$'s are not known).
By Theorem~\ref{th: Implicit Goldreich Levin}, the $(i,j)$th entry of $Q$ is $\chi_{s_j}(x_i) \in \{1, \omega_\lcmG, \ldots, \omega_\lcmG^{\lcmG-1}\}$.

Now, let us consider all the group automorphisms on $\G$ which are formed by only choosing the images of the elements in $\Tilde{B}$. Then, by the property of group automorphism, the images of any other element get defined automatically. That is, for an element $r \notin \Tilde{B}$, let $\lambda_1, \ldots, \lambda_{|\Tilde{B}|}, \lambda \in \mathbb{Z}_\lcmG$, with $\lambda$ invertible in $\mathbb{Z}_\lcmG$ be such that $\lambda r = \sum_{i=1}^\lcmG \lambda_i r_i$. If $\psi$ is a group automorphism on $\G$, then $\psi(r) = \frac{1}{\lambda} \sum_{i=1}^\lcmG \lambda_i \psi(r_i)$.

For a subset $S = \{s_1, \dots, s_{|S|}\} \subseteq [T]$, let $\Pi_{S,\lambda_1, \ldots, \lambda_{|S|}}$ denote an element satisfying $(\Pi_{S,\lambda_1, \ldots, \lambda_{|S|}})_i = \prod_{j = 1}^{|S|} Q[i,s_j]^{\lambda_j}$ for all $i \in [\Tilde{m}]$. Observe that, if a dependent $\Tilde{B}$ is chosen in Step~\ref{line: query algo: find B}, then all the entries of $\Pi_{\Tilde{B}, \lambda_1, \ldots, \lambda_{|\Tilde{B}|}}$ are equal to $1$. 

Now we want to argue that for any pseudo-independent $R = \{r_1, \dots,r_{|R|}\} \subseteq \{M_1, \dots, M_{T}\}$ and $\lambda_1,\ldots, \lambda_{|R|} \in \mathbb{N}\cup\{0\}$, the probability that $\Pi_{R,\lambda_1,\ldots, \lambda_{|R|}}$ has all entries $=1$ is at most $\lcmG^{-\Tilde{m}} = \lcmG^{-\Theta(\frac{1}{\theta^{8}}\log\frac{1}{\delta\theta^2})}$.

To see this, fix $R = \{r_1, \dots,r_{|R|}\}$ and observe that $r_j$'s are pseudo-independent. Since $x^{(1)}, \dots, x^{(\Tilde{m})}$ are uniformly and independently distributed in $\mathcal{G}$, $(\Pi_{R, \lambda_1,\ldots, \lambda_{|R|}})_i = 1$ with probability $1/\lcmG$ for all $i \in [\Tilde{m}]$. This follows from the fact that each entry of $\Pi_{R, \lambda_1,\ldots, \lambda_{|R|}}$ can take values from $\{1, \omega_\lcmG, \ldots, \omega_\lcmG^{\lcmG-1}\}$, hence the probability that it takes the value $1$ is $=\frac{1}{\lcmG}$. Thus the probability that $\Pi_{R, \lambda_1,\ldots, \lambda_{|R|}} = 1^{\Tilde{m}}$ is upper bounded by $\lcmG^{-\Tilde{m}}$, for each $\lambda_1, \ldots, \lambda_{|R|}$. Therefore, the probability that a dependent set is chosen in Step~\ref{line: query algo: find B} is at most $\lcmG^{-\Tilde{m}} \times \frac{\binom{\lcmG}{1} \varphi(\lcmG) \lcmG^{|R|-1}}{\varphi(\lcmG)!} = \lcmG^{-\Tilde{m}} \times \frac{\lcmG^{|R|}}{(\phi(\lcmG)-1)!}$, since it is upper bounded by the probability that $\Pi_{R, \lambda_1,\ldots, \lambda_{|R|}} = 1^{\Tilde{m}}$. Since, by Remark~\ref{remark: k small in implicit sieve algo}, the number of possible subsets of $[T]$ is upper bounded by $\lcmG^{\Tilde{m}} = \lcmG^{\Theta(\frac{1}{\theta^{8}}\log\frac{1}{\delta\theta^2})}$, for a large enough constant in choice of $\Tilde{m}$, a pseudo-independent set is chosen with probability $\geq 1-\frac{1}{\lcmG^{\Tilde{m}-|R|} (\varphi(\lcmG)-1)!}$.

In particular, let $\Tilde{B} \subseteq \{M_1, \dots, M_k\}$ such that $\{M_1, \dots, M_k\} \subseteq S_R$, where $S_R$ is the subgroup generated by $R$, and elements of $\Tilde{B}$ are pseudo-independent. With probability at least $99/100$, $B$ will be considered in Step~\ref{line: query algo: find B} of the algorithm (if not such set has already been found). By choosing the elements of $\Tilde{B}$ to be $e_1, \dots, e_{|\Tilde{B}|}$, a brute-force algorithm (which makes no queries to the unknown function) can be used to check whether $\{M_1, \dots, M_{T}\} \subseteq S_R$. Furthermore, the dependencies obtained can be used in the next step (Step~\ref{line: query algo: relabel columns}) of the algorithm.

Thus the probability of bad events happening for Algorithm~\ref{algo: estimation} is at most $1/3$. 
\end{proof}

\end{document}